\documentclass[10pt,a4paper]{article}

\usepackage{hyperref}

\usepackage{a4wide}                
\usepackage{graphicx}
\usepackage{xcolor}
\usepackage{subcaption}
\usepackage{float}
\usepackage[format=hang,labelfont=sc,font=small]{caption}
\usepackage{multirow}
\usepackage{amsmath,amssymb,amsthm,amsfonts,mathrsfs}
\numberwithin{equation}{section}
\allowdisplaybreaks
\usepackage{mathtools}
\graphicspath{{figures/}}
\DeclareGraphicsExtensions{.pdf}

\newcommand{\revision}[1]{\textcolor{black}{#1}}
\newcommand{\minor}[1]{\textcolor{black}{#1}}

\newcommand{\dd}{\mathrm{d}}

\newtheorem{remark}{Remark}[section]
\newtheorem{property}{Property}
\newtheorem{definition}[remark]{Definition}
\newtheorem{theorem}[remark]{Theorem}
\newtheorem{proposition}[remark]{Proposition}

\numberwithin{figure}{section}
\numberwithin{table}{section}

\title{The BGK approximation of kinetic models for traffic}
\date{\today}
\author{M.~Herty 
	\thanks{
		Institut f\"{u}r Geometrie und Praktische Mathematik - 
		RWTH Aachen University --
		Templergraben 55, 52062 Aachen, Germany --
		{\sl herty@igpm.rwth-aachen.de}
	}
	\and 
	G.~Puppo 
	\thanks{
		Dipartimento di Scienza ed Alta Tecnologia -
		Universit\`a degli Studi dell'Insubria --
		Via Valleggio 11, 22100 Como, Italy --
		{\sl gabriella.puppo@uninsubria.it}
	}
	\and
	S.~Roncoroni 
	\thanks{
		Department of Mathematics and Statistics -
		University of Reading --
		Exhibition Road, London SW7 2AZ, UK --
		{\sl s.roncoroni@pgr.reading.ac.uk}
	}
	\and
	G.~Visconti
	\thanks{
		Institut f\"{u}r Geometrie und Praktische Mathematik - 
		RWTH Aachen University --
		Templergraben 55, 52062 Aachen, Germany --
		{\sl visconti@igpm.rwth-aachen.de}
	}
}

\begin{document}
\maketitle

\begin{abstract}
	We study spatially non-homogeneous kinetic models for vehicular traffic flow. Classical formulations, as for instance the BGK equation, lead to unconditionally unstable solutions in the congested regime of traffic. We address this issue by deriving a modified formulation of the BGK-type equation. The new kinetic model allows to reproduce conditionally stable non-equilibrium phenomena in traffic flow. In particular, stop and go waves appear as bounded backward propagating signals occurring in bounded regimes of the density where the model is unstable. The BGK-type model introduced here also offers the mesoscopic description between the microscopic follow-the-leader model and the macroscopic Aw-Rascle and Zhang model. 
\end{abstract}

\paragraph{MSC} 90B20, 35Q20, 35Q70

\paragraph{Keywords} Vehicular traffic, BGK models, Chapman-Enskog expansion, multi-scale models, kinetic equations

\section{Introduction} \label{sec:Introduction}

There are mainly three modeling scales in the mathematical description of vehicular traffic flow. The microscopic one is based on the prediction of the trajectory of each single vehicle via systems of ordinary differential equations. The macroscopic one is based on the assumption that traffic flow behaves like a fluid where individual vehicles cannot be identified. Therefore, the flow is represented by a density function and evolved in space and time via partial differential equations. The last one is the mesoscopic scale which offers an intermediate description of traffic between the microscopic and the macroscopic scale. In fact, mesoscopic or kinetic models are characterized by a statistical description of the microscopic states of vehicles but, at the same time,  still provide the macroscopic aggregate representation of traffic flow e.g.  linking collective dynamics to pairwise interactions among vehicles at a smaller microscopic scale.

Prototype examples of microscopic models are given by the optimal velocity model by Bando~\cite{Bando1995} and the follow-the-leader model~\cite{FTL1961}. The first and classical example of a macroscopic model for traffic was introduced by Lighthill and Whitham~\cite{lighthill1955PRSL} and independently by Richards~\cite{richards1956OR}. This model is based on the single conservation law for the density assuming that the dynamics is governed by a velocity function which is a function of the density only. An extension of this is provided by the second-order macroscopic model introduced by Aw and Rascle~\cite{aw2000SIAP} and by Zhang~\cite{Zhang2002}. It is based on a system of two partial differential equations describing also the evolution of the speed of the flow. There is a proved link between common microscopic and macroscopic models. More precisely, the follow-the-leader model approaches to the Lighthill-Wihtham-Richards model and to the Aw-Rascle-Zhang model in the limit of infinitely many vehicles~\cite{aw2002SIAP,DiFrancescoFagioliRosini2017,DiFrancescoEtAl2017,HoldenRisebro2018b,HoldenRisebro2018a}.

This work is mainly concerned with the analysis of traffic flow models at kinetic level. However, we point-out that throughout the paper the study is performed by linking models at different scales. Classical kinetic models for vehicular traffic~\cite{KlarWegener96,klar1997Enskog,paveri1975TR,Prigogine61,PrigogineHerman} are based on a Boltzmann-type equation that, in classical kinetic theory, describes the statistical behavior of a system of particles in a gas. Other type of kinetic formulations are available for modeling traffic flow and they also allow to simplify the complex structure of the integro-differential Boltzmann-type equation. Among them, we mention Vlasov-Fokker-Planck type models in which the integral structure of the collision kernel is replaced by differential operators, obtained also by means of suitable time scaling, e.g. see~\cite{HertyPareschi10,IllnerKlarMaterne}, and discrete-velocity models, e.g. see~\cite{DelitalaTosin2007,FermoTosin13,HertyPareschiSeaid2006}.

We study spatially non-homogeneous Boltzmann-type kinetic models for traffic and in particular the BGK approximation, originally introduced by Bhatnagar, Gross and Krook~\cite{BGK1954} for mesoscopic models of gas particles. The BGK equation replaces the Boltzmann-type kernel with a relaxation towards the equilibrium distribution of the full kinetic equation. For this reason, the BGK model provides a good description of the Boltzmann-type equation in the limit of many interactions, since in this regime the kinetic distribution is close to equilibrium. However, when the interaction frequency is extremely high to immediately drive the kinetic distribution towards the equilibrium one, the kinetic equation converges to the solution of the Lighthill-Whitham-Richards model. Since the flow is at equilibrium, non-equilibrium phenomena are not described in this regime. If the interaction frequency is still large but there is some residual between the kinetic and the equilibrium distribution, then the kinetic model is able to reproduce non-equilibrium waves. We are interested in the study of the stability of the BGK model in the latter case. Using a Chapman-Enskog expansion we show that the BGK equation for traffic flow problems yields  an advection-diffusion equation having a negative diffusion coefficient in dense traffic. This analysis is closely related to the investigation of the sub-characteristic condition for relaxation systems~\cite{Chen92hyperbolicconservation,Jin95therelaxation}. \minor{This implies possible} unbounded growth of non-equilibrium waves and the corresponding behavior is investigated at the particle description and \revision{also at the macroscopic level of the BGK model by looking at the system of the second-order moment equations}.

Due to the possible instability of the kinetic model, here we aim to derive a novel BGK equation which is stable or weakly-stable in dense traffic. As in Section~\ref{sec:unstableKinetic}, we define the BGK model being stable if the sub-characteristic condition is always verified and weakly-stable if the sub-characteristic condition is violated in a proper subdomain of the admissible values of the density. As discussed in~\cite{SeiboldFlynnKasimovRosales2013}, a weakly-stable model for traffic is not necessarily unrealistic:  the instability occurring in the regime  can be regarded as model for the appearance of moving unstable but bounded waves, as e.g. stop-and-go waves. The \minor{boundedness} of these instabilities is numerically observed in Section~\ref{sec:macroModels}. In~\cite{SeiboldFlynnKasimovRosales2013} the analysis of the sub-characteristic condition has been performed for the Aw-Rascle and Zhang model, showing that it could be either satisfied or weakly-satisfied, depending on the choices of the equilibrium speed or flux and of the pressure function. Based on this study, our goal is to develop a BGK-type equation which has as \revision{second-order macroscopic structure a Aw-Rascle and Zhang type model} so that it may be at least weakly-stable. The new kinetic model for traffic is derived by starting from the follow-the-leader model that can be seen as a discretization of the Aw-Rascle and Zhang model. Therefore, we do not only introduce a well-posed or weakly well-posed kinetic equation but we also contribute in deriving the mesoscopic description of common models for traffic flow. The diagram in Figure~\ref{fig:diagram} schematically synthesizes the contribution of this work.

\begin{figure}[t!]
	\centering
	\includegraphics[width=\textwidth]{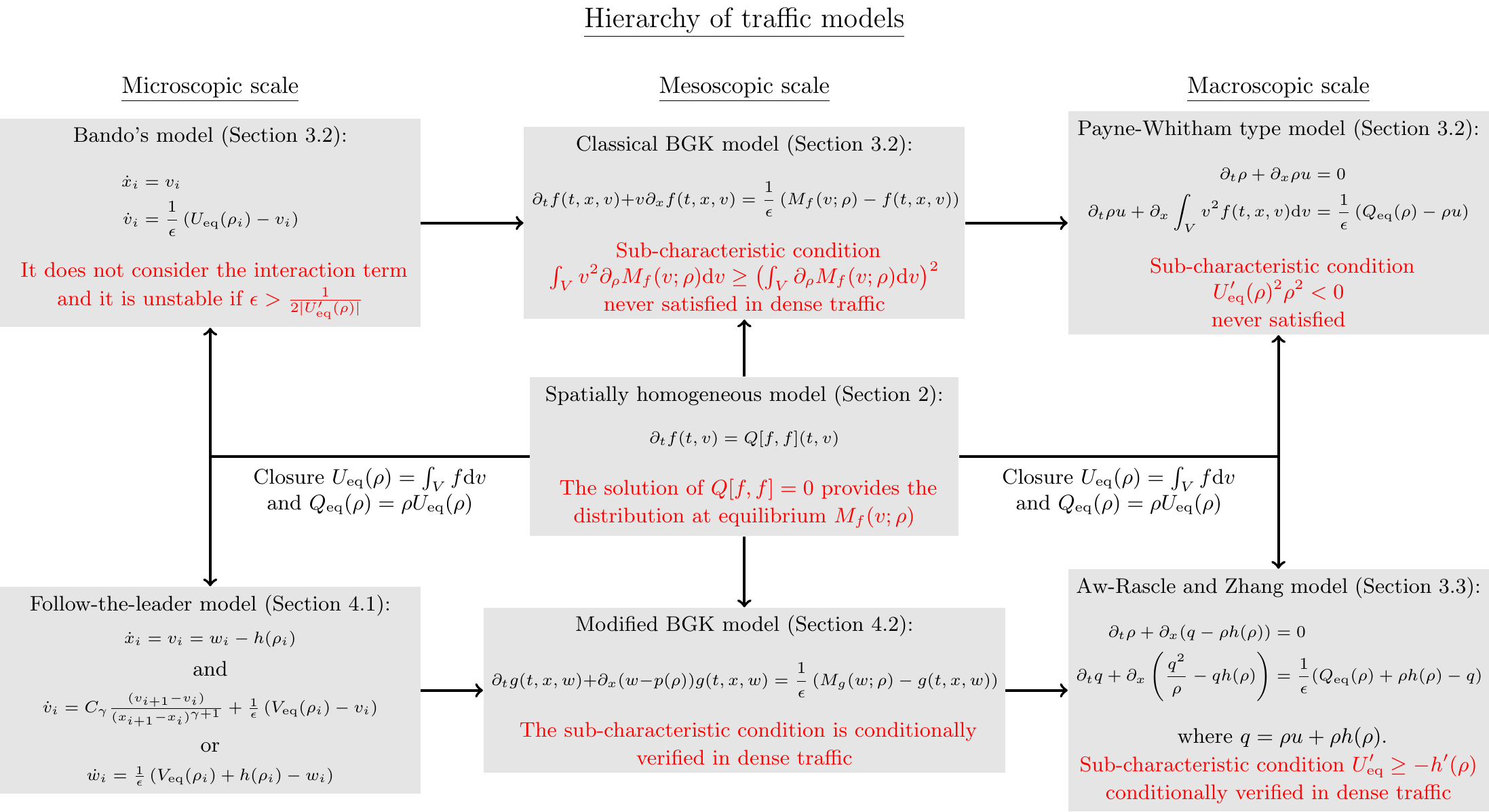}
	\caption{Schematic summary of the work.\label{fig:diagram}}	
\end{figure}

The manuscript is organized as follows. In Section~\ref{sec:homogeneous} we briefly review the homogeneous kinetic model for traffic flow introduced in~\cite{PSTV2}. In Section~\ref{sec:inhomogeneous} we introduce the non-homogeneous version of the kinetic model and we study non-equilibrium phenomena by using the classical BGK approximation of the full kinetic equation. We show that model is unstable, namely it is able to produce unbounded growth of small perturbations, via Chapman-Enskog expansions and \revision{Grad's moment method}. In Section~\ref{sec:correctBGK} we propose a way to derive an alternative formulation of a BGK-type model for traffic flow. This approach has also the advantage of prescribing the mesoscopic step between the microscopic follow-the-leader model and the macroscopic Aw-Rascle and Zhang model. Finally, we discuss the results and possible perspectives in Section~\ref{sec:conclusions}.

\section{Review of a homogeneous kinetic model} \label{sec:homogeneous}

Recent homogeneous kinetic model for traffic flow have been discussed e.g. in~\cite{PSTV2,TesiRonc}. The model proposed therein is used as starting point for the discussion in the following. However, the particular choice below is not required for the shown results later. 

General  kinetic traffic models of the form read 
\begin{equation} \label{eq:generalKinetic}
	\partial_t f(t,x,v) + v \partial_x f(t,x,v) = \frac{1}{\epsilon} Q[f,f](t,x,v),
\end{equation}
where
\begin{equation} \label{eq:kineticDistr}
	f(t,x,v) : \mathbb{R}^+ \times \mathbb{R} \times V \rightarrow \mathbb{R}^+
\end{equation}
is the mass distribution function of the flow, i.e. $f(t,x,v)\dd x\dd v$ gives the number of vehicles in $[x,x+\dd x]$ with velocity in $[v,v+\dd v]$ at time $t>0$. Moments of the kinetic distribution function $f$ allow to define the usual macroscopic quantities for traffic flow
\begin{equation} \label{eq:macroQuantities}
	\rho(t,x)= \int_V f(t,x,v) \dd v, \quad (\rho u)(t,x) = \int_V v f(t,x,v) \dd v, \quad u(t,x) = \frac{1}{\rho} \int_V v f(t,x,v) \dd v
\end{equation}
yielding  density,  flux and  mean speed of vehicles, respectively, at time $t$ and position $x$. The space $V:=[0,V_M]$ is the space of microscopic speeds. We suppose that $V$ is limited by a maximum speed $V_M>0$, which may depend on several factors, and that the density is also limited by a maximum density $\rho_M$. Throughout the work  we will take dimensionless quantities, thus $V_M=1$ and $\rho_M=1$.

In~\cite{PSTV2}  the spatially homogeneous kinetic model corresponding to~\eqref{eq:generalKinetic} has been discussed: 
\begin{equation} \label{eq:homogeneousKinetic}
	\partial_t f(t,v) = \frac{1}{\epsilon} Q[f,f](t,v).
\end{equation}
Both in~\eqref{eq:generalKinetic} and in~\eqref{eq:homogeneousKinetic}, $\epsilon$ is a positive quantity. It is reasonable to assume that it may depend on the density $\rho$ and possibly its first derivative, see Remark~\ref{rem:backward}.  The equilibrium solution of~\eqref{eq:homogeneousKinetic} is  called equilibrium distribution, or Maxwellian in analogy with kinetic models for rarefied gas dynamics, and we indicate it by $M_f$. 

The equilibrium distribution $M_f$ depends on the modeling of the collision kernel $Q[f,f]$ that accounts for the kinetic interactions between vehicles. The formulation of $Q[f,f]$ draws inspiration from the classical Boltzmann collision kernel for gas particles. Thus $Q[f,f]$ is split in the difference of a gain term and a loss term. More specifically,
\begin{equation} \label{eq:collisionKernel}
	Q[f,f](t,x,v) = \int_V \dd v_* \int_V \dd v^* f(t,v_*,x) \mathcal{P}(v_*\to v | v^*;\rho) f(t,v^*,x) - f(t,v,x) \int_V \dd v^* f(t,v^*,x),
\end{equation}
where we consider binary interactions, in which a test vehicle with velocity $v_*$ interacts with a field vehicle with speed $v^*$, emerging with speed $v$ as a result of the interaction and with probability given by $\mathcal{P}$. The kinetic model for traffic flow studied in~\cite{PSTV2} can be generalized to the following probability distribution accounting for stochastic speed transitions, see~\cite{TesiRonc}:
\begin{equation} \label{eq:rule}
\mathcal{P}(v_*\to v | v^*;\rho)=
\begin{cases} 
	P\delta (v-\min{ \lbrace v_* + \Delta_a, V_M \rbrace }) + (1-P)\delta (v-\max{\lbrace v_*- \Delta_b, 0}\rbrace) & v_* \le v^* \\ 
	P\delta (v-\min{\lbrace v_*+ \Delta_a, V_M}\rbrace)+(1-P)\delta (v-\max{\lbrace v^*- \Delta_b, 0}\rbrace) & v_* > v^* 
\end{cases}
\end{equation}
where $P=P(\rho)\in[0,1]$ is a decreasing function of the density giving the probability of accelerating, $\Delta_a$ and $\Delta_b$ are the acceleration and the braking parameters, respectively. We can think of $\Delta_a$ as the instantaneous physical acceleration of a vehicle. The parameter $\Delta_b$ instead corresponds to an uncertainty in the estimate of the field vehicle speed, either one's own, or the one of the vehicle ahead, as in the first and the second line respectively of~\eqref{eq:rule}. Note that the model is continuous across the line $v_*=v^*$ and that mass conservation holds true since
$$
	\mathcal{P}(v_*\to v | v^*;\rho) \geq 0, \quad \int_V \mathcal{P}(v_*\to v | v^*;\rho) \mathrm{d}v = 1.
$$
In the following we will consider $\Delta_b=0$ and $\Delta_a = \Delta v$ constant, so that one recovers exactly the model in~\cite{PSTV2}. 

In~\cite{PSTV2} a result on the existence and uniqueness of a particular class of stationary solutions has been established:

\begin{theorem}[Theorem 3.1 and Theorem 3.4 in~\cite{PSTV2}] \label{th:continuousEq}
	Let $P=P(\rho)$ be a given function of the density $\rho\in[0,\rho_M]$ such that $P\in[0,1]$. Let $\{v_j\}_{j=1}^N$ be a set of velocities in $[0,V_M]$. The distribution function
	\[
		M_f(v;\rho) = \sum_{j=1}^N f^{\infty}_j(\rho) \delta_{v_j}(v), \quad f^{\infty}_j > 0 \quad \forall\; j=1,\dots,N,
	\]
	with $\sum_{j=1}^N f_j^\infty(\rho)=\rho$, is the unique stable weak stationary solution of the model~\eqref{eq:homogeneousKinetic}-\eqref{eq:rule} with $\Delta_b=0$ provided  $v_j=v_1+j\Delta v$, $j=1,\dots,N$, and
	\begin{equation*}
		f_j^\infty = \begin{cases}
		0 & P \geq \tfrac12 \\ \frac{-2(1-P)\sum_{k=1}^{j-1} f_k^\infty+(1-2P)\rho+\sqrt{\left[(1-2P)\rho-2(1-P)\sum_{k=1}^{j-1}f_k^\infty\right]^2+4P(1-P)\rho f_{j-1}^\infty}}{2(1-P)} &\mbox{otherwise}
		\end{cases}
	\end{equation*}
	for $j=1,\dots,N-1$ and $f_N^\infty = \rho - \sum_{j=1}^{N-1} f_j^\infty$.
\end{theorem}

Modeling with  kinetic equations such as~\eqref{eq:generalKinetic} is driven by the interaction kernel $Q[f,f]$ and for this reason homogeneous kinetic models are widely studied for traffic flow problems. 
Using the steady-state of the kinetic model it is possible to define the flux and the mean speed of vehicles at equilibrium as
\begin{equation} \label{eq:macroQuantitiesEq}
	Q_\text{eq}(\rho) = \left(\rho U_\text{eq}(\rho)\right) = \int_V v M_f(v;\rho) \dd v, \quad U_\text{eq}(\rho) = \frac{1}{\rho} \int_V v M_f(v;\rho) \dd v.
\end{equation}
Notice that Theorem~\ref{th:continuousEq} states that $M_f$ depends only on a finite number of speeds, that it is parameterized by the local density $\rho$ and that, for each fixed value of $\rho$, it is unique. Therefore, the maps $\rho \mapsto Q_\text{eq}(\rho)$ and $\rho \mapsto U_\text{eq}(\rho)$ are actually functions prescribing a uniquely defined correspondence between the density and the quantities~\eqref{eq:macroQuantitiesEq}. These relations provide the so-called simulated fundamental diagrams of traffic which are used in order to validate the model.

\begin{figure}[t!]
	\centering
	\includegraphics[width=0.49\textwidth]{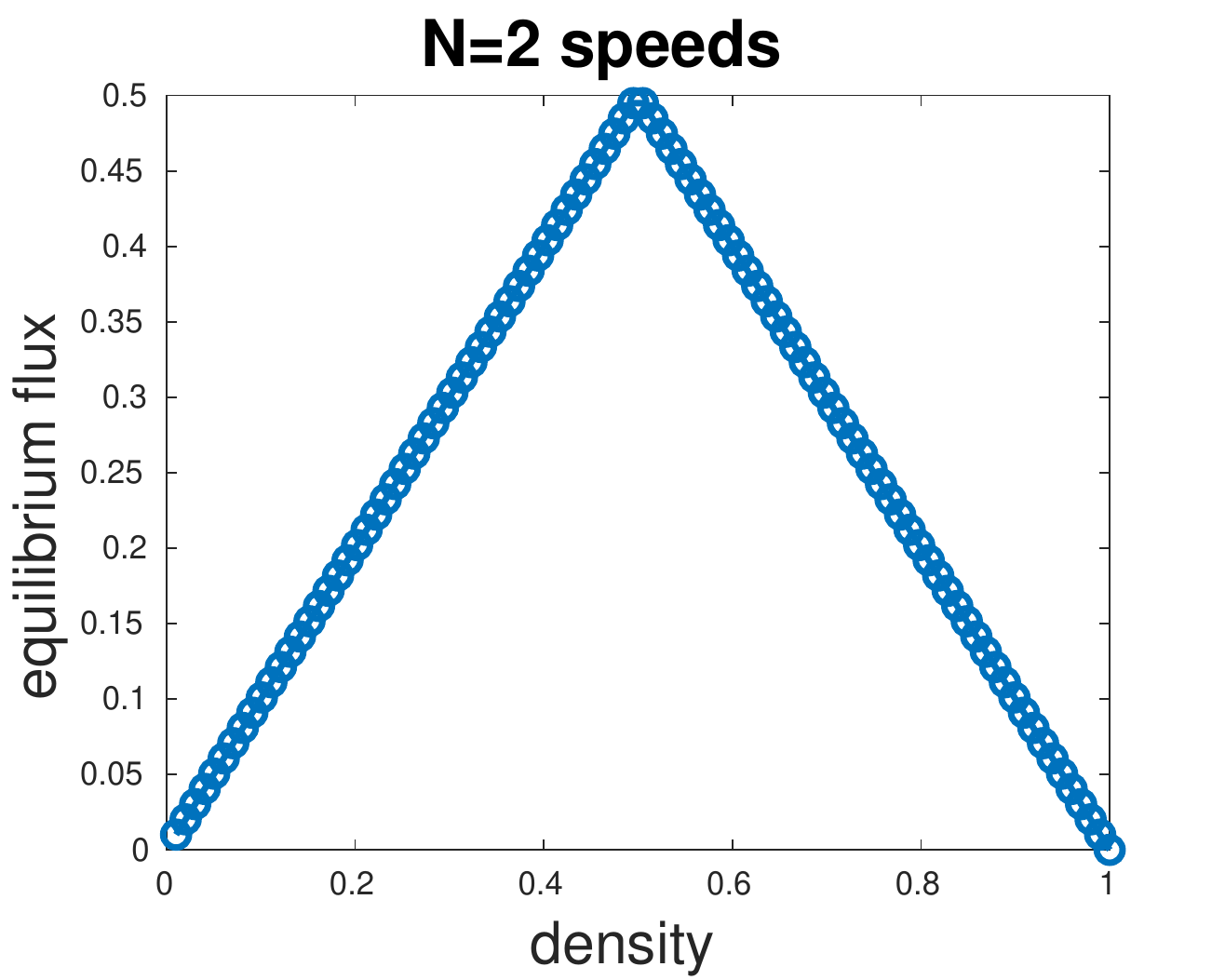}
	\includegraphics[width=0.49\textwidth]{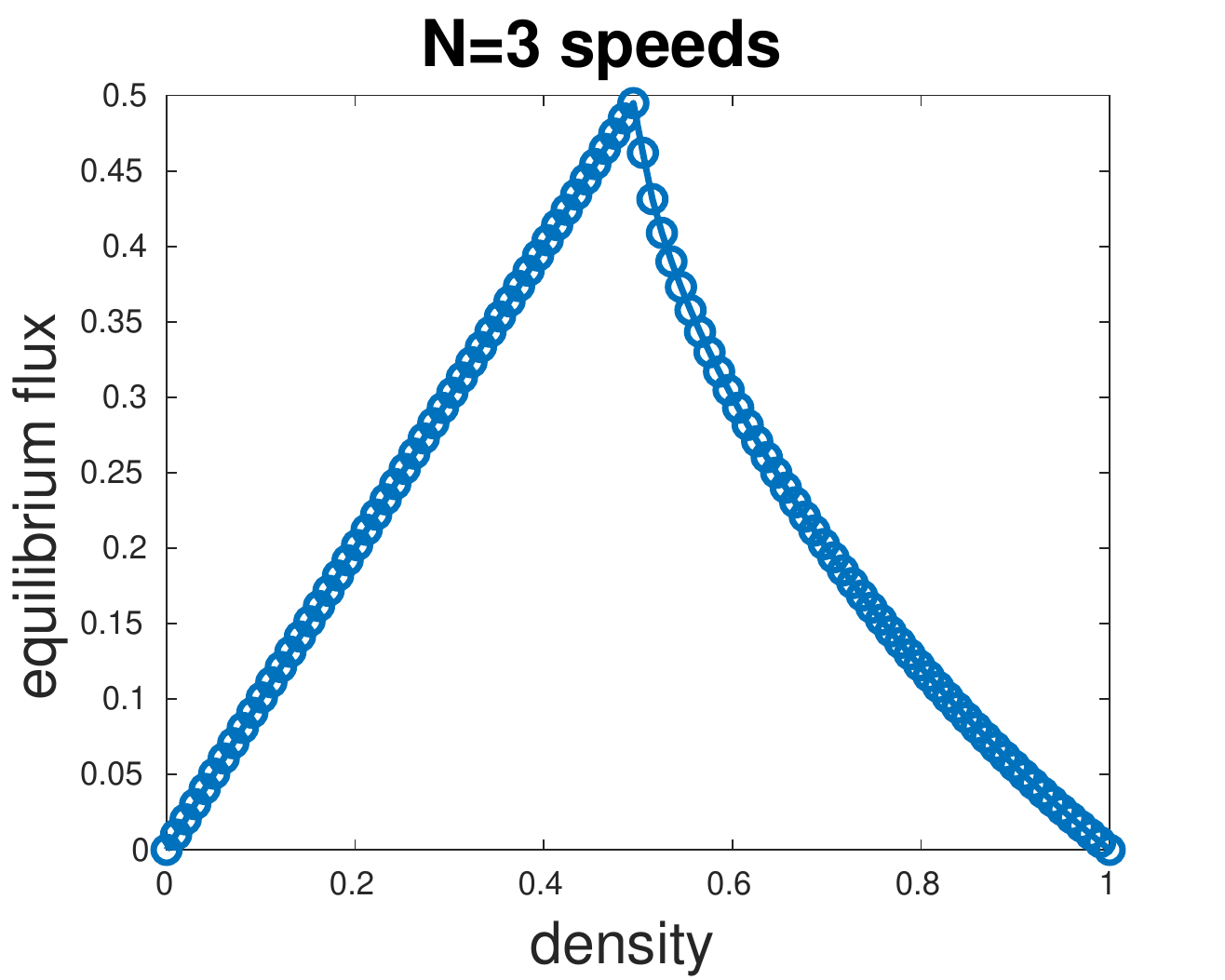}
	\\
	\includegraphics[width=0.49\textwidth]{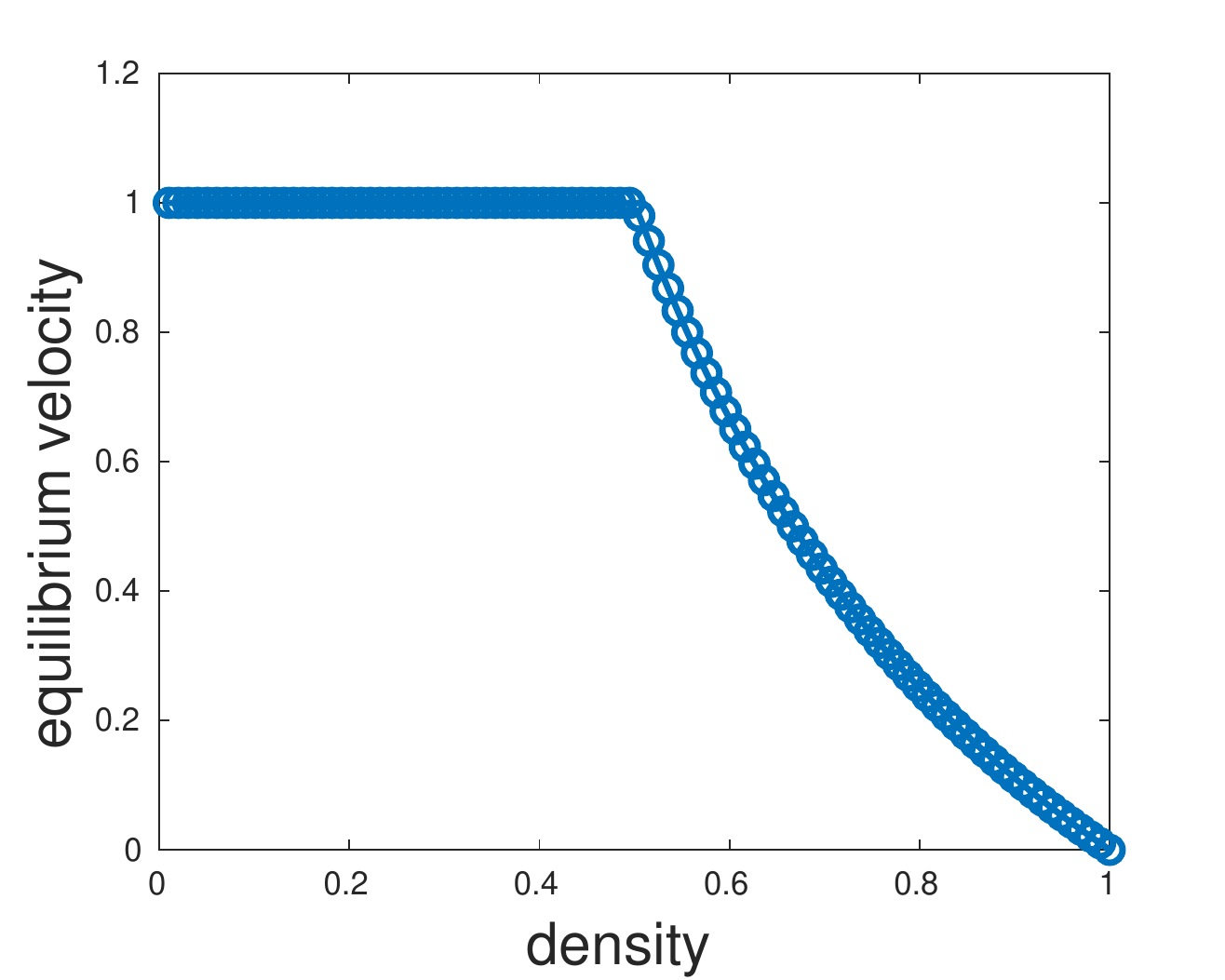}
	\includegraphics[width=0.49\textwidth]{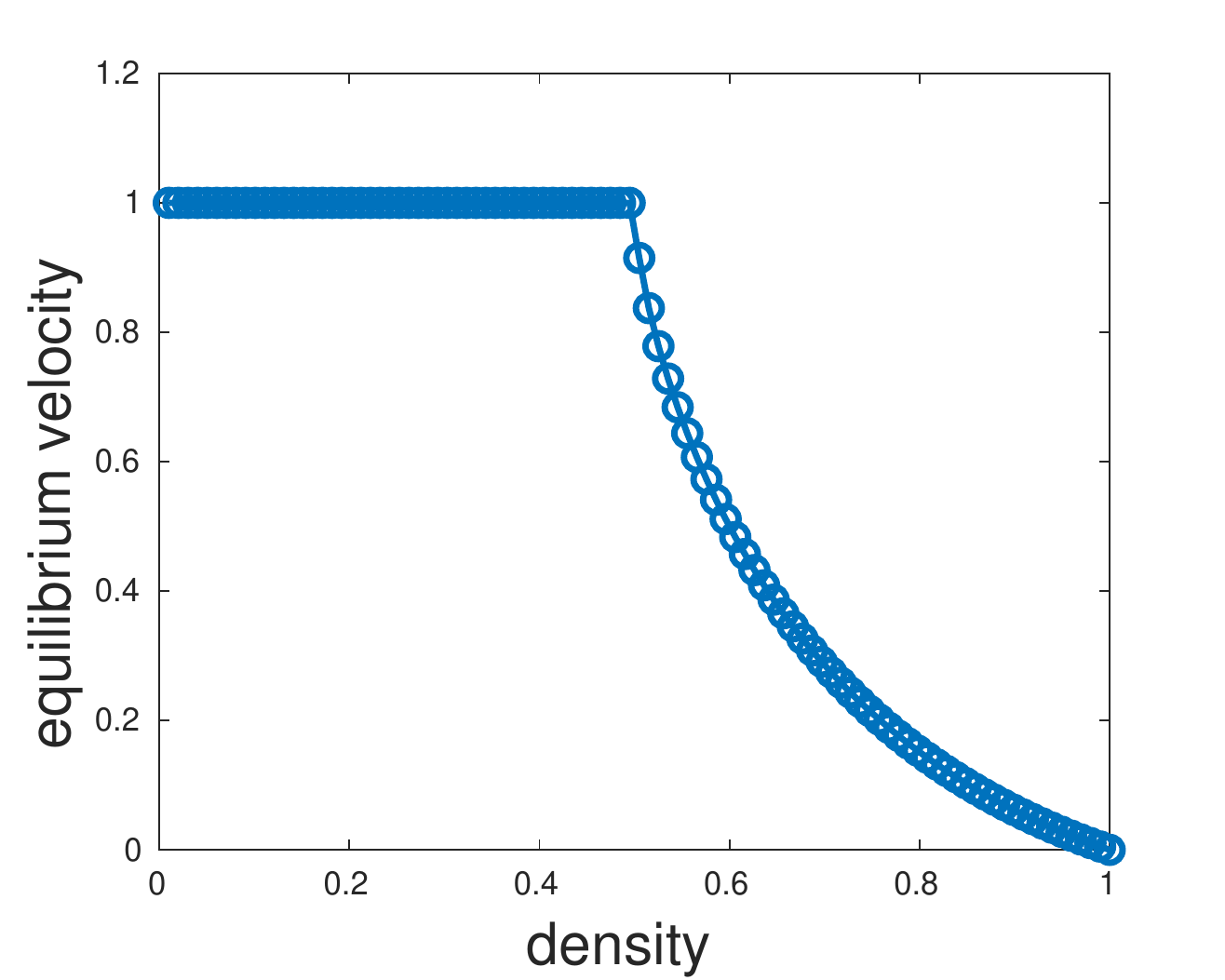}
	\caption{Equilibrium flux diagrams (top row) and equilibrium speed diagrams (bottom row) of the homogeneous kinetic model~\eqref{eq:homogeneousKinetic}-\eqref{eq:rule} with $N=2$ and $N=3$ speeds for the case $\Delta a=\frac12$ and $\Delta b=0$.\label{fig:fundDiag}}
\end{figure}
\begin{figure}[t!]
	\centering
	\includegraphics[width=0.49\textwidth]{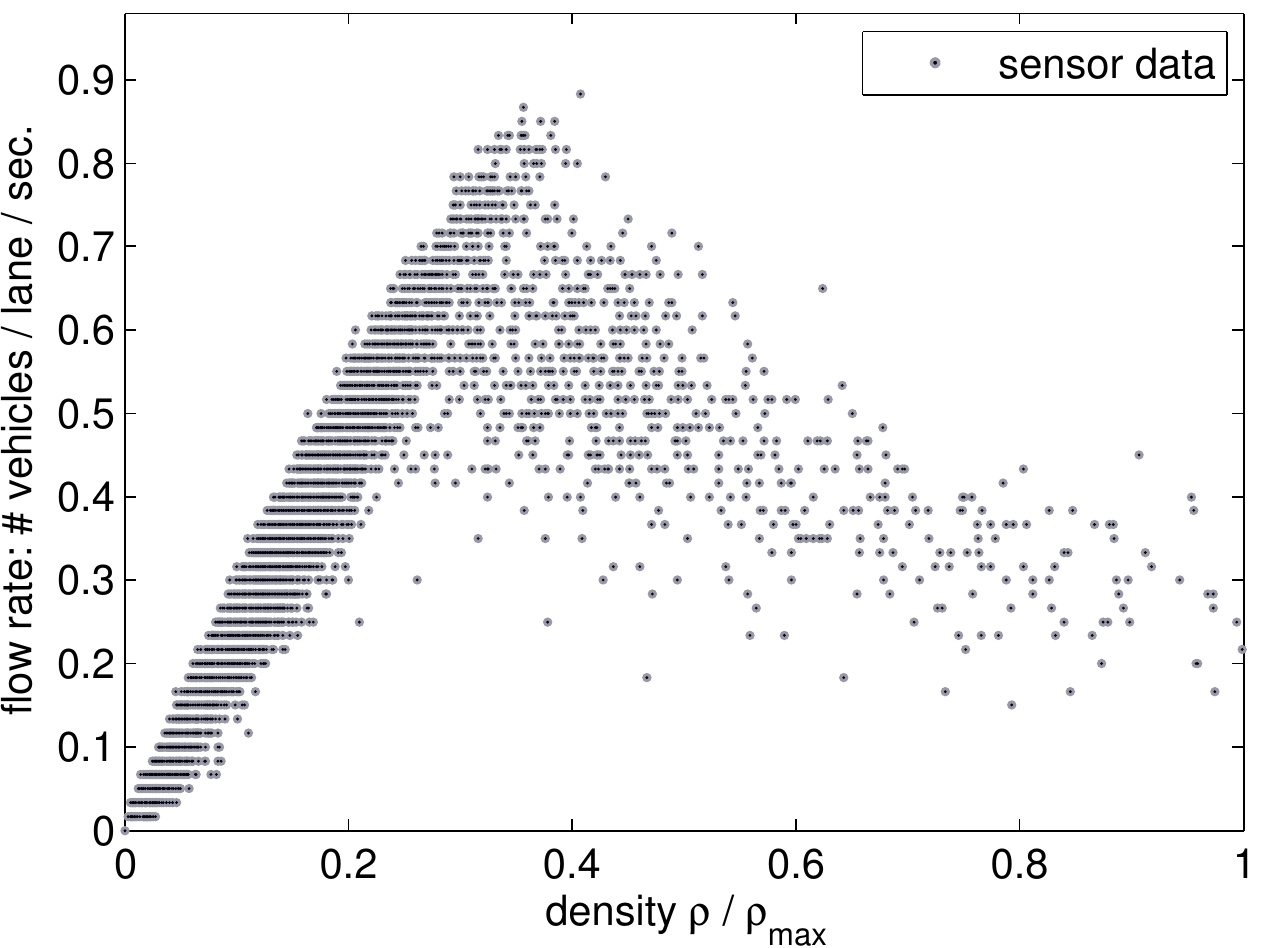}
	\includegraphics[width=0.49\textwidth]{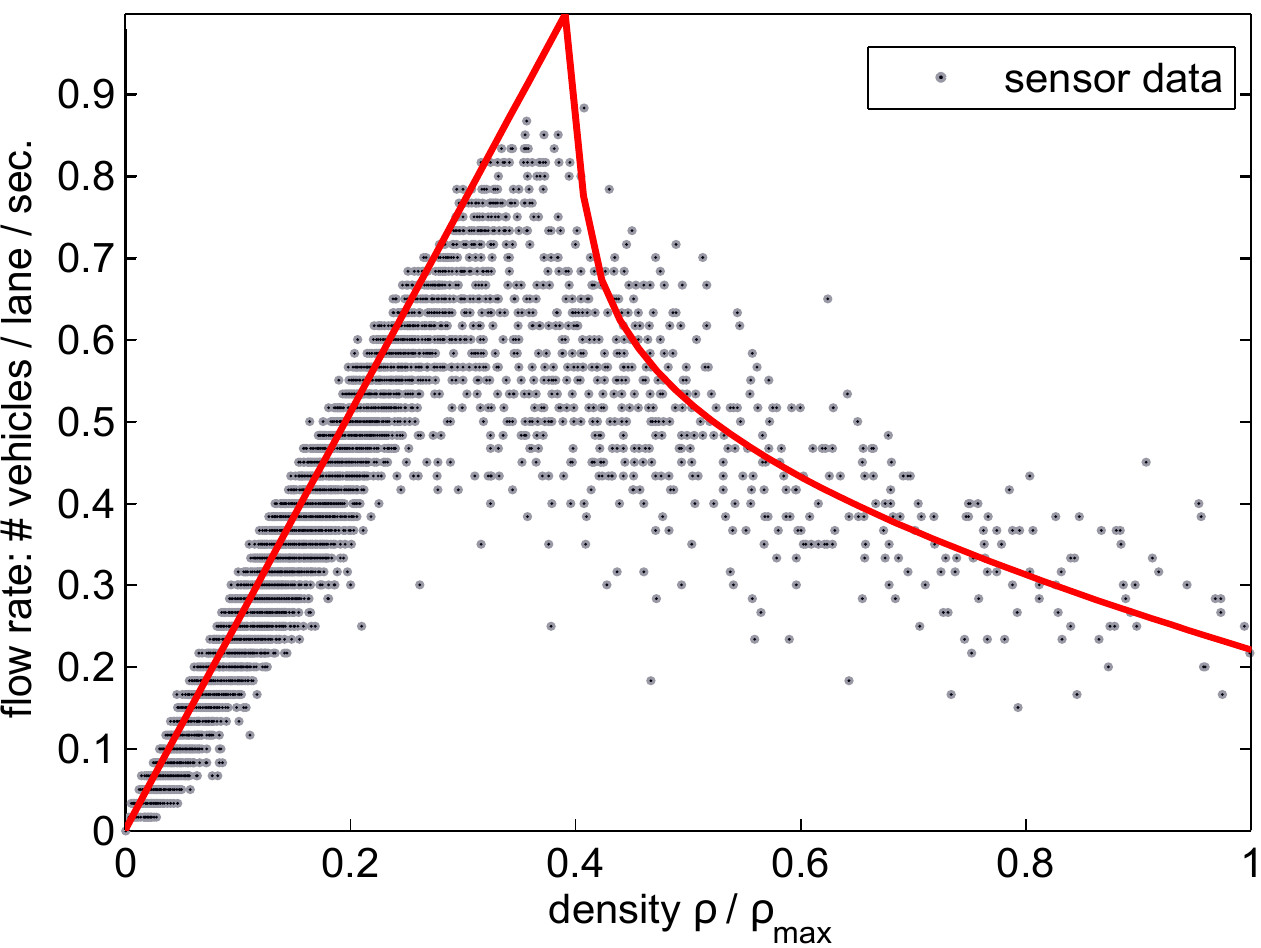}
	\caption{Left: experimental flux diagram from measurements by the Minnesota Department of Transportation, reproduced by kind permission of Seibold et al.~\cite{SeiboldFlynnKasimovRosales2013}. Right: comparison between experimental data and the flux diagram resulting from the model~\eqref{eq:homogeneousKinetic}-\eqref{eq:rule}, with $N=3$ speeds and $P(\rho) = 1 - \rho^{\frac14}$ for the case $\Delta a=\frac12$ and $\Delta b=0$.\label{fig:diagComparison}}
\end{figure}

In Figure~\ref{fig:fundDiag} we show the fundamental diagrams of  homogeneous kinetic model~\eqref{eq:homogeneousKinetic}-\eqref{eq:rule} that are computed by using the equilibrium distribution prescribed by Theorem~\ref{th:continuousEq}. The top row shows the correspondence $\rho \mapsto Q_\text{eq}(\rho)$ and the bottom row shows $\rho \mapsto U_\text{eq}(\rho)$. We take $P(\rho) = 1-\rho$ and we consider the case of $N=2$ speeds (left panels) and $N=3$ speeds (right panels) with $\Delta v = \frac{1}{N-1}$. In Figure~\ref{fig:diagComparison} we show that the fundamental diagrams can fit experimental flux diagrams. This is for instance the case of the experimental measurements provided by the Minnesota Department of Transportation in 2013 and tuning the parameters of the model as $N=3$ speeds, $\Delta v=\frac12$ and $P(\rho)=1-\rho^{\frac14}$. Observe that the flux at maximum density is not zero since we compute the diagram up to a value of the density $\widetilde{\rho}$ such that $\widetilde{\rho}<\rho_M$. This is due to the observation
that experimental data contain a residual movement even in the congested phase. The model is able to describe the phase transition and to catch the two regimes of traffic compared with experimental data. For low values of $\rho$, more precisely for $\rho<\rho_{\mbox{\scriptsize crit}}$, where $\rho_{\mbox{\scriptsize crit}}$ is the solution of $P(\rho)=\frac12$, the flux is linear in $\rho$. This is the phase of free flow, when $U_\text{eq}(\rho)=V_M$. For larger values of $\rho$, i.e. $\rho$ larger than the critical density $\rho_{\mbox{\scriptsize crit}}$, the role of the interactions increases, and the flux decreases. This corresponds to the congested phase of traffic flow, see also~\cite{kerner2004BOOK}.  

\begin{remark}
	Thanks to Theorem~\ref{th:continuousEq} we have that the equilibrium distribution $M_f$ has a fixed dependence on $v$, and depends on $x$ and $t$ only through the local density $\rho(x,t)$. However, for $\Delta_b=0$, it was shown in~\cite{PSTV2} that unstable equilibria can be found, for which $M_f$ depends not only on $\rho$, but also on the initial distribution $f(t=0,x,v)$. Those equilibria are unstable, in the sense that they disappear and drive to the solutions given by Theorem~\ref{th:continuousEq} when the initial data are perturbed. If the braking uncertainty $\Delta_b > 0$, it was shown in~\cite{TesiRonc} that only stable equilibria are found, they still depend on a finite number of speeds. In this work we still focus on the case $\Delta_b = 0$ and consider initial distributions $f(t=0,x,v)$ such that only stable equilibria occur.
\end{remark}

\section{Uncontrolled instabilities in standard non-homogeneous kinetic models} \label{sec:inhomogeneous}

\subsection{Motivation: non-equilibrium regimes} \label{sec:motivation}

The properties of the fundamental diagrams derived in Section~\ref{sec:homogeneous} are the result of the modeling of the microscopic interactions. Although the model is able to reproduce typical features of the experimental diagrams, we point-out that the simulated diagrams are only single valued. This is a direct consequence of the uniqueness of the equilibrium distribution provided by Theorem~\ref{th:continuousEq}. This means that the model cannot reproduce and explain the scattering behavior of the experimental measurements, at least not at equilibrium. Mathematical models being able to describe multivalued fundamental diagrams were also studied and they link the observed scattered data to several reasons: for instance the presence of heterogeneous components, such as different drivers~\cite{LoSchiavo2002,MendezVelasco13} or populations of vehicles~\cite{PgSmTaVg3}, or the deviation of microscopic speeds at equilibrium~\cite{FermoTosin14}, or the presence of non-equilibrium phenomena such as stop and go waves~\cite{KimZhang2008,SeiboldFlynnKasimovRosales2013,ZhangKim2005}.

In this work we will focus on the last motivation and  consider non-equilibrium phenomena that might lead to complex traffic phenomena as e.g.  stop and go waves. Such phenomena cannot be described by spatially homogeneous models but there is the need to reconsider the space dependence and thus to take into account the full kinetic equation~\eqref{eq:generalKinetic}.


\subsection{Non-homogeneous kinetic models} \label{sec:unstableKinetic}

In the non-homogeneous case, the relaxation parameter $\epsilon$ plays a  crucial role since it balances the weight between the convection and the source term. In analogy to gas dynamics, where the interaction frequency is given by the Knudsen number, different values of $\epsilon$ lead to different regimes. The situation is synthesized in Table~\ref{tab:epsilon}. If we allow for $\epsilon = 0$, i.e. we suppose that the interactions are so frequent to instantaneously relax $f$ to the local equilibrium distribution $M_f$, we are in the so-called equilibrium flow regime where~\eqref{eq:generalKinetic} is solving the conservation law for the density
\begin{equation} \label{eq:modelLWR}
	\partial_t \rho + \partial_x Q_\text{eq}(\rho) = 0
\end{equation}
with $Q_\text{eq}(\rho)$ defined in~\eqref{eq:macroQuantitiesEq}. Equation~\eqref{eq:modelLWR} is a first-order macroscopic model for traffic flow ~\cite{lighthill1955PRSL,richards1956OR}.  Instead, we expect that if $\epsilon$ is small, but not vanishing, then we are in a regime where the kinetic equation~\eqref{eq:generalKinetic} is solving the continuity equation~\eqref{eq:modelLWR} with a diffusive perturbation of order $O(\epsilon)$. For $\epsilon \asymp 1$ we are in the transitional regime where the counterpart of the kinetic equation is given by extended continuum hydrodynamic equations. In traffic flow, this is the case of the Aw-Rascle~\cite{aw2000SIAP} and Zhang~\cite{Zhang2002} model. For $\epsilon > 1$, but not too large, we are in the kinetic regime or rarefied gas regime and finally for $\epsilon \to \infty$ we obtain the regime of the collision-less kinetic equation.

\begin{table}[t!]
	\begin{center}
		\caption{Classification of flow regimes for the Boltzmann-type kinetic equation~\eqref{eq:generalKinetic}.}
		\label{tab:epsilon}
		\small
		\begin{tabular}{|c|l|l|l|}
			\hline
			$\epsilon$ & Regime & Kinetic model & Continuum flow model \\
			\hline
			$0$ & Equilibrium flow & \multirow{4}{*}{Boltzmann} & Mass conservation law (LWR model) \\
			$\to 0$ & Viscous flow & & Diffusion equation \\
			$\asymp 1$ & Transitional & & Extended hydrodynamic equations \\
			$>1$ & Rarefied & & - \\ \cline{3-3}
			$\to\infty$ & Free molecular flow & Collision-less Boltzmann & - \\
			\hline
		\end{tabular}
	\end{center}
\end{table} 

Here, we will consider a type of approximation for the full non-linear Boltzmann-type equation which formally holds for small values of $\epsilon$, namely the BGK approximation. This technique is based on replacing the Boltzmann-type kernel in~\eqref{eq:generalKinetic} by a relaxation term:
\begin{equation} \label{eq:BGK}
	\partial_t f(t,x,v) + v \partial_x f(t,x,v) = \frac{1}{\epsilon} \left( M_f(v;\rho) - f(t,x,v) \right),
\end{equation}
where $M_f(v;\rho)$ is the distribution of the microscopic speeds $v$ at equilibrium, corresponding to the density $\rho$ at time $t$ and position $x$ derived from the spatially homogeneous equation~\eqref{eq:homogeneousKinetic}. Equation~\eqref{eq:BGK} assumes that the equilibrium distribution is known. For small values of $\epsilon$, the BGK model provides the same behavior of the full kinetic equation~\eqref{eq:generalKinetic}. In contrast, \eqref{eq:BGK} is only an approximation of~\eqref{eq:generalKinetic} in the regime with $\epsilon$ large. 

\begin{figure}[t!]
	\centering
	\includegraphics[width=0.49\textwidth]{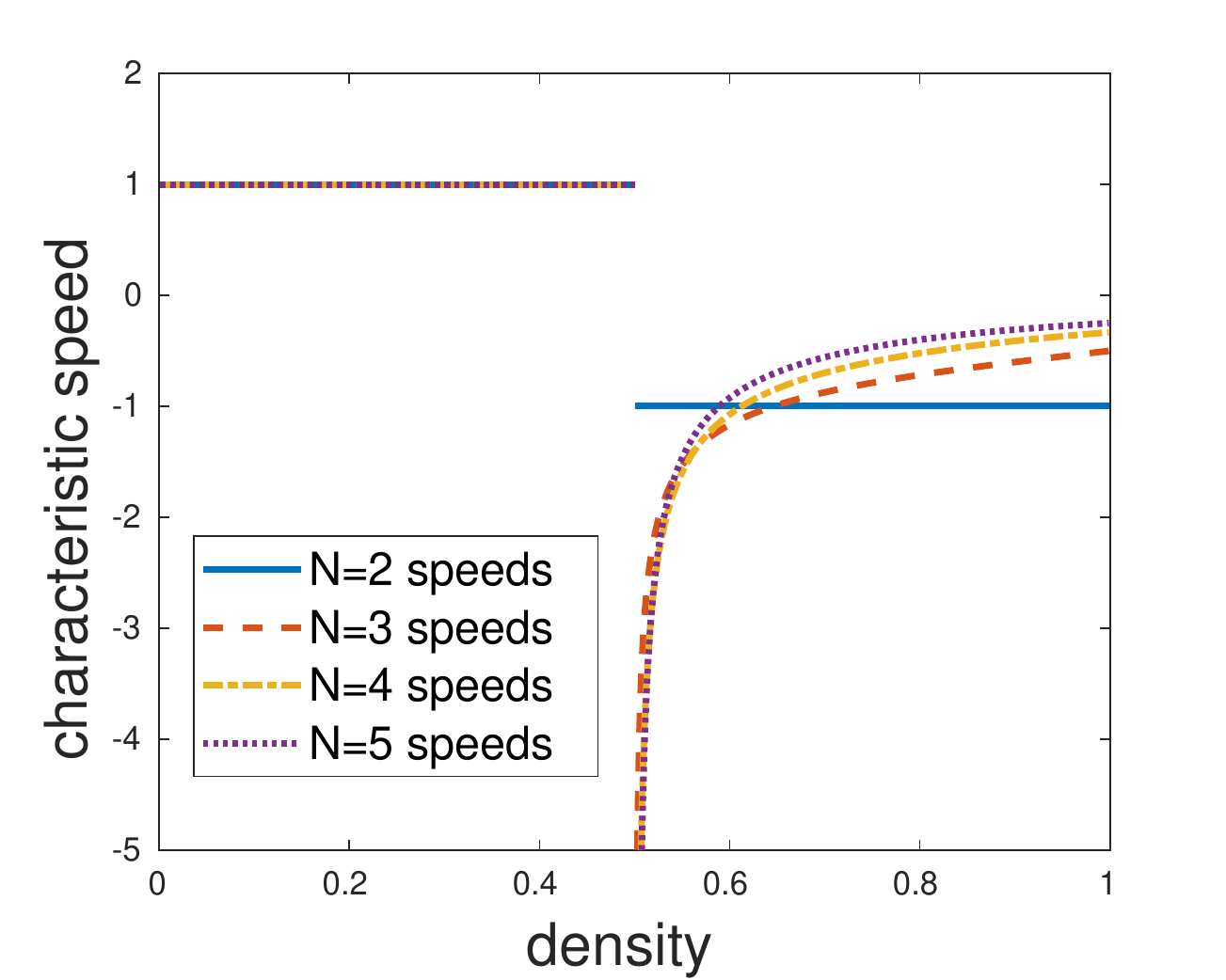}
	\includegraphics[width=0.49\textwidth]{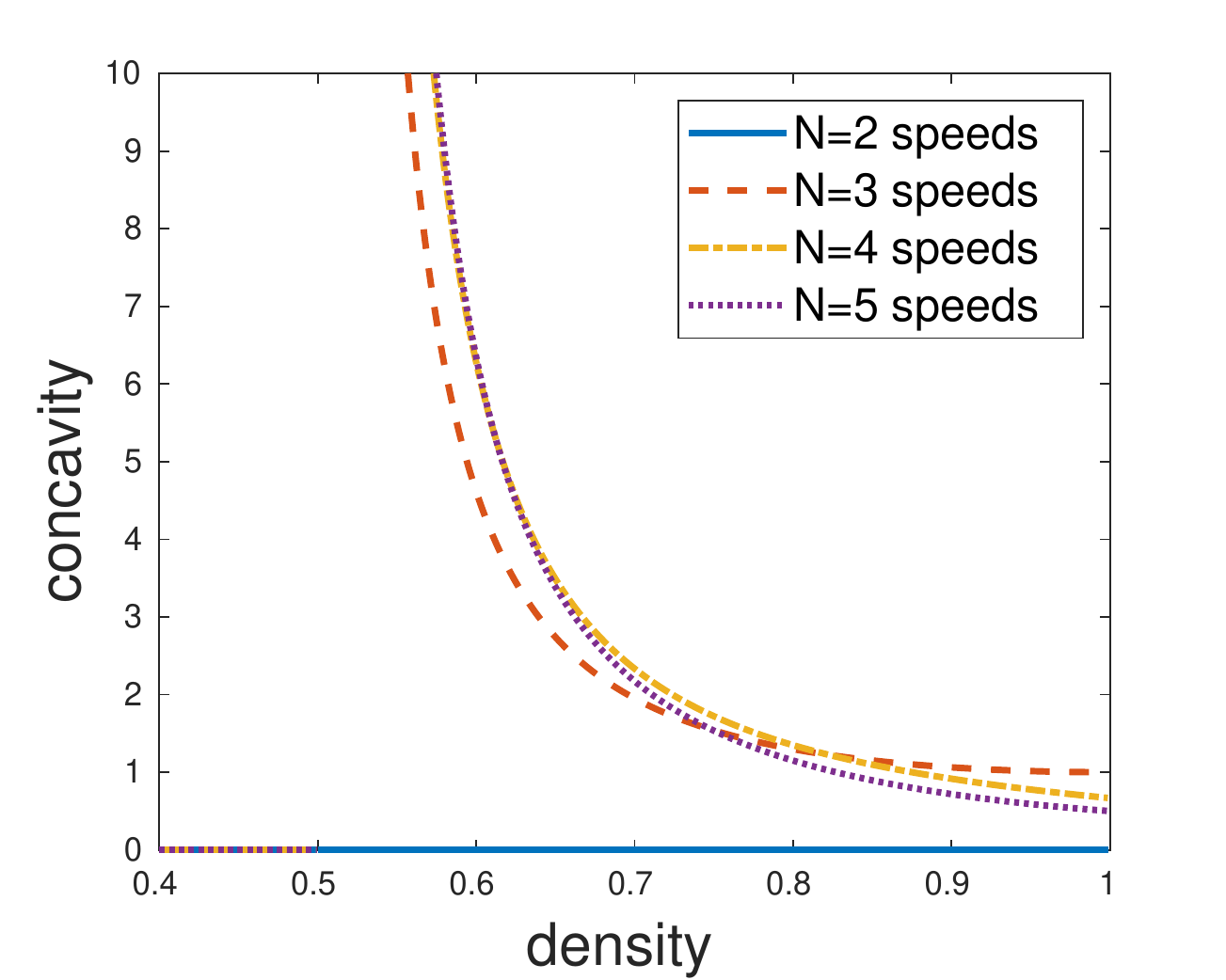}
	\caption{Left: characteristic speeds at equilibrium of the homogeneous kinetic model~\eqref{eq:homogeneousKinetic}. Right: concavity of the equilibrium flux $Q_\text{eq}$ given in~\eqref{eq:macroQuantitiesEq} for the homogeneous kinetic model~\eqref{eq:homogeneousKinetic}.\label{fig:propertiesFD}}
\end{figure}

\begin{remark}[Backward propagating information] \label{rem:backward}
	Let us consider $\epsilon$ fixed and $\epsilon \gg 1$, i.e. the regime in which the transport term in~\eqref{eq:BGK} is stronger then the relaxation term. It is well known that in this regime, since the velocities in traffic are non-negative, the information propagate in the wrong direction for congested traffic situations. In this phase, in fact, we expect to see backward propagation of signals but the transport part move them forward. This phenomenon is shared both by a full kinetic model~\eqref{eq:generalKinetic} and the BGK equation~\eqref{eq:BGK}. For this reason, several models with the ability of overcoming this drawback were introduced in the mathematical literature~\cite{FermoTosin13,klar1997Enskog}. This effect could be linked to a fixed choice of $\epsilon$. Instead, in analogy with the Knudsen number in gas models, $\epsilon$ should be a decreasing function of the density. In this way, $\epsilon$ is large in the free flow phase since the interactions are less frequent and the convective term rules the dynamics. On the contrary, since $\epsilon$ decreases when density increases, i.e. interactions among vehicles are dominant, we expect that relaxation is faster when $\rho$ is high. It is clear from the fundamental diagrams that, in the regimes in which $f$ is close to the local equilibrium $M_f$, signals move backward in congested phases. See also the left panel of Figure~\ref{fig:propertiesFD} where we show the characteristic speeds of the kinetic model when $f$ is close to equilibrium for several values of the speed number $N$. In other words we claim that the regimes described in Table~\ref{tab:epsilon} depend on the phase of traffic.
\end{remark}

\paragraph{Numerical scheme.} We briefly describe the numerical discretization of BGK using a first-order implicit-explicit (IMEX) scheme in time~\cite{PareschiRusso2005,PieracciniPuppo2006} by treating the convective term explicitly and the collision term implicitly. This allows to treat the stiffness for small values of $\epsilon$. Moreover, observe that, due to the linearity of the source term with respect to $f$, the scheme remains still explicit.

Let us consider the computational domain $(t,x,v)\in[0,T_M]\times[0,L]\times [0,1]$ where $T_M$ is the final time and introduce a mesh $I_{ij} = [x_{i-\frac12},x_{i+\frac12}]\times \{v_j\}$, where $x_i = (i-\frac12)\delta x$, $i=1,\dots,N_x$ and $v_j=(j-1)\delta v$, $j=1,\dots,N_v$. Here, the numerical parameter $\delta v$ is a sub-multiple of the physical parameter $\Delta v$ defined in Section~\ref{sec:homogeneous}. In this way, $N$, i.e. the number of discrete speeds defining the equilibrium distribution $M_f$, see Theorem~\ref{th:continuousEq}, is a sub-multiple of $N_v$. Although in the transient speeds that are not spaced by $\Delta v$ give a contribution to the dynamics, while at equilibrium only the $N$ speeds spaced by $\Delta v$ give a non-zero contribution, here we always consider the simple case with $\delta v=\Delta v$.
For piece-wise constant approximation of cell-averages in space we have  $f_{ij}^n \approx f(n\delta t,x_i,v_j)$ and obtain
\begin{subequations} \label{eq:fullScheme1}
	\begin{align}
	f_{ij}^{(1)} &= \frac{\epsilon}{\epsilon+\delta t} f_{ij}^n + \frac{\delta t}{\epsilon+\delta t}  M_f(v_j;\rho_i^n) \label{eq:stageValue}\\
	f_{ij}^{n+1} &= \frac{\epsilon}{\epsilon+\delta t} f_{ij}^n - \frac{\delta t}{\delta x} \left(\mathcal{F}\left(f_{i+1,j}^{(1)},f_{ij}^{(1)}\right)-\mathcal{F}\left(f_{ij}^{(1)},f_{i-1,j}^{(1)}\right)\right) + \frac{\delta t}{\epsilon+\delta t} M_f(v_j;\rho_i^n) \label{eq:stepBGK1}
	\end{align}
\end{subequations}
for $i=1,\dots,N_x$, $j=1,\dots,N_v$ and $n\geq 0$. Observe that the explicit knowledge of the equilibrium distribution $M_f$, given by Theorem~\ref{th:continuousEq}, is crucial for the performance of the numerical scheme. The quantity $\mathcal{F}$ can be one of the known Lipschitz continuous, conservative and consistent numerical flux function. The  asymptotic preserving~\cite{FilbetJin2010,PareschiRusso2011} result holds for the scheme~\eqref{eq:fullScheme1}.
\begin{proposition}[Asymptotic preserving] \label{th:AP}
	Scheme~\eqref{eq:fullScheme1} is asymptotic preserving in the sense that it solves a global first-order discretization of the macroscopic equation~\eqref{eq:modelLWR} with the flux function $Q_\text{eq}(\rho)$ given by~\eqref{eq:macroQuantitiesEq} when $\epsilon \to 0$.
\end{proposition}
\begin{proof}
	By equation~\eqref{eq:stageValue}, we have that $f_{ij}^{(1)} \to M_f(v_j;\rho_i^n)$, if $\epsilon \to 0$. Substituting $f_{ij}^{(1)}$ in~\eqref{eq:stepBGK1} then, for $\epsilon \to 0$, we find
	$$
		f_{ij}^{n+1} = M_f(v_j;\rho_i^n) - \frac{\delta t}{\delta x} \left(\mathcal{F}\left(M_f(v_j;\rho_{i+1}^n),M_f(v_j;\rho_i^n)\right)-\mathcal{F}\left(M_f(v_j;\rho_i^n),M_f(v_j;\rho_{i-1}^n)\right)\right),
	$$
	and hence 
	$$
		\rho_i^{n+1} = \rho_i^n - \frac{\delta t}{\delta x} \left(\mathcal{F}\left(\rho_{i+1}^n,\rho_i^n\right)-\mathcal{F}\left(\rho_i^n,\rho_{i-1}^n\right)\right). 
	$$
This scheme  is an explicit conservative, consistent and stable discretization of the macroscopic equation~\eqref{eq:modelLWR} with the flux function $Q_\text{eq}(\rho)$ given by~\eqref{eq:macroQuantitiesEq}, provided $\mathcal{F}(\cdot,\cdot)$ is a conservative, consistent and stable numerical flux function for the non-linear conservation law~\eqref{eq:modelLWR}.
\end{proof}

Hence, we choose the time-step $\delta t$ using the CFL condition applied to the non-linear conservation law~\eqref{eq:modelLWR} and in view of the nonlinear transport on macroscopic level we employ 
the Lax-Friedrichs flux:
$$
	\mathcal{F}\left(f_{i+1,j}^{(1)},f_{ij}^{(1)}\right) = \frac12\left(v_jf_{i+1,j}^{(1)}+v_jf_{ij}^{(1)} -\alpha\left(f_{i+1,j}^{(1)}-f_{i+1,j}^{(1)}\right)\right).
$$

\paragraph{Non-equilibrium effects.}

We show numerically that the multivaluedness in fundamental diagrams can be reproduced as a result of non-equilibrium effects.

\begin{figure}[t!]
	\centering
	\includegraphics[width=0.49\textwidth]{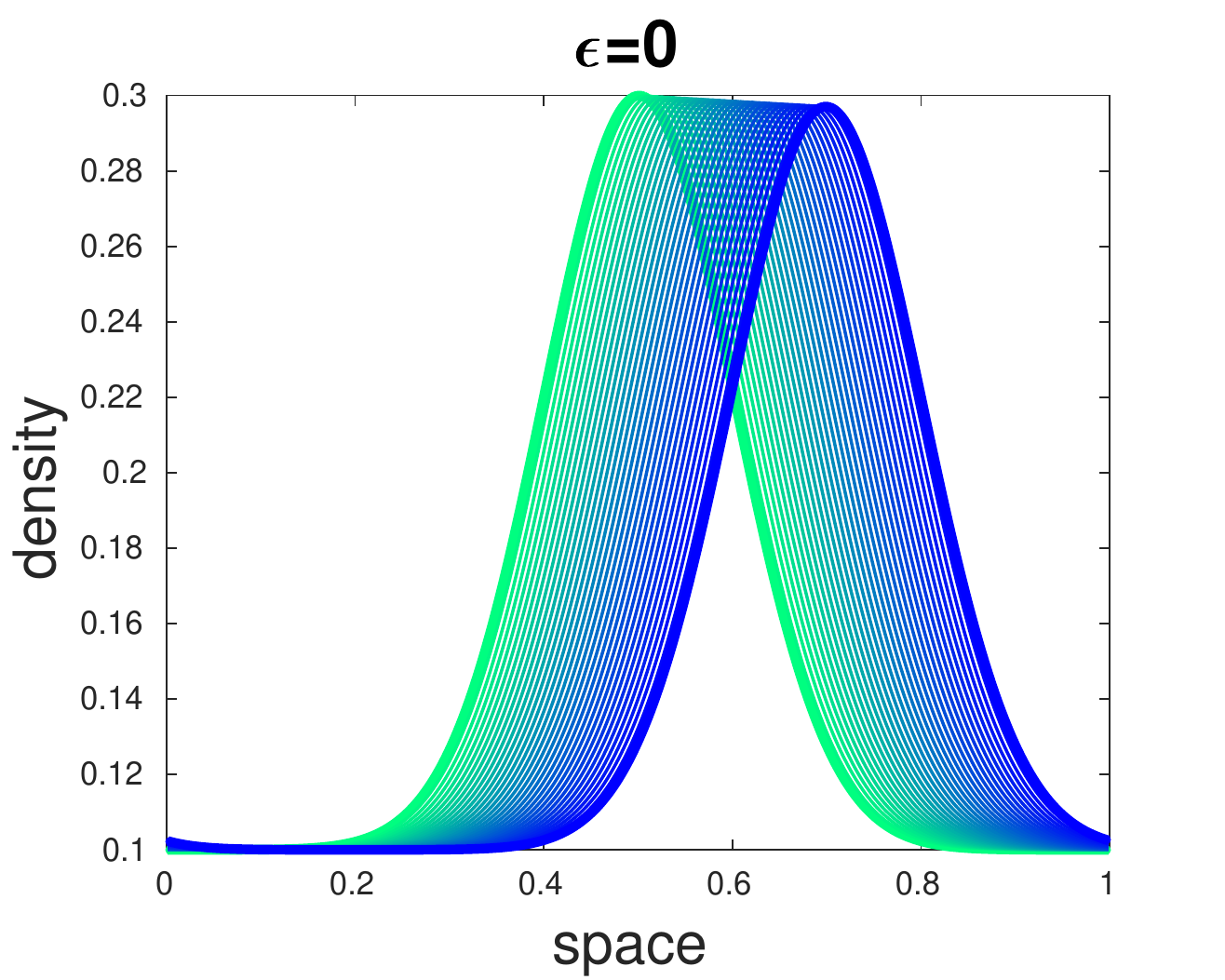}
	\includegraphics[width=0.49\textwidth]{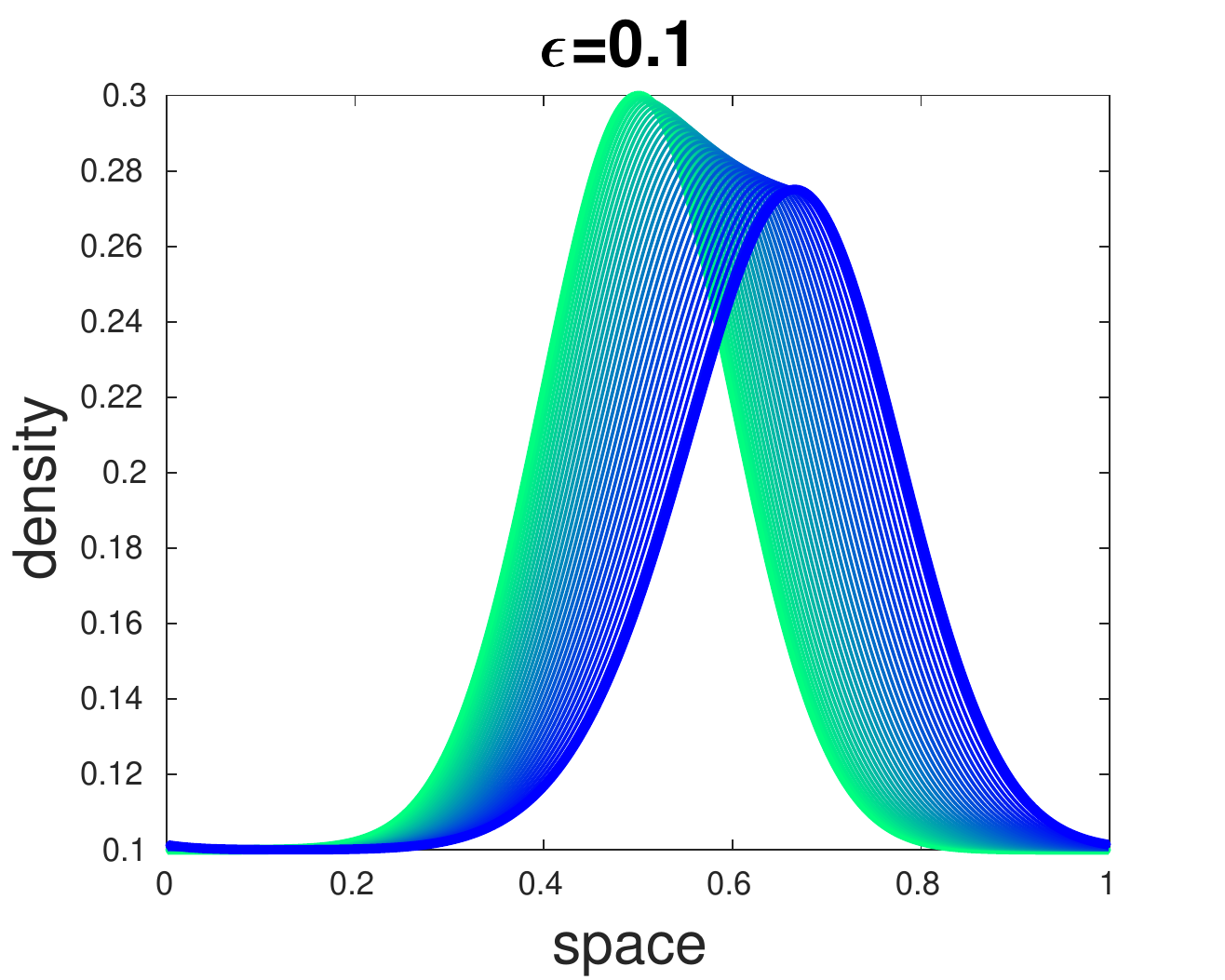}
	\\
	\includegraphics[width=0.49\textwidth]{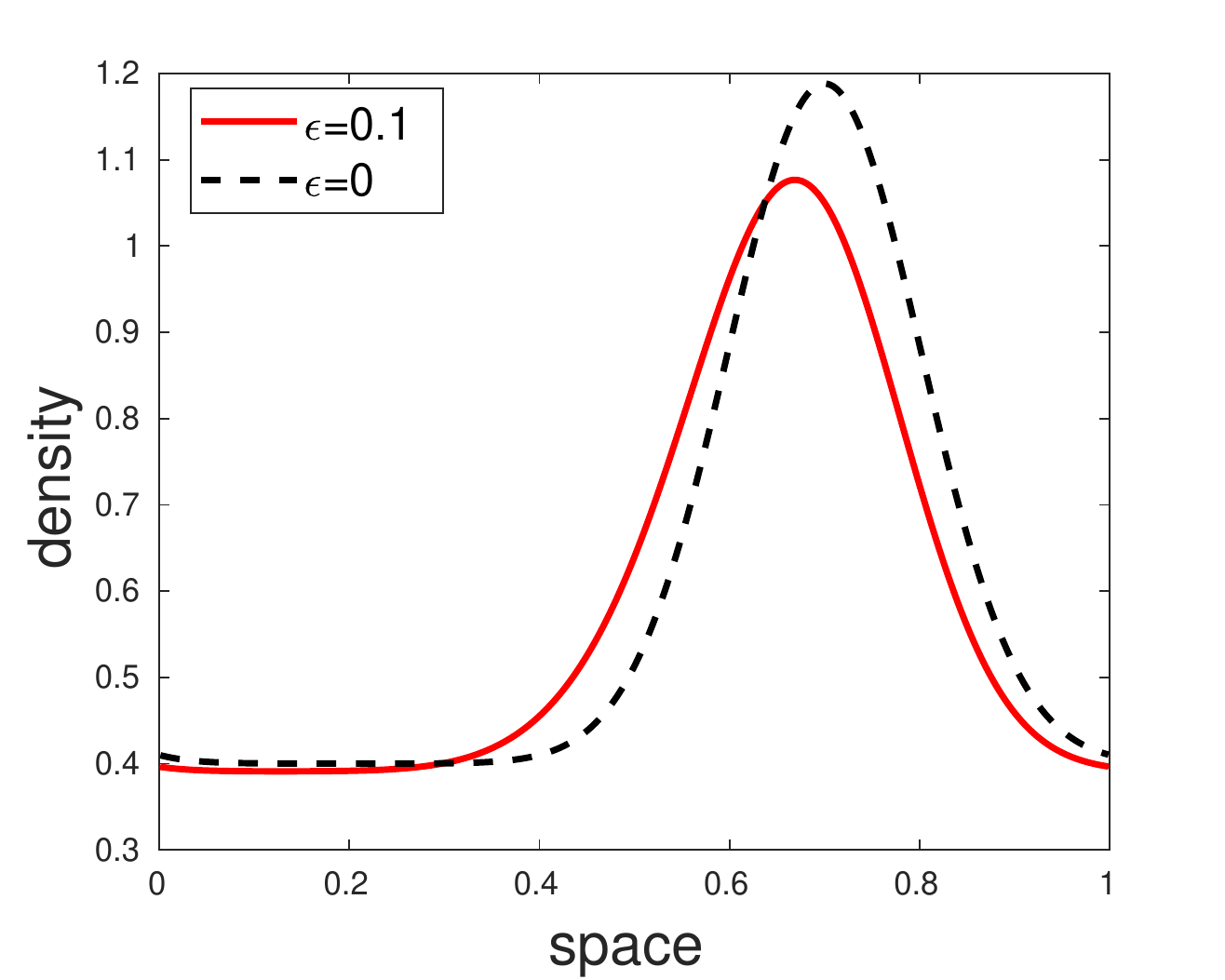}
	\includegraphics[width=0.49\textwidth]{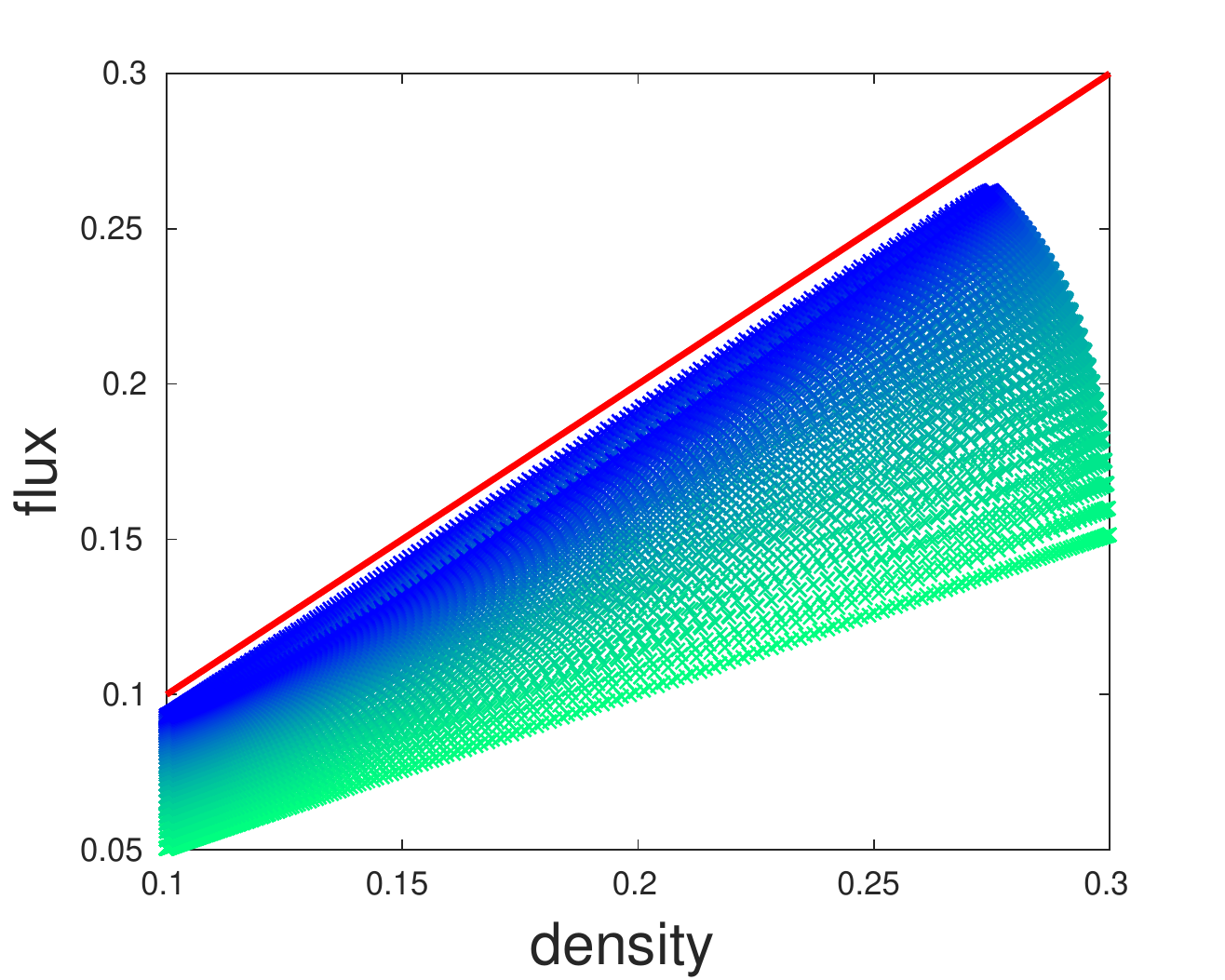}
	\caption{Propagation of a low density perturbation up to the final time $T_M=0.2$. Top left: solution at equilibrium with $\epsilon=0$. Top right: solution with $\epsilon=10^{-1}$. Bottom left: comparison of the two solutions at final time. Bottom right: flux-diagram during the evolution of the BGK model with $\epsilon=10^{-1}$ and at equilibrium (red line).\label{fig:noneqfree}}
\end{figure}

\begin{figure}[t!]
	\centering
	\includegraphics[width=0.49\textwidth]{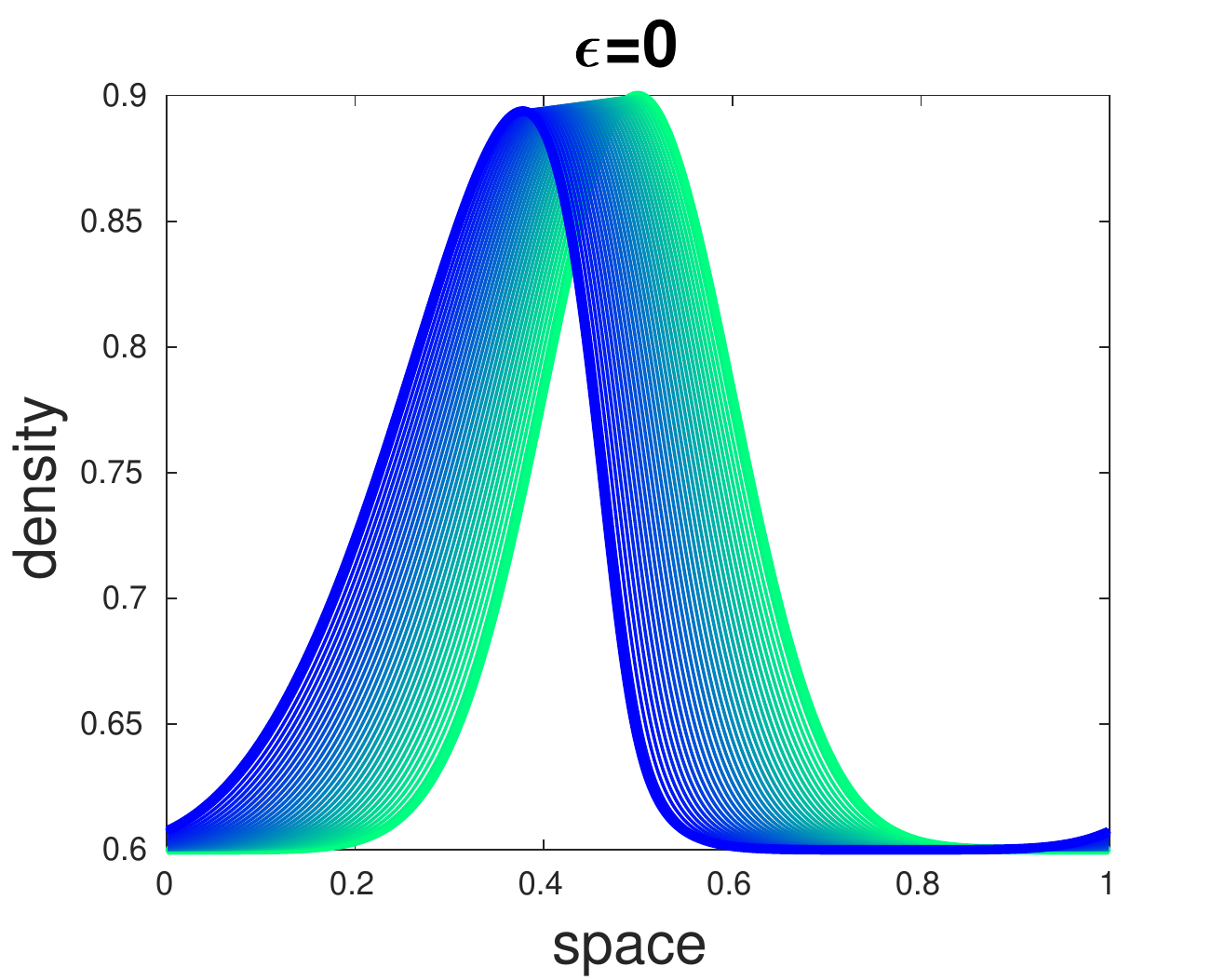}
	\includegraphics[width=0.49\textwidth]{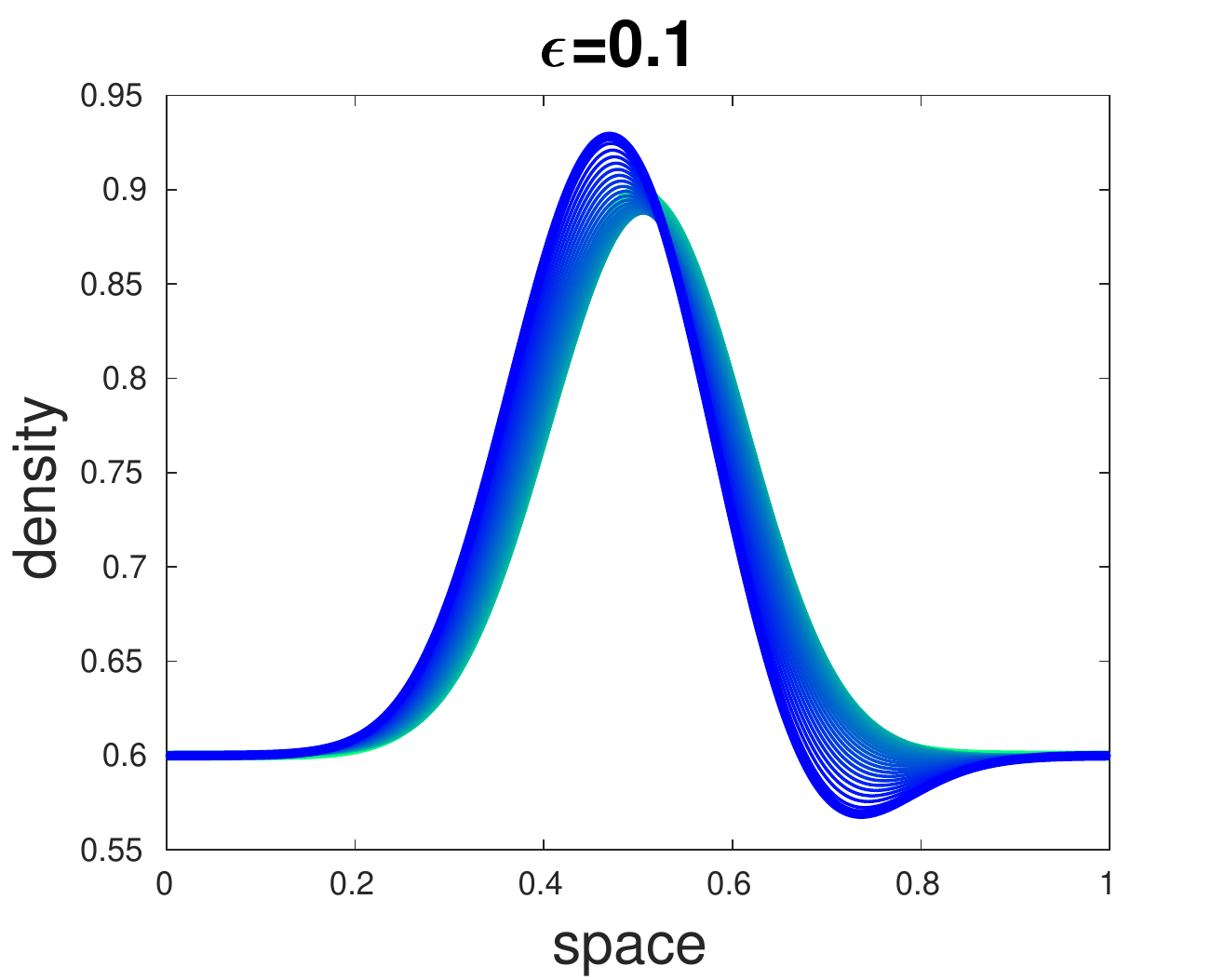}
	\\
	\includegraphics[width=0.49\textwidth]{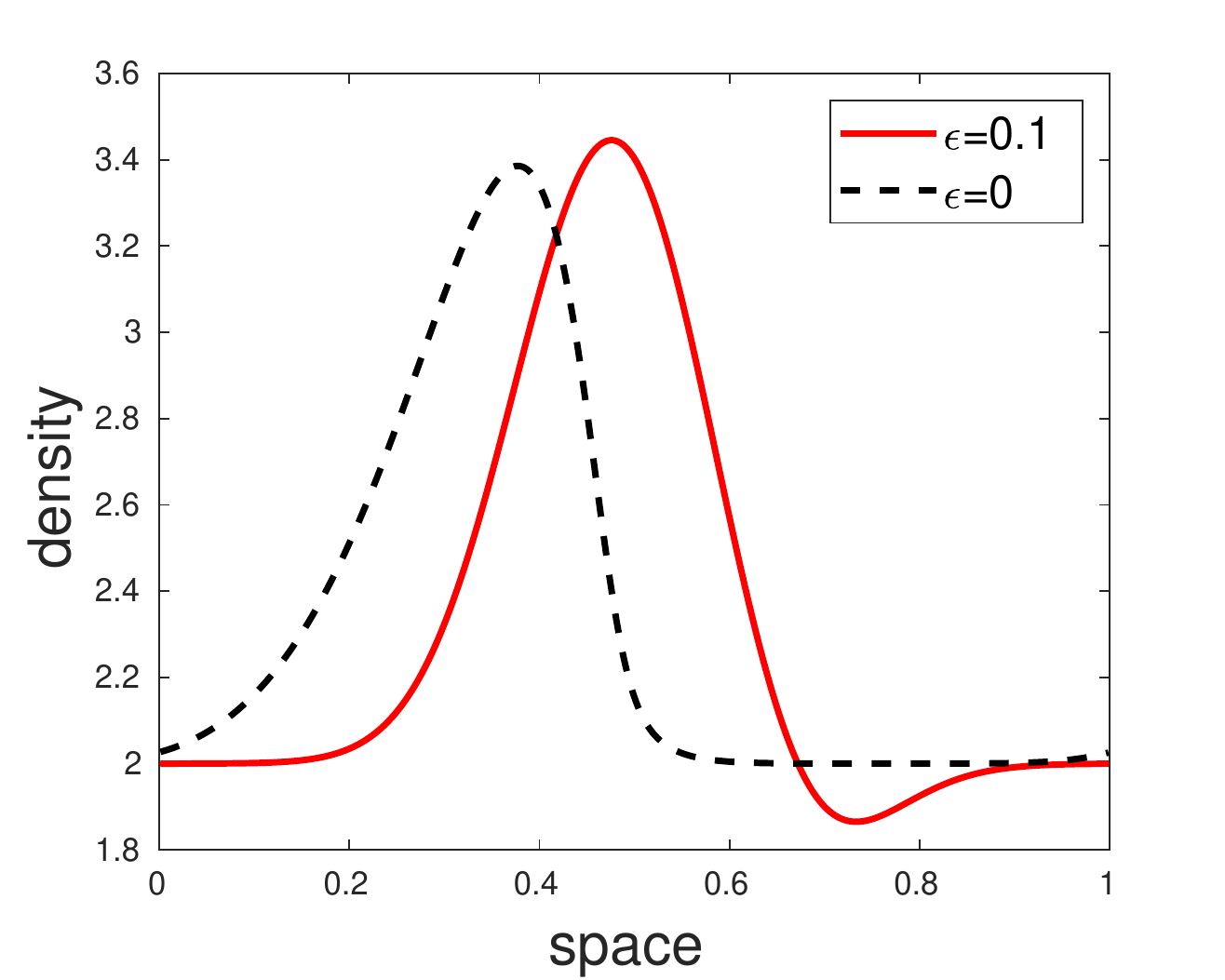}
	\includegraphics[width=0.49\textwidth]{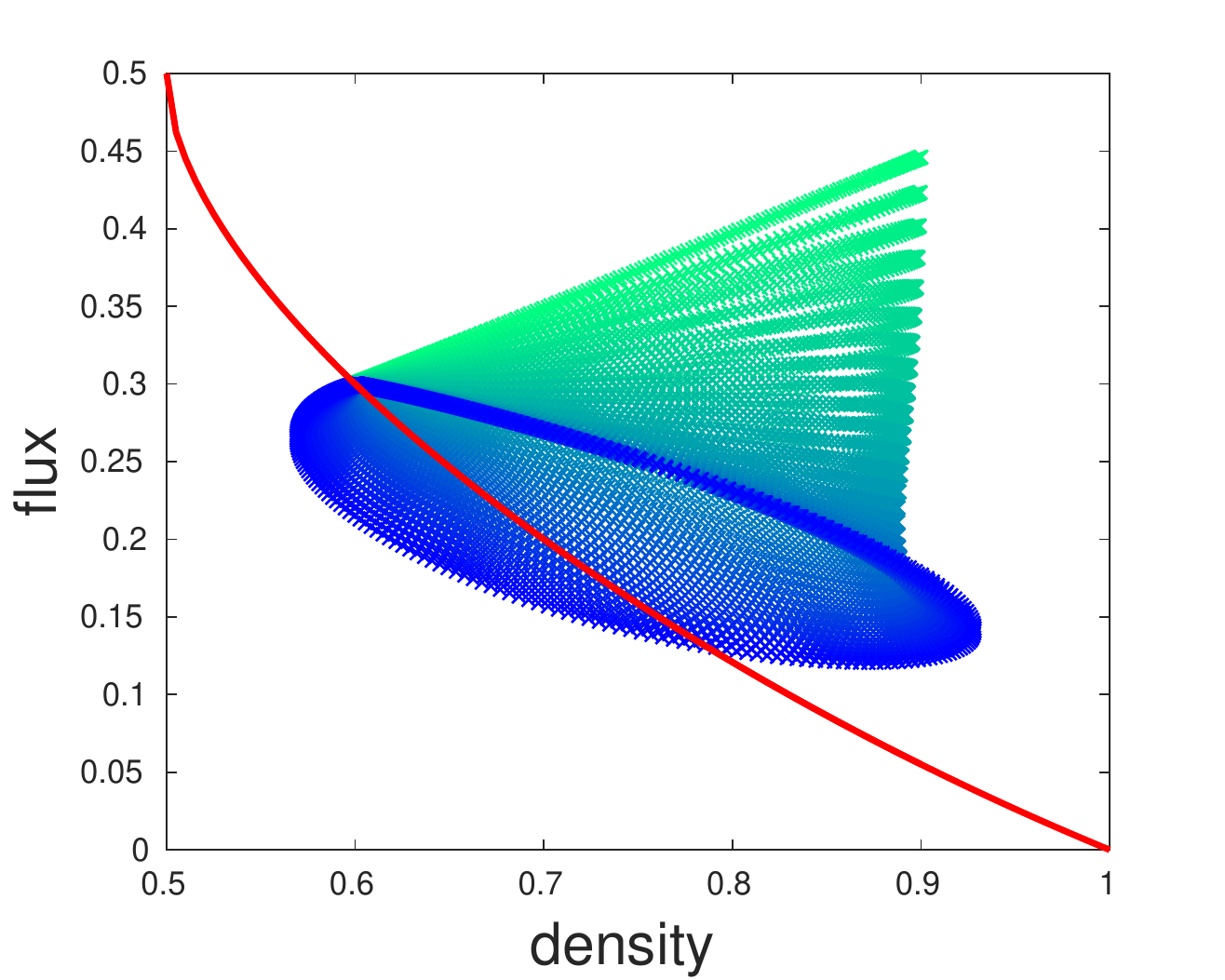}
	\caption{Propagation of a high density perturbation up to the final time $T_M=0.2$. Top left: solution at equilibrium with $\epsilon=0$. Top right: solution with $\epsilon=10^{-1}$. Bottom left: comparison of the two solutions at final time. Bottom right: flux-diagram during the evolution of the BGK model with $\epsilon=10^{-1}$ and at equilibrium (red line).\label{fig:noneqcongested}}
\end{figure}

\begin{figure}[t!]
	\centering
	\includegraphics[width=0.49\textwidth]{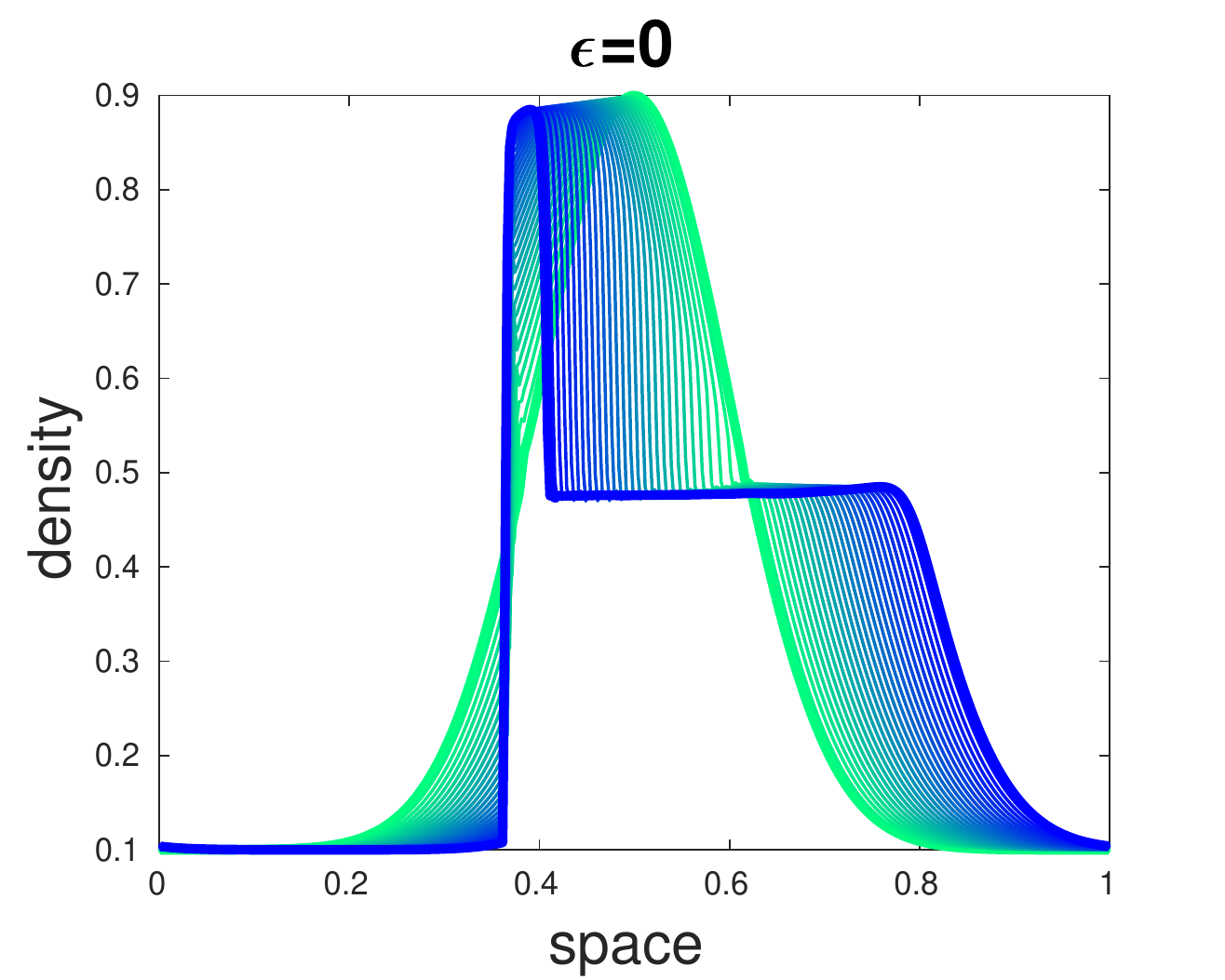}
	\includegraphics[width=0.49\textwidth]{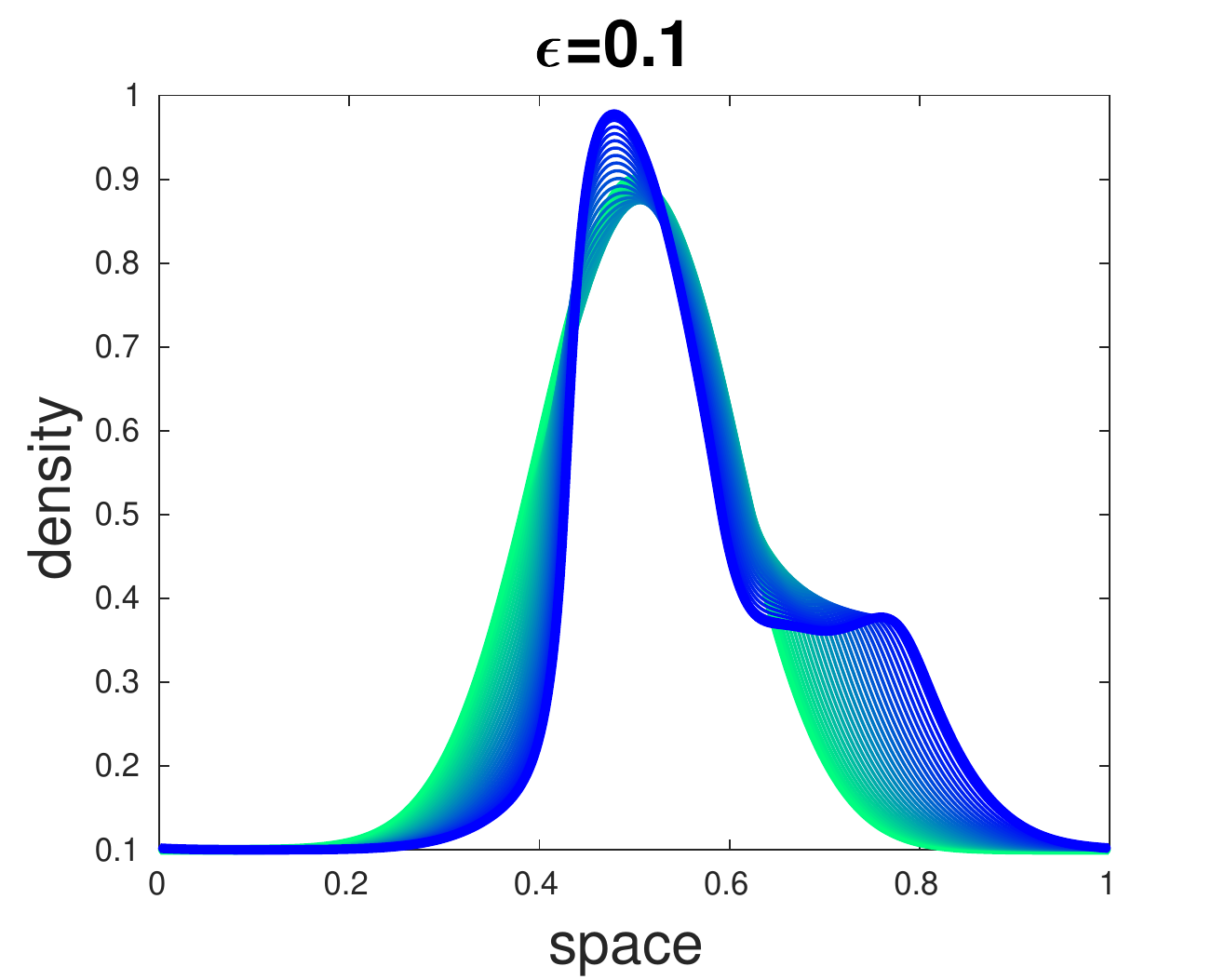}
	\\
	\includegraphics[width=0.49\textwidth]{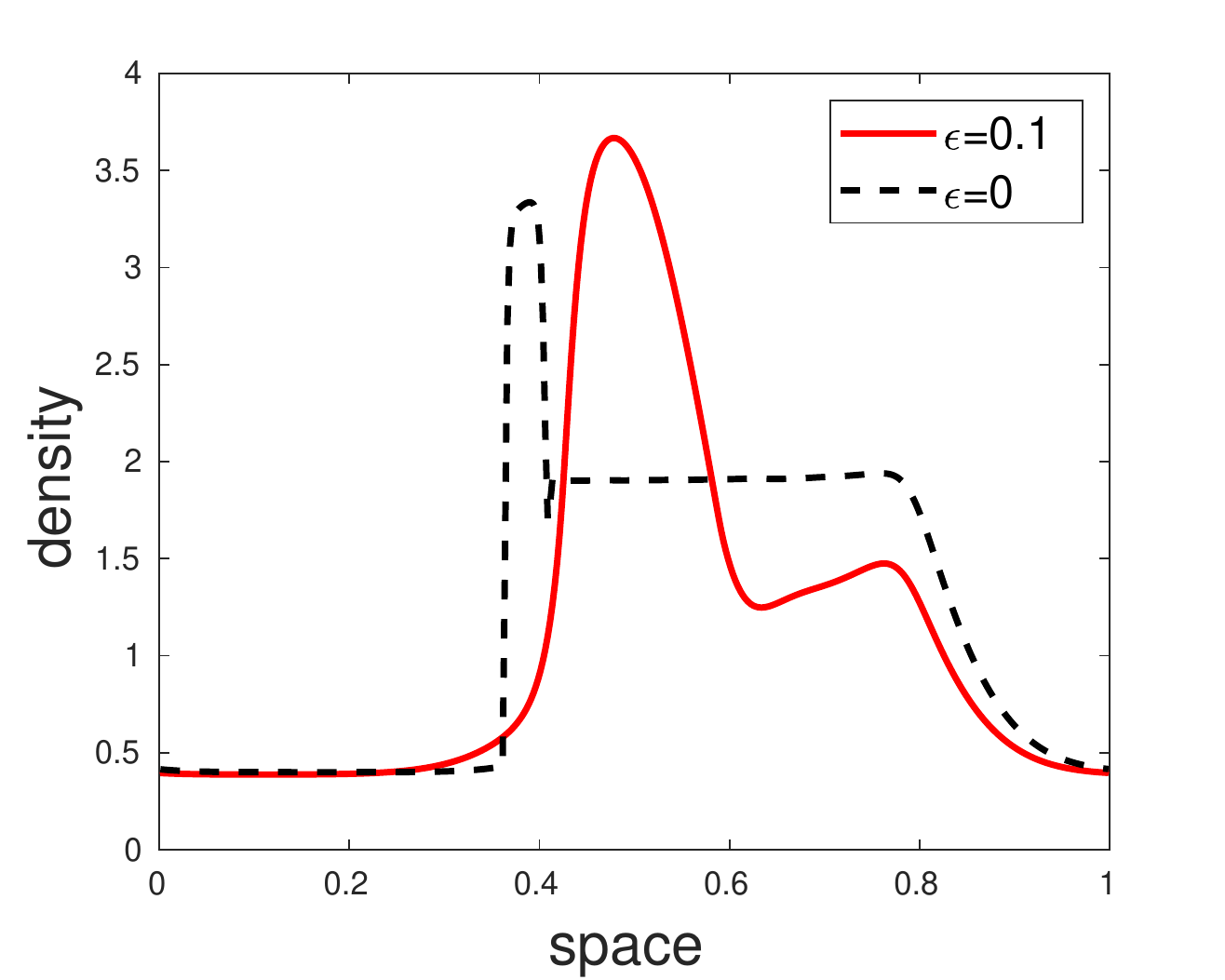}
	\includegraphics[width=0.49\textwidth]{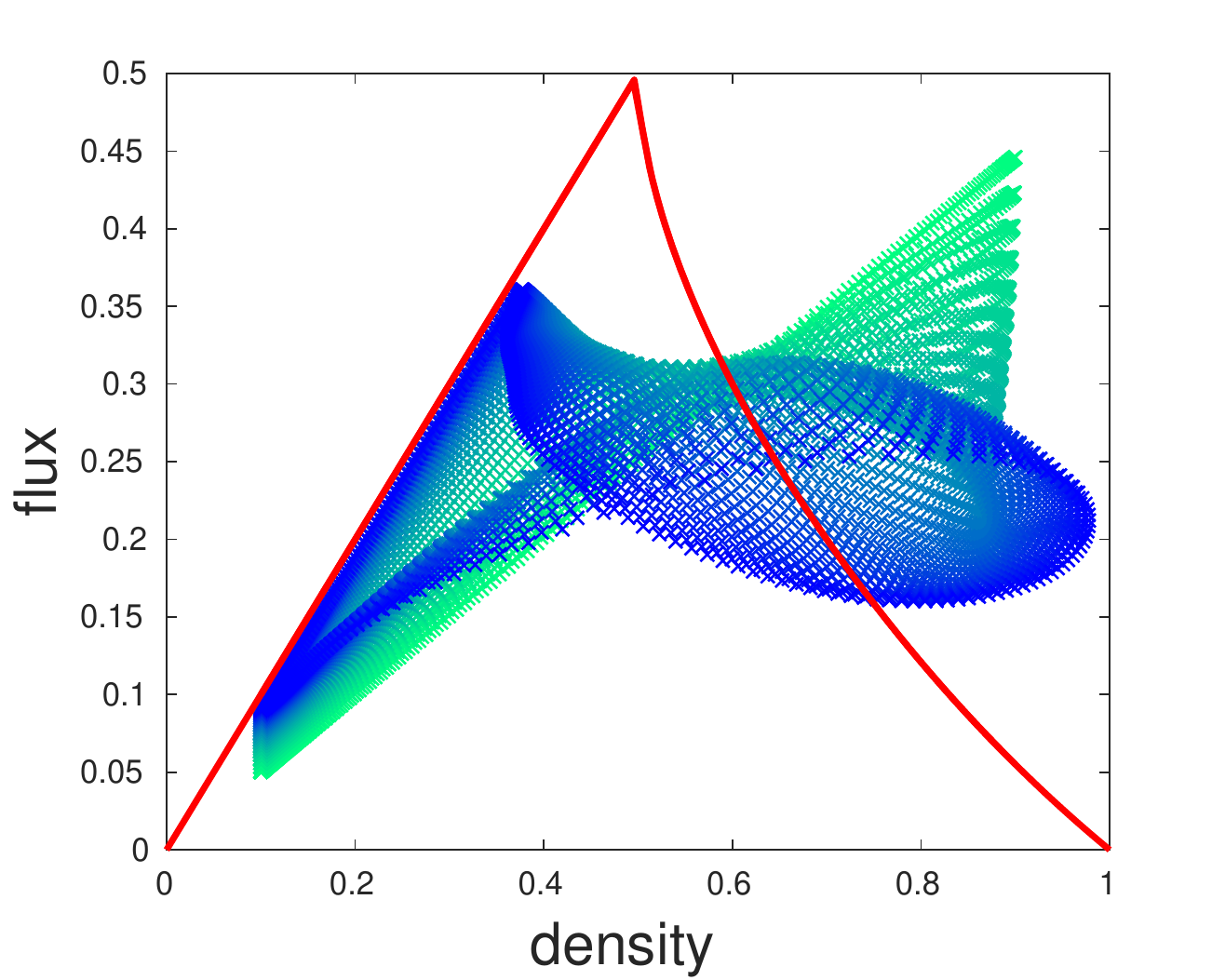}
	\caption{Propagation of a large perturbation up to the final time $T_M=0.2$. Top left: solution at equilibrium with $\epsilon=0$. Top right: solution with $\epsilon=10^{-1}$. Bottom left: comparison of the two solutions at final time. Bottom right: flux-diagram during the evolution of the BGK model with $\epsilon=10^{-1}$ and at equilibrium (red line).\label{fig:noneqtotal}}
\end{figure}

The top rows of Figure~\ref{fig:noneqfree}, \ref{fig:noneqcongested} and \ref{fig:noneqtotal} contain the solution obtained with $\epsilon=0$, i.e. in the equilibrium regime, on the left, and the solution found with $\epsilon=10^{-1}$ on the right. The solution is shown in several instants of time, starting from the green curve and ending with the blue thick profile. The final solutions are compared in the bottom left part of the figure: the profile with $\epsilon=10^{-1}$ is depicted with the continuous red line, and the profile with $\epsilon=0$ with the dashed black line.

The dots in the right bottom panel contain the evolution of the flux $\rho u$ defined in~\eqref{eq:macroQuantities} which is drawn as a function of $\rho$. Note that, at equilibrium, the fundamental diagram $Q_\text{eq}(\rho)$ is a function of $\rho$. When the system is not in equilibrium, the flux $\rho u$ depends not only on $\rho$, but also on the distribution $f$. 
Again, the green dots are obtained in the initial stages of the simulation, and they fade into blue as the solution evolves. If the flow were in equilibrium, as in the case $\epsilon=0$, then $f$  coincides with the equilibrium distribution $M_f$ (uniquely determined by $\rho$). Thus the dispersion in these plots is a measure of deviations from equilibrium. Hence, the dispersion of the points in blue may reflect the multivaluedness of the fundamental diagram

In all simulations we consider $N=N_v=3$ speeds and so the physical acceleration is $\Delta v=\frac12$. The spatial domain is $[0,1]$ and discretized by $N_x=400$ cells. The final time is $T_M=0.2$. We consider an initial smooth perturbation in the density
\begin{equation} \label{eq:densityBump}
	\rho_0(x) = \rho_{\min} + \left(\rho_{\max}-\rho_{\min}\right) \exp\left(-50(x-0.5)^2\right).
\end{equation}
with periodic boundary conditions. The initial kinetic distribution $f(t=0,x,v)$ is the uniform distribution over the velocity space with density $\rho_0(x)$.

In Figure~\ref{fig:noneqfree}, $\rho_{\min}=0.1$ and $\rho_{\max}=0.3$. The perturbation in the density is below the critical density $\rho_\text{crit}=\frac12$. Thus here the whole solution moves towards the right for $\epsilon=0$ because in this regime the flux is linear with respect to the density. In the case $\epsilon=10^{-1}$ the solution also moves to the right but with very little deformation. As consequence, the flux diagram shows a very small dispersion, which means that the flow is almost at equilibrium.

If we choose $\rho_{\min}=0.6$ and $\rho_{\max}=0.9$, the initial perturbation occurs on a dense traffic flow, see Figure~\ref{fig:noneqcongested}. Now, the expected propagation speed in the fundamental diagram is negative, and in fact the solution for $\epsilon=0$ moves towards the left. The solution with $\epsilon=10^{-1}$ also moves towards the left, but new maxima and minima appear which means that the maximum principle is no longer satisfied. These are purely non-equilibrium effects, as shown also in the flux diagram. The equilibrium fundamental diagram is shown with the red line, and it can be seen bisecting the $\rho u$ values computed during the evolution of the case $\epsilon=10^{-1}$.

The last smooth test we show has a large initial pulse, with $\rho_{\min}=0.1$ and $\rho_{\max}=0.9$. The results are in Figure~\ref{fig:noneqtotal}. The low density part of the wave on the right accelerates and propagates towards the right. The high density part of the wave propagates backward, increasing the slope of the profile. Thus a shock forms, propagating towards the left. This structure appears quickly in the case $\epsilon=0$ while it is slower for $\epsilon=10^{-1}$ where however we see the appearance of a new maximum. Again the solution with $\epsilon=10^{-1}$ exhibits non equilibrium waves which are represented by the scattered values in the flux diagram.

\begin{remark}[Solutions of Riemann problems]
	As already observed, in the relaxation limit the full kinetic equation~\eqref{eq:generalKinetic} and the BGK approximation~\eqref{eq:BGK} solve the conservation law~\eqref{eq:modelLWR} with $Q_\text{eq}$ computed using the equilibrium distribution of the homogeneous kinetic model. We notice that the fundamental relation $Q_\text{eq}$ defined in~\eqref{eq:macroQuantitiesEq} has two main properties.
	\begin{enumerate}
		\item It is continuous: in fact, from Theorem~\ref{th:continuousEq} it is clear that
		\[
		\lim_{P \to \frac12^+} \int_{\mathcal{V}} v M_f(v;\rho) \mathrm{d}v = \lim_{P \to \frac12^-} \int_{\mathcal{V}} v M_f(v;\rho) \mathrm{d}v = \rho_\text{crit} V_M.
		\]
		Thus, there is no need to treat the macroscopic model~\eqref{eq:modelLWR} with special techniques for conservation laws with discontinuous fluxes.
		\item It is not convex: in fact, it is linear in the free flow regime (i.e. when $\rho\leq\rho_\text{crit}$) and it is convex in the congested regime (i.e. when $\rho>\rho_\text{crit}$). See the right panel in Figure~\ref{fig:propertiesFD} which shows the behavior of the second derivative of $Q_\text{eq}$ for different values of the speed number.
	\end{enumerate}
	
	The non-convexity of the flux-density diagram represents the most interesting property since it implies that the wave structure at equilibrium is non-trivial, because even a simple Riemann problem initial data will result in more than a single wave.
	%
	For a more
	detailed discussion we refer to~\cite{BressanBook,LevequeBook}.
	
	
	%
\end{remark}

\paragraph{Unbounded instabilities.} 

The results of the previous simulations show that non-homogeneous models are important to describe non-equilibrium phenomena. Note that, in the case of the perturbation in dense traffic, the non-equilibrium effects provides new maxima and minima that apparently grow in time. These results suggest  that the appearance of stop and go waves on highways could be interpreted and reproduced as non-equilibrium effects.

\begin{figure}[t!]
	\centering
	\includegraphics[width=0.49\textwidth]{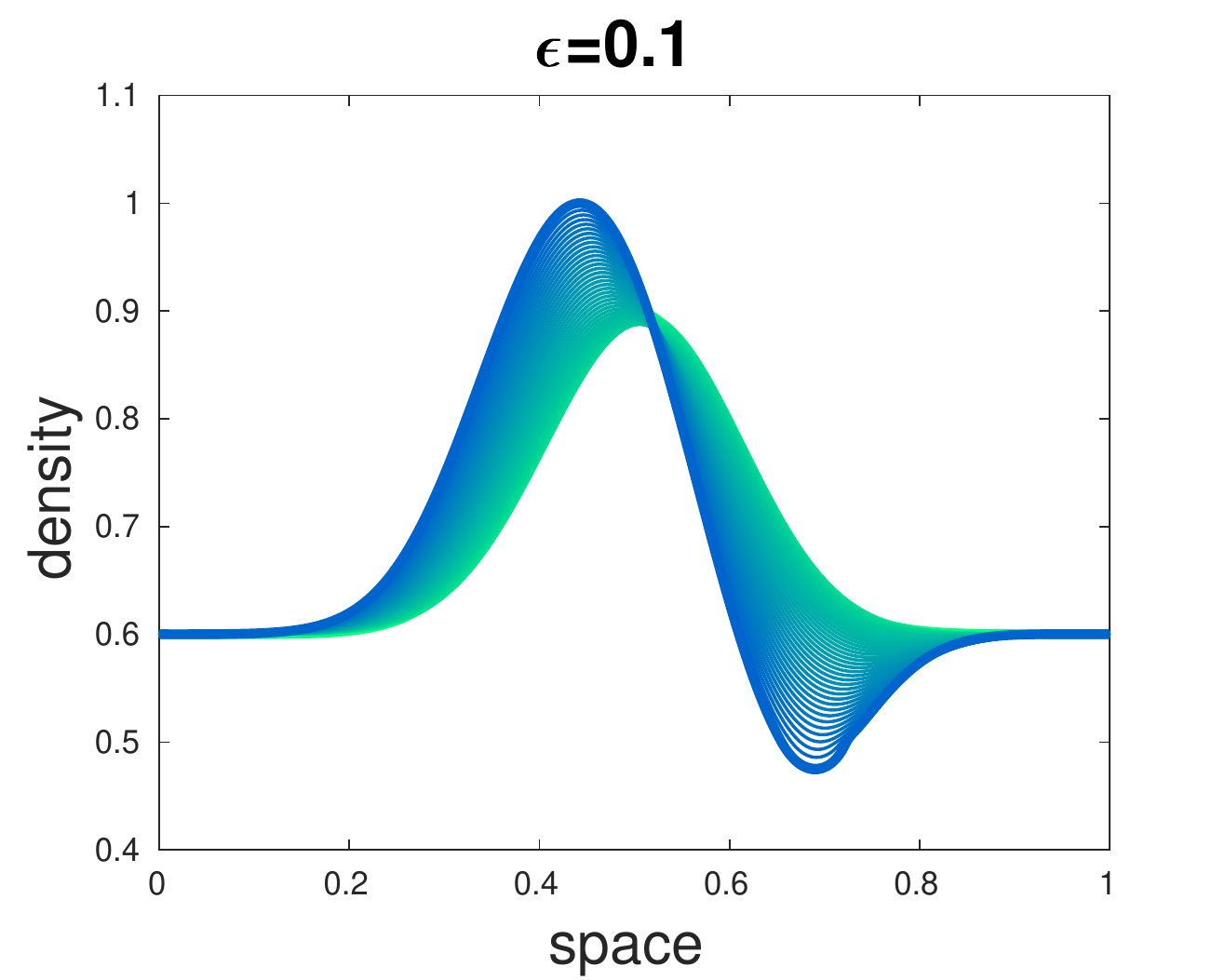}
	\includegraphics[width=0.49\textwidth]{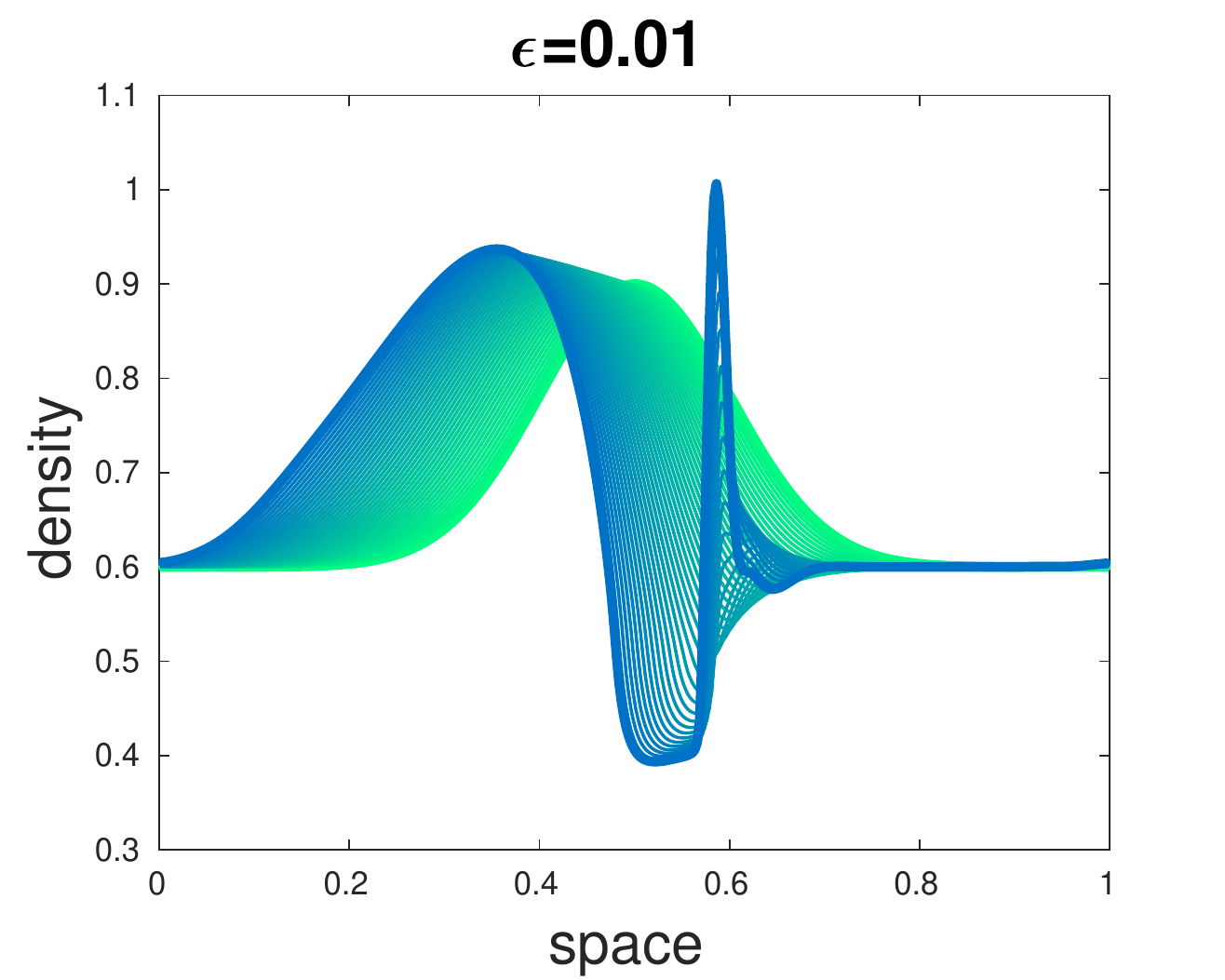}
	\caption{Propagation of a high density perturbation up to the time in which the density profile is bounded by the maximum value of the density. Left: solution with $\epsilon=10^{-1}$. Right: solution with $\epsilon=10^{-2}$.\label{fig:nonequnstable}}
\end{figure}

Here, we prove that the BGK model~\eqref{eq:BGK} has the drawback that it could lead to unbounded unstable waves in the congested regime of traffic. In fact, let us consider the same setting of the simulations provided in Figure~\ref{fig:noneqcongested}, namely a perturbation of the density in the congested regime. In Figure~\ref{fig:nonequnstable} we show the solutions for a larger final time with $\epsilon=10^{-1}$ in the left panel and $\epsilon=10^{-2}$ in the right panel. We stop the simulations as soon as  the density profile exceeds the maximum density. We observe that the BGK model propagates the initial profile backward producing unstable waves that do not remain bounded. This phenomenon will be  investigated below using the Chapman-Enskog expansion. 

Consider the Chapman-Enskog expansion of the BGK equation~\eqref{eq:BGK} as first order perturbation in $\epsilon$ of $f$ around the equilibrium distribution $M_f(v;\rho)$. This leads to the advection-diffusion equation
\begin{equation} \label{eq:advdiff}
	\partial_t \rho + \partial_x Q_{\text{eq}}(\rho) = \epsilon \partial_x (\mu(\rho) \rho_x),
\end{equation}
where the diffusion coefficient $\mu(\rho)$ is given by 
\begin{equation} \label{eq:diffBGK}
	\mu(\rho) = \int_V v^2 \partial_\rho M_f(v;\rho) \mathrm{d}v - \left(\int_V v \partial_\rho M_f(v;\rho) \mathrm{d}v\right)^2 = \int_V v^2 \partial_\rho M_f(v;\rho) \mathrm{d}v - Q_\text{eq}^\prime(\rho)^2.
\end{equation}
In fact, considering the expansion
$$
	f(t,x,v) =  M_f(v;\rho) + \epsilon f_1(t,x,v), \quad \text{with } \int_V f_1(t,x,v) \mathrm{d}v = 0,
$$
observing that
$$
	f_1(t,x,v) = -\partial_t M_f(v;\rho) - v \partial_x M_f(v;\rho) + O(\epsilon),
$$
plugging the approximation of $f$ into the BGK equation~\eqref{eq:BGK} and integrating with respect to the velocity we obtain
\begin{align*}
	\partial_t \rho + \partial_x Q_\text{eq}(\rho) &= - \epsilon \partial_x \int_V v f_1(t,x,v) \mathrm{d}v\\
	&= \epsilon \partial_x \int_V v (\partial_t M_f(v;\rho) + v \partial_x M_f(v;\rho)) \mathrm{d}v + O(\epsilon^2)\\
	&= \epsilon \partial_x \left[ \rho_t \int_V v \partial_\rho M_f(v;\rho) \mathrm{d}v + \rho_x \int_V v^2 \partial_\rho M_f(v;\rho) \right] + O(\epsilon^2)\\
	&= \epsilon \partial_x \left[ -\rho_x Q_\text{eq}^\prime(\rho) \int_V v \partial_\rho M_f(v;\rho) \mathrm{d}v + \rho_x \int_V v^2 \partial_\rho M_f(v;\rho) \right] + O(\epsilon^2).
\end{align*}

If $\mu(\rho)<0$ then the advection-diffusion equation is ill-posed and therefore has solutions with unbounded growth even starting from small perturbations. In the case of the kinetic model~\eqref{eq:BGK}, the sign of the diffusion coefficient depends on the equilibrium distribution $M_f$. The request $\mu(\rho) > 0$ is
\begin{equation} \label{eq:stabilityBGK}
	\partial_\rho \left( \int_V v^2 M_f(v;\rho) \mathrm{d}v \right) > Q_\text{eq}^\prime(\rho)^2
\end{equation}
and, since $Q_\text{eq}(\rho)$ is the fundamental diagram at equilibrium, this condition requires that the square of the characteristic velocities is bounded by the variation of the kinetic energy in each regime. The condition~\eqref{eq:stabilityBGK} depends only on the equilibrium distribution $M_f$ which is given by the homogeneous kinetic model and is the result of the modeling of the microscopic interactions. 

\begin{figure*}[t!]
	\centering
	\begin{subfigure}[t]{0.5\textwidth}
		\centering
		\includegraphics[width=\textwidth]{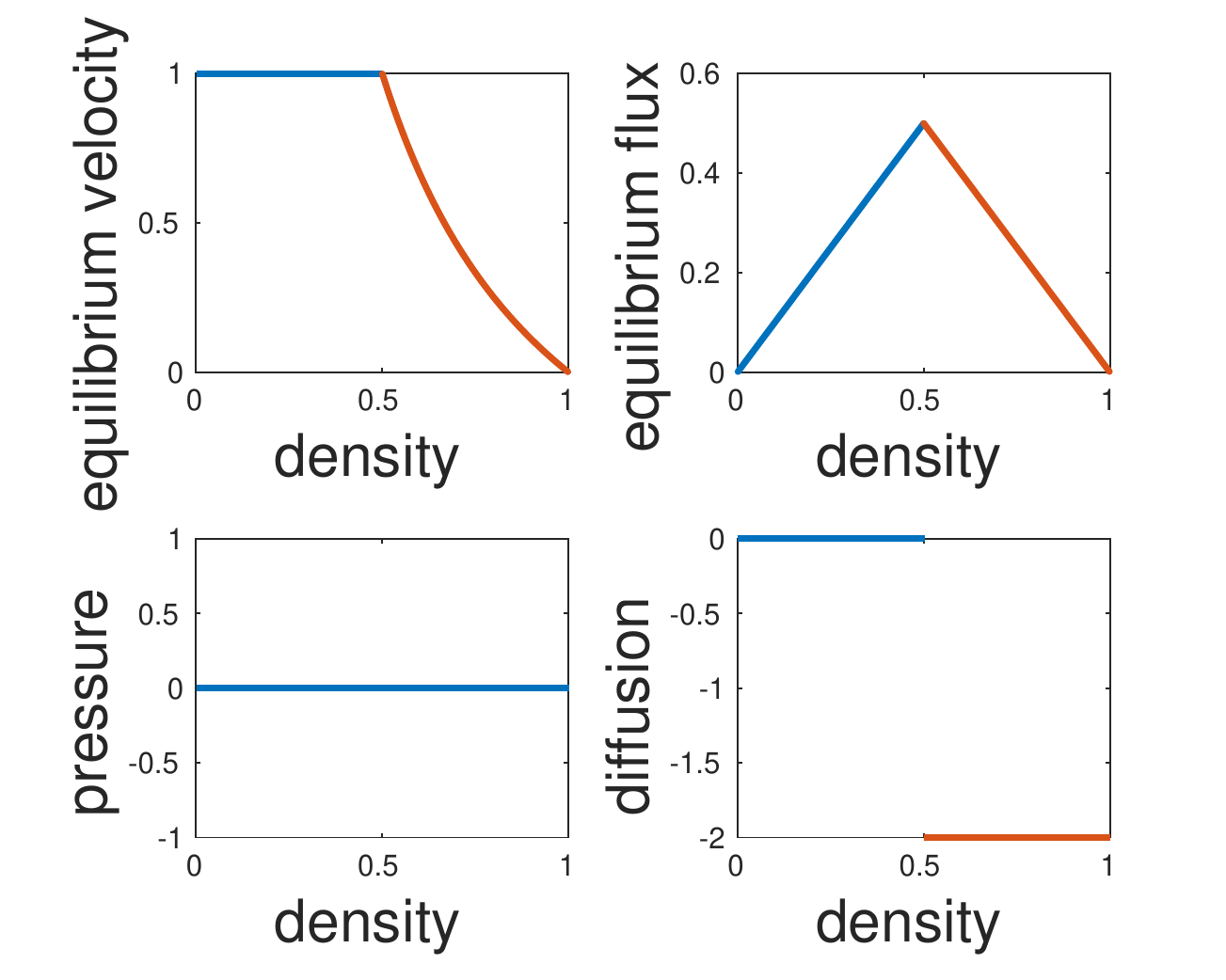}
		\caption{Two speeds.\label{fig:diffBGKa}}
	\end{subfigure}%
	~
	\begin{subfigure}[t]{0.5\textwidth}
		\centering
		\includegraphics[width=\textwidth]{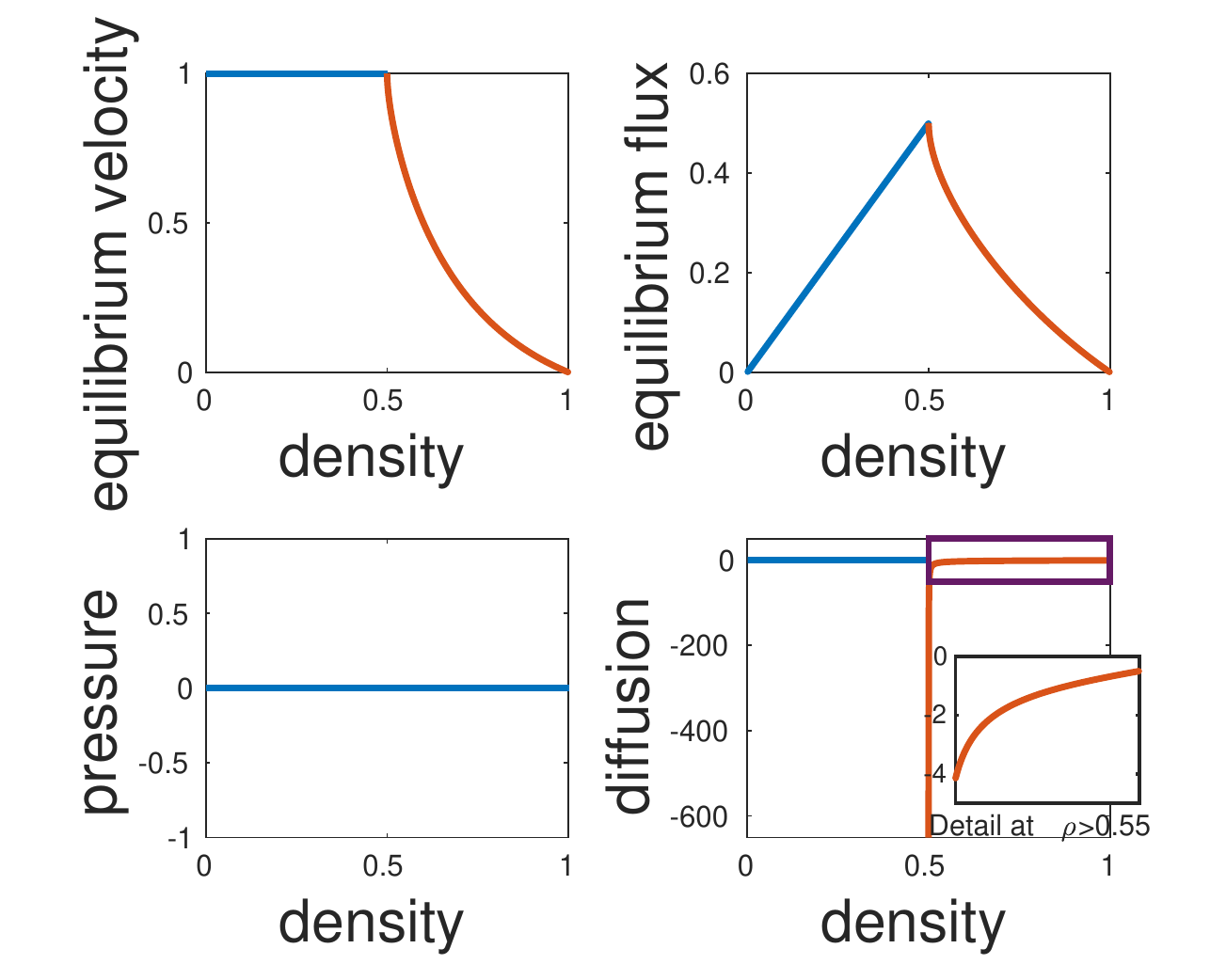}
		\caption{Three speeds.\label{fig:diffBGKb}}
	\end{subfigure}
	\caption{The right bottom panels show the sign of the diffusion coefficient~\eqref{eq:diffBGK} for the BGK model~\eqref{eq:BGK} with the equilibrium distribution of the homogeneous kinetic model~\eqref{eq:homogeneousKinetic}-\eqref{eq:rule}. The blue line corresponds to the positive sign of the coefficient and the red line to the negative sign.\label{fig:diffBGK}}
\end{figure*}

In Figure~\ref{fig:diffBGK} we show the sign of the diffusion coefficient~\eqref{eq:diffBGK} in the case of the equilibrium distribution computed in~\cite{PSTV2}, and recalled in Section~\ref{sec:homogeneous}, with two and three speeds corresponding to $\Delta v=1$ and $\Delta v=\frac12$, respectively. We observe that $\mu(\rho) \geq 0$ if $\rho \leq \rho_\text{crit}=\frac12$, while $\mu(\rho) < 0$ if $\rho > \rho_\text{crit}=\frac12$. This result could explain the unbounded growth of perturbations in the density, numerically observed in Figure~\ref{fig:nonequnstable}.

In the following result, we prove that the instability of the model does not depend on the choice of the equilibrium distribution but it occurs for any suitable equilibrium of kinetic traffic models.

\begin{proposition} \label{th:negativeDiff}
	Assume that $\exists \, \widetilde{\rho} \in (0,1)$ such that
	\begin{equation} \label{eq:conditions}
		Q_\text{eq}^\prime(\rho) = \int_v v \partial_\rho M_f(v;\rho) \mathrm{d}v < 0, \quad \partial_\rho \mathrm{Var}(v) = \partial_\rho \int_V (v-U_\text{eq}(\rho))^2 M_f(v;\rho) \mathrm{d}v < 0
	\end{equation}
	for all $\rho\in(\widetilde{\rho},1)$. Then the quantity $\mu(\rho)$ given in~\eqref{eq:diffBGK}
	is negative $\forall\,\rho\in(\widetilde{\rho},1)$.
\end{proposition}
\begin{proof}
	Conditions~\eqref{eq:conditions} are sufficient to $\int_V v^2 \partial_\rho M_f(v;\rho) \mathrm{d}v < 0$. In fact
	\begin{align*}
		\partial_\rho \int_V v^2 M_f(v;\rho) \mathrm{d}v &= \partial_\rho \mathrm{Var}(v) + \partial_\rho \left( U_\text{eq}(\rho) Q_\text{eq}(\rho) \right) < 0
	\end{align*}
	$\forall\,\rho\in(\widetilde{\rho},1)$.
\end{proof}

More specifically, for the case of two microscopic speeds and in view of the previous proposition, the model can be seen as the kinetic relaxation system of Jin and Xin~\cite{Jin95therelaxation} that cannot satisfy the sub-characteristic condition since $v>0$ and $Q_\text{eq}^\prime(\rho)<0$ in the congested regime.

\begin{remark}
	We observe that conditions~\eqref{eq:conditions} do occur in traffic flow. In particular, the first one, $Q^\prime_\text{eq}(\rho) < 0$, is always verified in the congested phase of traffic. The second condition instead requires that the variance of the microscopic speeds decreases as the density increases. This is  a reasonable assumption in traffic since we expect that the freedom in choosing a microscopic velocity reduces when traffic becomes dense.
\end{remark}

In view of the results provided by the Chapman-Enskog analysis, we state the following definition which establishes that the BGK model is unstable.

\begin{definition} \label{def:stability}
	A kinetic model is said to be stable if $\mu(\rho)\geq 0$, $\forall\,\rho\in[0,\rho_M]$, weakly-unstable if $\mu(\rho)<0$ on an interval $(\rho_1,\rho_2)$ properly contained in $[0,\rho_M]$ and unstable if $\mu(\rho)<0$ on an interval $(\rho_1,\rho_2)$ in which either $\rho_1=0$ or $\rho_2=\rho_M$.
\end{definition}

\revision{In order to further investigate the instability of the BGK model when applied to traffic flow problems, we must go beyond its hydrodynamic limit. For this reason, we consider higher order macroscopic moments, following the philosophy of Grad's moment method in gas dynamics. Therefore, the second-order macroscopic model derived from the BGK equation~\eqref{eq:BGK} is obtained without closing the continuity equation with the distribution at equilibrium and computing the evolution equation for the first moment.} Then we have
\begin{equation} \label{eq:hydrolimitBGK}
\begin{aligned}
	\partial_t \rho + \partial_x (\rho u) &= 0\\
	\partial_t (\rho u) + \partial_x \int_V v^2 f(t,x,v) \mathrm{d}v &= \frac{1}{\epsilon} \left( Q_\text{eq}(\rho) - \rho u \right).
\end{aligned}
\end{equation}
System~\eqref{eq:hydrolimitBGK} is closed and thus solvable when a closure law for the second moment of the kinetic distribution $f$ is prescribed.

\begin{enumerate}
\item If we consider the following approximation
$$
	\int_V v^2 f(t,x,v) \mathrm{d}v \approx \rho u^2 + h(\rho)
$$
with $h(\rho)$ being an increasing function of the density, then from~\eqref{eq:hydrolimitBGK} it is possible to recover the Payne and Whitham model~\cite{Payne1971}  second-order model for traffic flow:
\begin{equation} \label{eq:modelPW}
\begin{aligned}
\partial_t \rho + \partial_x (\rho u) &= 0\\
\partial_t (\rho u) + \partial_x (\rho u^2 + h(\rho)) &= \frac{1}{\epsilon} \left( Q_\text{eq}(\rho) - \rho u \right).
\end{aligned}
\end{equation}
It is well-know that model~\eqref{eq:modelPW} suffers of the following drawback (see for instance the works by Daganzo~\cite{Daganzo1995}): the characteristic velocities of~\eqref{eq:modelPW} are $\lambda_1 = u - \sqrt{h^\prime(\rho)}$ and $\lambda_2 = u + \sqrt{h^\prime(\rho)}$. Thus one characteristic speed is always larger than the speed $u$ of vehicles.
%
\item If we instead approximate
$$
\int_V v^2 f(t,x,v) \mathrm{d}v \approx \rho u^2
$$
then~\eqref{eq:hydrolimitBGK} reads
\begin{equation} \label{eq:modelPWwithoutAcc}
\begin{aligned}
\partial_t \rho + \partial_x (\rho u) &= 0\\
\partial_t u + u \partial_x u &= \frac{1}{\epsilon} \left( U_\text{eq}(\rho) - u \right)
\end{aligned}
\end{equation}
that is again the Payne and Whitham model without the anticipation term. However, the model still suffers of the following drawbacks. First, the system is not strictly hyperbolic since the eigenvalues are $\lambda_1 \equiv \lambda_2 = u$ leading to a lack of strict hyperbolicity.  The possible unbounded growth of small perturbations is  highlighted by the fact that the sub-characteristic condition for the model is $-\rho^2 U_\text{eq}^\prime(\rho)>0$ which is never satisfied.
\item It is easy to show that if we approximate around the equilibrium distribution, i.e.
$$
\int_V v^2 f(t,x,v) \mathrm{d}v \approx \int_V v^2 M_f(v;\rho) \mathrm{d}v,
$$
the relaxation system~\eqref{eq:hydrolimitBGK} has a sub-characteristic condition $\mu(\rho) < 0$ with $\mu(\rho)$ given by~\eqref{eq:diffBGK} and therefore in the whole congested regime.
\end{enumerate}

\begin{remark}
	The particle dynamics resulting from the BGK model~\eqref{eq:BGK} are described by the ODE system
	\begin{equation} \label{eq:modelBando}
	\begin{aligned}
		\dot{x}_i(t) &= v_i(t) \\
		\dot{v}_i(t) &= \frac{1}{\epsilon} \left( U_\text{eq}(\rho) - v_i \right)
	\end{aligned}
	\end{equation}
	where $i$ is the label of a vehicle. System~\eqref{eq:modelBando} represents the microscopic model by Bando~\cite{Bando1995} which is know to be stable provided $\epsilon < \frac{1}{2 |U_\text{eq}^\prime(\rho)|}$, i.e. if the relaxation is large.
\end{remark}

The last results of this paragraph \minor{show} that properties of the BGK model studied via Chapman-Enskog expansion occurs also in its second-order \revision{systems of moments} and in its particle \minor{formulations}.

\revision{Closing BGK at second order moments one provides the structure of common second-order models in the literature on traffic flow. Due to this consideration we aim to understand what is the kinetic structure representing the mesoscopic scale of suitable second-order macroscopic models.}

\minor{In} the following section we investigate properties of the second-order macroscopic model~\cite{aw2000SIAP,Zhang2002,SeiboldFlynnKasimovRosales2013}.

\subsection{The case of the Aw-Rascle and Zhang model} \label{sec:macroModels}

The prototype example of second-order macroscopic traffic models is given by the Aw and Rascle~\cite{aw2000SIAP} and Zhang~\cite{Zhang2002} (ARZ) model where the conservation law for the density $\rho(t,x)$ of vehicles at time $t \geq 0$ and position $x$ is coupled with an additional equation for the evolution of the macroscopic speed of the flow $u(t,x)$. We focus on the non-homogeneous version of the ARZ model which writes in conservation form
\begin{equation} \label{eq:modelARZ}
\begin{aligned}
	\partial_t \rho + \partial_x (q-\rho h(\rho)) &= 0\\
	\partial_t q + \partial_x \left(\frac{q^2}{\rho}-q h(\rho) \right) &= \frac{1}{\epsilon} (Q_{\text{eq}}(\rho) + \rho h(\rho) - q)
\end{aligned}
\end{equation}
where $q$ is the second conservative variable and it is defined as $q:=\rho(u+h(\rho))$. The function $h=h(\rho)$ is a strictly increasing function of the density and it is called hesitation function or traffic pressure although it is homogeneous to a velocity. The quantity $\epsilon$ is a time which rules the relaxation speed of the flux $\rho u$ to the equilibrium flux $Q_{\text{eq}}(\rho)$ which is a given function of the density, i.e. the fundamental diagram. Here $Q_\text{eq}$ is not necessarily the one given in~\eqref{eq:macroQuantitiesEq}.

The ARZ model was introduced in order to overcome the drawbacks of the second-order macroscopic model introduced by Payne and Whitham and recalled in equation~\eqref{eq:modelPW}.
The characteristic velocity of the ARZ model are
$\lambda_1 = u - \rho h^\prime(\rho) \leq u = \lambda_2$
and thus  information propagates at most with the vehicles' velocity.

In primitive variables the system reads
\begin{equation} \label{eq:modelARZnoncons}
\begin{aligned}
	\partial_t \rho + \partial_x (\rho u) &= 0\\
	\partial_t \big(u+h(\rho)\big) + u \partial_x \big(u+h(\rho)\big) &= \frac{1}{\epsilon} (U_{\text{eq}}(\rho) - u).
\end{aligned}
\end{equation}
In~\eqref{eq:modelARZnoncons} $U_{\text{eq}}(\rho) = \frac{1}{\rho} Q_{\text{eq}}(\rho)$ is the equilibrium speed. 

Both~\eqref{eq:modelARZ} and~\eqref{eq:modelARZnoncons} can be understood as relaxation systems~\cite{Jin95therelaxation}. In the  limit $\epsilon \to 0$, model~\eqref{eq:modelARZnoncons} converges towards solutions of the conservation law~\eqref{eq:modelLWR}. In contrast to second-order macroscopic models, equation~\eqref{eq:modelLWR} has only one characteristic velocity given by
$$\lambda = (\rho U_{\text{eq}}(\rho))^\prime = Q_{\text{eq}}^\prime(\rho).$$
If $\epsilon$ is small, but not vanishing, \eqref{eq:modelARZnoncons}
approaches the advection-diffusion equation~\eqref{eq:advdiff}
where the diffusion coefficient $\mu(\rho)$ is given by
\begin{equation} \label{eq:diffARZeulnoncons}
	\mu(\rho) = -\rho^2 U^\prime_{\text{eq}}(\rho) \big( U^\prime_{\text{eq}}(\rho) + h^\prime(\rho) \big).
\end{equation}
In fact, using a Chapman-Enskog expansion, the coefficient~\eqref{eq:diffARZeulnoncons} is obtained by considering a first-order expansion of the speed $u=U_{\text{eq}}(\rho) + \epsilon u_1$ around the equilibrium velocity function $U_{\text{eq}}(\rho)$
\begin{align*}
	u 
	=& U_{\text{eq}}(\rho) + \epsilon \Big( -\big( U^\prime_{\text{eq}}(\rho) + h^\prime(\rho) \big) \partial_t \rho - \big( U_{\text{eq}}(\rho) U^\prime_{\text{eq}}(\rho) + U_{\text{eq}}(\rho) h^\prime(\rho) \big) \partial_x \rho \Big) +  O(\epsilon^2) \\
	=& U_{\text{eq}}(\rho) + \epsilon \Big( \rho U^\prime_{\text{eq}}(\rho) \big( U^\prime_{\text{eq}}(\rho) + h^\prime(\rho) \big) \partial_x \rho \Big) +  O(\epsilon^2).
\end{align*}
Finally, plugging the expansion of $u$ just obtained into the first equation of~\eqref{eq:modelARZnoncons} we get the advection-diffusion equation~\eqref{eq:advdiff} with coefficient $\mu(\rho)$ given by~\eqref{eq:diffARZeulnoncons}. The condition $\mu(\rho) > 0$ provides the so-called sub-characteristic condition~\cite{Chen92hyperbolicconservation,Jin95therelaxation}. 
For the ARZ model $\mu(\rho)>0$ is satisfied if
\begin{equation} \label{eq:subchCondRed}
	0>U_{\text{eq}}^\prime(\rho) > -h^\prime(\rho).
\end{equation}
The sub-characteristic condition guarantees the stability of solutions of the ARZ model based on initial conditions where $\rho(t,x) = \overline{\rho}$ and $\overline{u} = U_\text{eq}(\overline{\rho})$. In fact, the LWR model~\eqref{eq:modelLWR} and the ARZ model~\eqref{eq:modelARZnoncons} share solutions of the type $(\overline{\rho},\overline{u})$ and in this case the sub-characteristic condition~\eqref{eq:subchCondRed} can be rewritten in terms of the characteristic velocities of the two models computed in $(\overline{\rho},\overline{u})$ as
$$
	\lambda_1(\overline{\rho},\overline{u}) = U_{\text{eq}}(\overline{\rho}) - \overline{\rho} h^\prime(\overline{\rho})< \lambda(\overline{\rho}) = U_{\text{eq}}(\overline{\rho}) + \overline{\rho} U^\prime_{\text{eq}}(\overline{\rho}) < \lambda_2(\overline{\rho},\overline{u}) = U_{\text{eq}}(\overline{\rho}).
$$
The above condition implies that the characteristic velocity of the LWR model must lie between the two characteristic velocities of the ARZ model. Whitham~\cite{Whitham1959} proved that this inequality is verified if and only if solutions of the second-order macroscopic model with initial condition $(\overline{\rho},\overline{u})$ are linearly stable, i.e. small perturbations to $(\overline{\rho},\overline{u})$ decay in time \cite{Chen92hyperbolicconservation}.

We stress the fact that condition~\eqref{eq:subchCondRed} strongly 
restricts the possible choice of $U_\text{eq}$ and $h$. For instance, the Greenshields' closure $U_\text{eq}(\rho) = 1-\rho$ and the hesitation function $h(\rho) = C \rho^\gamma$, $C,\gamma>0$ given in~\cite{aw2000SIAP} violate~\eqref{eq:subchCondRed} if $\rho^{\gamma-1} \leq \frac{1}{C\gamma}$.

\begin{figure}
	\centering
	\begin{subfigure}[t]{0.5\textwidth}
		\centering
		\includegraphics[width=\textwidth]{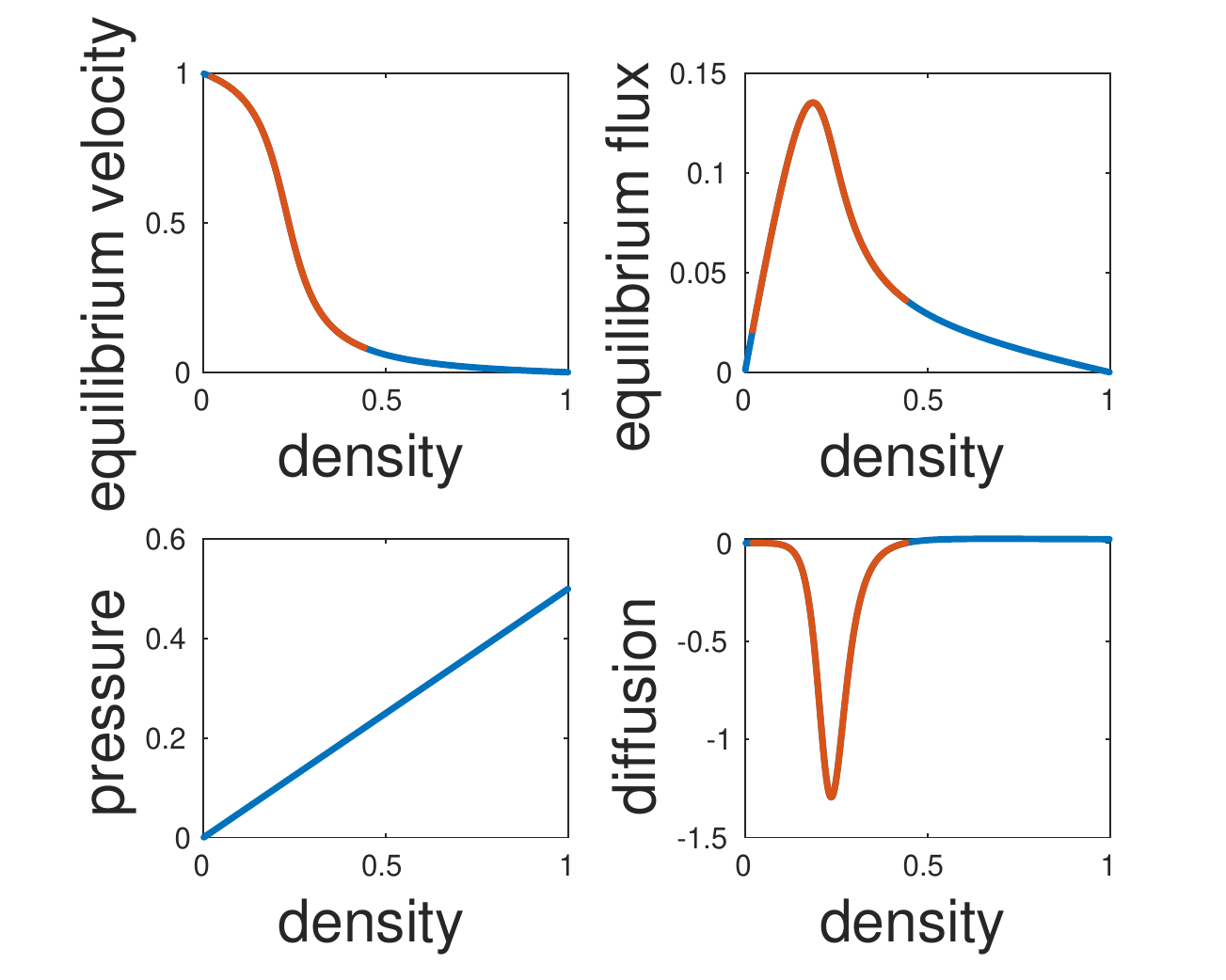}
		\caption{Closure $U_\text{eq}(\rho) = \frac{\frac{\pi}{2}+\arctan(11\frac{\rho-0.22}{\rho-1})}{\frac{\pi}{2}+\arctan(11\cdot0.22)}$ given in~\cite{aw2002SIAP} with pressure $h(\rho)=\rho/2$.}
	\end{subfigure}%
	~
	\begin{subfigure}[t]{0.5\textwidth}
		\centering
		\includegraphics[width=\textwidth]{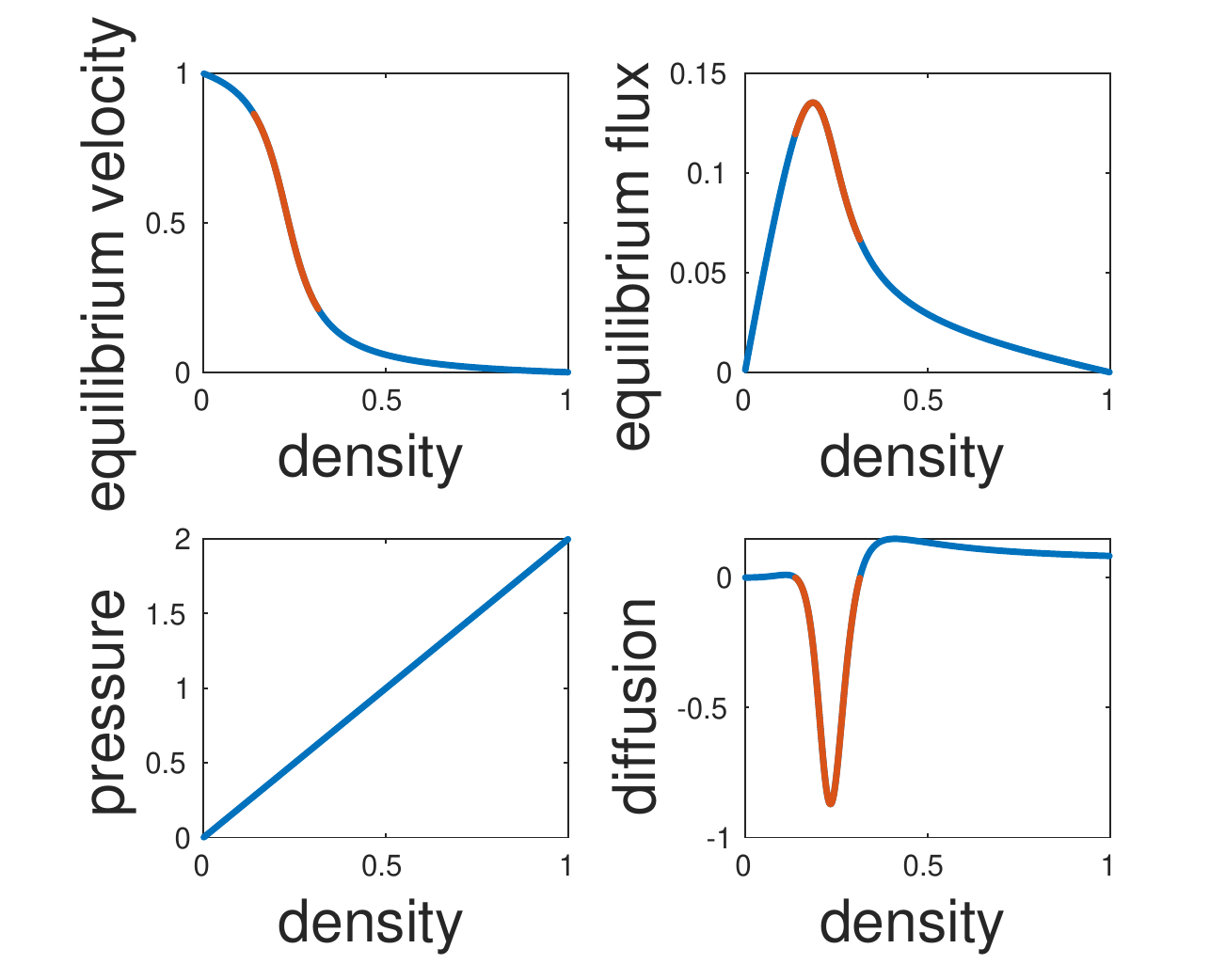}
		\caption{Closure $U_\text{eq}(\rho) = \frac{\frac{\pi}{2}+\arctan(11\frac{\rho-0.22}{\rho-1})}{\frac{\pi}{2}+\arctan(11\cdot0.22)}$ given in~\cite{aw2002SIAP} with pressure $h(\rho)=2\rho$.}
	\end{subfigure}
	\\
	\begin{subfigure}[t]{0.5\textwidth}
		\centering
		\includegraphics[width=\textwidth]{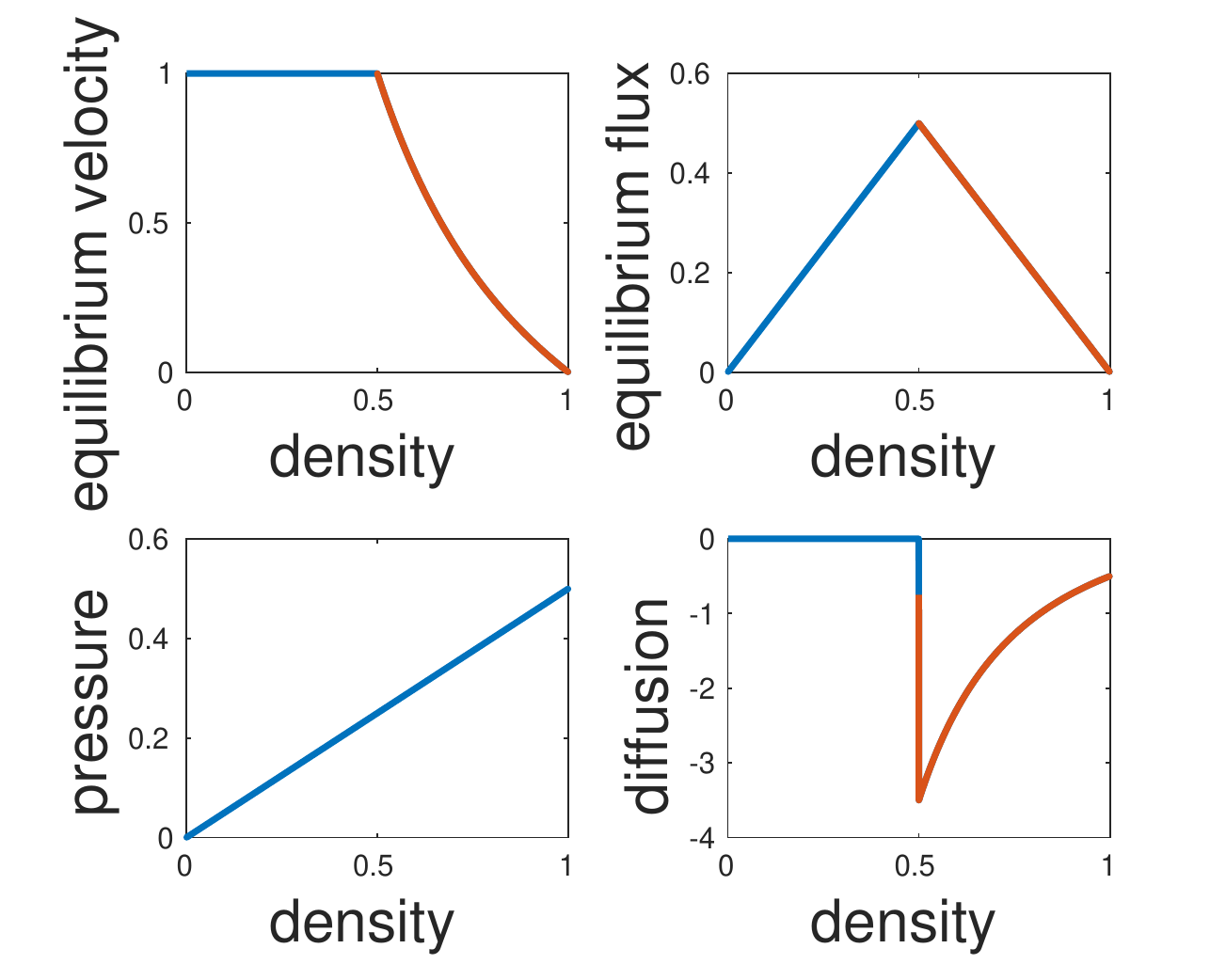}
		\caption{Closure $U_\text{eq}(\rho) = \frac{1}{\rho}\int_V vM_f(v;\rho)\mathrm{d}v$ with two speeds and pressure $h(\rho)=\rho/2$.}
	\end{subfigure}%
	~
	\begin{subfigure}[t]{0.5\textwidth}
		\centering
		\includegraphics[width=\textwidth]{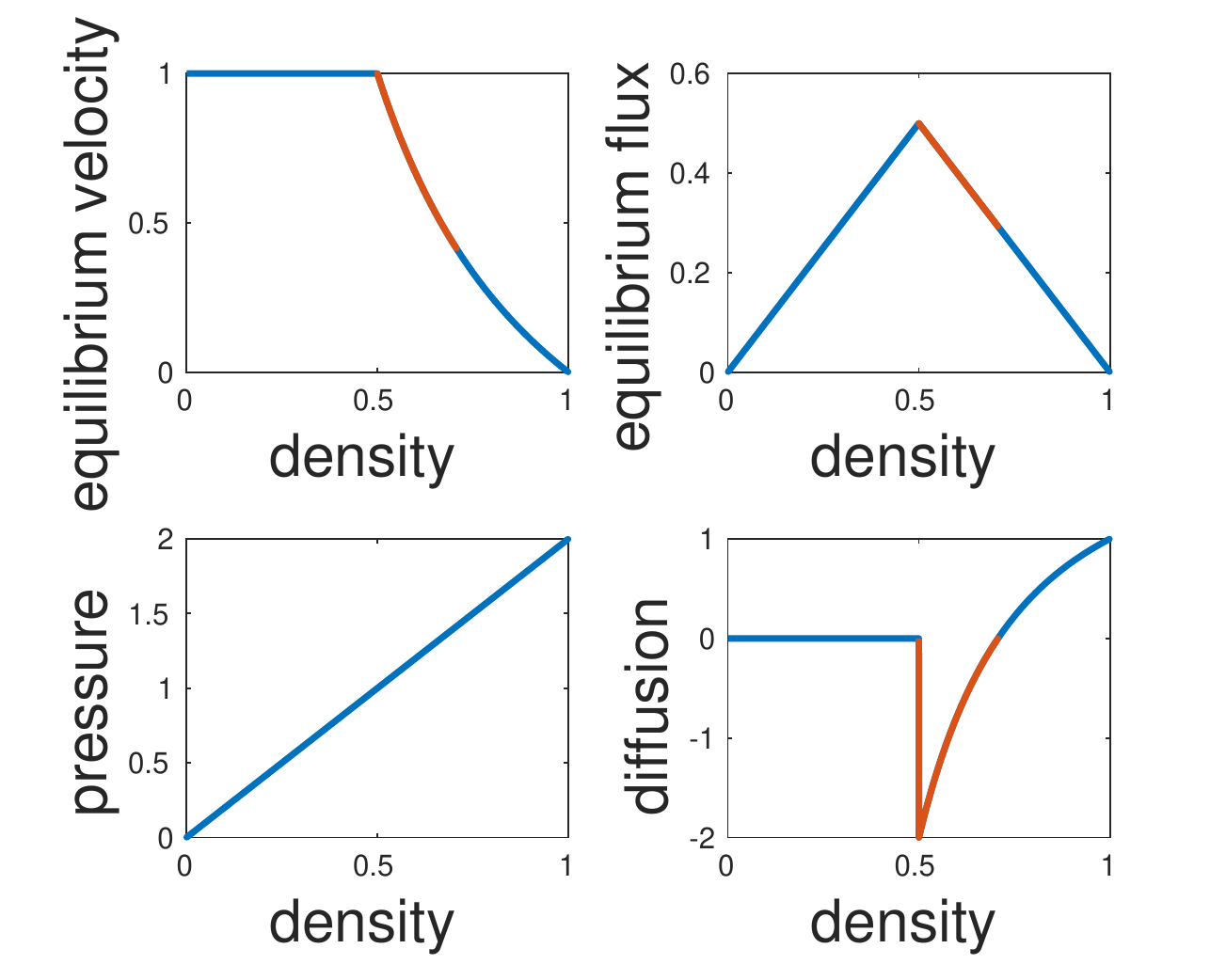}
		\caption{Closure $U_\text{eq}(\rho) = \frac{1}{\rho}\int_V vM_f(v;\rho)\mathrm{d}v$ with two speeds and pressure $h(\rho)=2\rho$.\label{fig:subchCondD}}
	\end{subfigure}
	\\
	\begin{subfigure}[t]{0.5\textwidth}
		\centering
		\includegraphics[width=\textwidth]{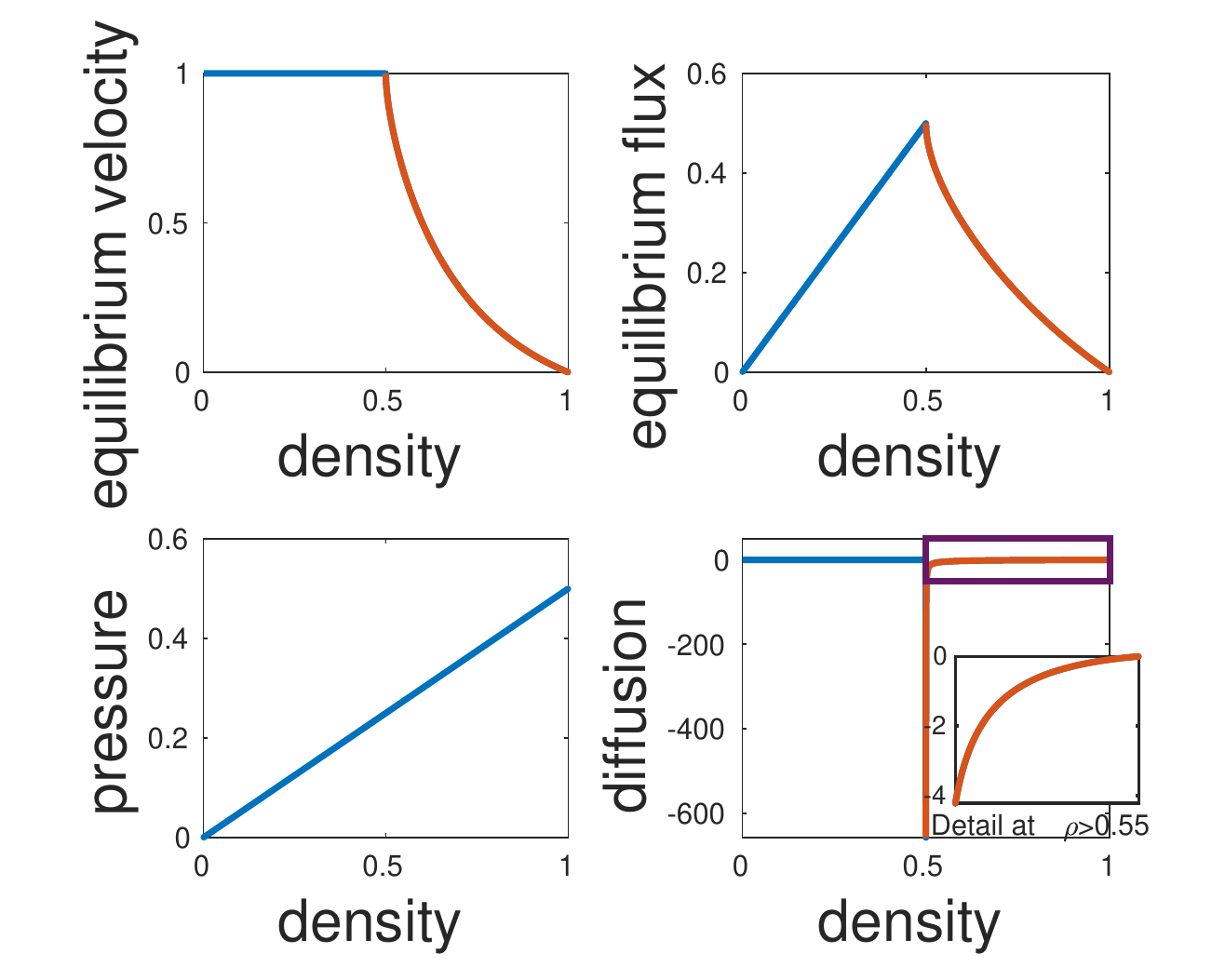}
		\caption{Closure $U_\text{eq}(\rho) = \frac{1}{\rho}\int_V vM_f(v;\rho)\mathrm{d}v$ with three speeds and pressure $h(\rho)=\rho/2$.}
	\end{subfigure}%
	~
	\begin{subfigure}[t]{0.5\textwidth}
		\centering
		\includegraphics[width=\textwidth]{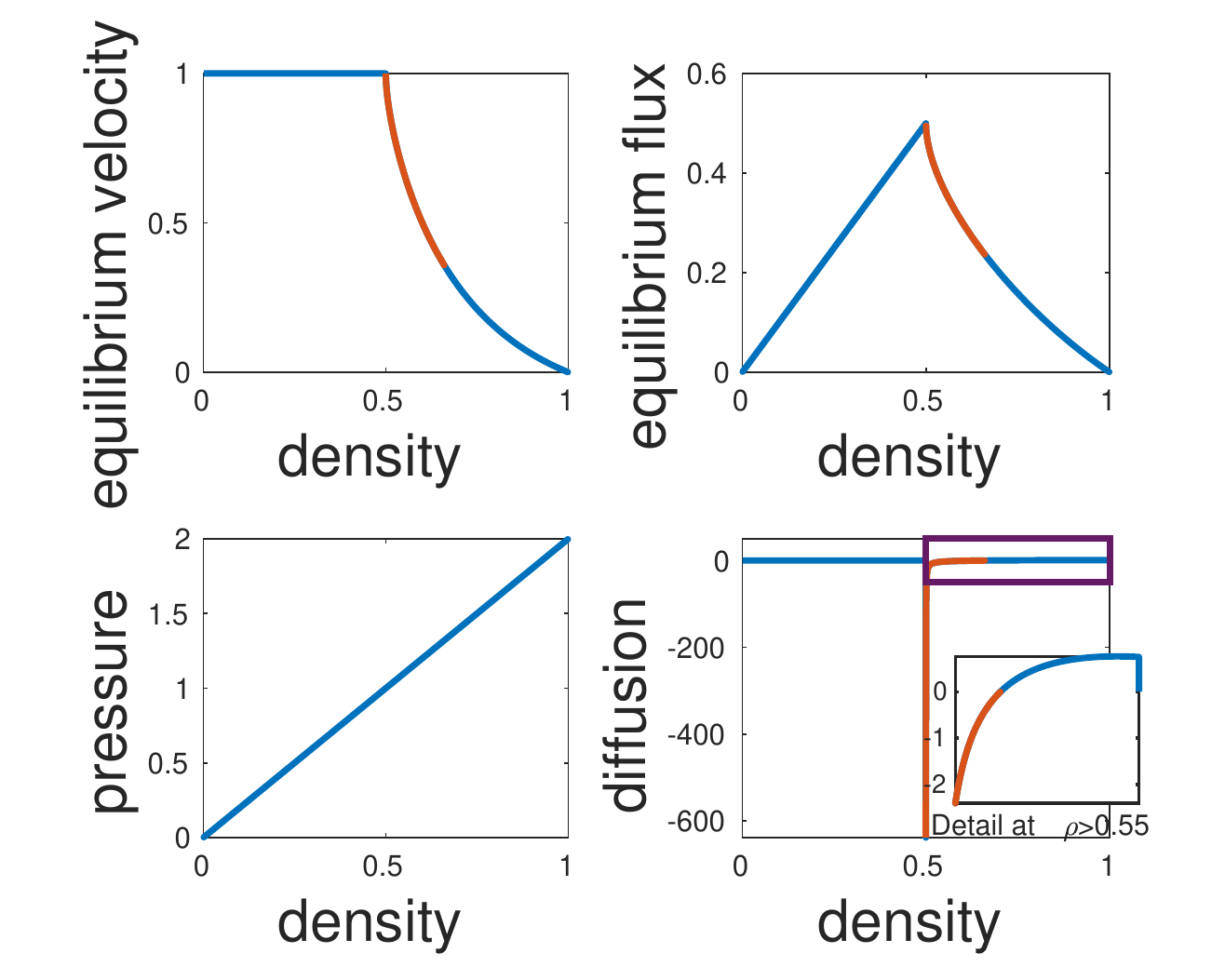}
		\caption{Closure $U_\text{eq}(\rho) = \frac{1}{\rho}\int_V vM_f(v;\rho)\mathrm{d}v$ with three speeds and pressure $h(\rho)=2\rho$.\label{fig:subchCondF}}
	\end{subfigure}
	\caption{The right bottom panels show the sign of $\mu(\rho)$ for the ARZ model~\eqref{eq:modelARZnoncons} with different closures and pressure functions. Blue lines correspond to regimes where $\mu(\rho)>0$ and red lines to regimes where $\mu(\rho)<0$.\label{fig:subchCond}}
\end{figure}

In Figure~\ref{fig:subchCond} we numerically analyze the sign of the diffusion coefficient $\mu(\rho)$ of the ARZ model given in~\eqref{eq:diffARZeulnoncons} for several choices of the equilibrium velocity and  the hesitation function. We observe that standard choices  lead to regimes where the diffusion coefficient is negative in the congested regime and consequently the sub-characteristic condition is not verified. This produces instabilities of uniform flow as discussed above. However, for suitable choices of $U_\text{eq}$ and $h$, regimes where $\mu(\rho)<0$ can be properly contained in $[0,\rho_M]$, \revision{and in this case numerical evidence shows} the existence of a bounded density regime.

This result was already mentioned and analyzed in~\cite{SeiboldFlynnKasimovRosales2013} where the authors link the appearance of such regimes to the presence of nonlinear traveling wave solutions, called jamitons, and identified as the stop and go waves observed in real traffic situations. 

\paragraph{Bounded waves.} Numerical evidence suggests that if the instability of the model is restricted on a set $(\rho_1,\rho_2)$ properly contained in $[0,\rho_M]$ \minor{the amplitude of unstable waves remains bounded}. In fact, the growth of a wave is smoothed when the solution enters in the regime $[0,\rho_M]\setminus(\rho_1,\rho_2)$ where the diffusion is positive.

We point-out that this consideration is purely heuristic and based on numerical investigations. To the best of our knowledge there are few contributions on the analysis of advection-diffusion equations with a \minor{changing} sign diffusion coefficient, see e.g.~\cite{CorliMalaguti2018,GildingTesei2010,LafitteMascia2012,MasciaTerracinaTesei2009}.

\begin{figure}[t!]
	\centering
	\includegraphics[width=0.49\textwidth]{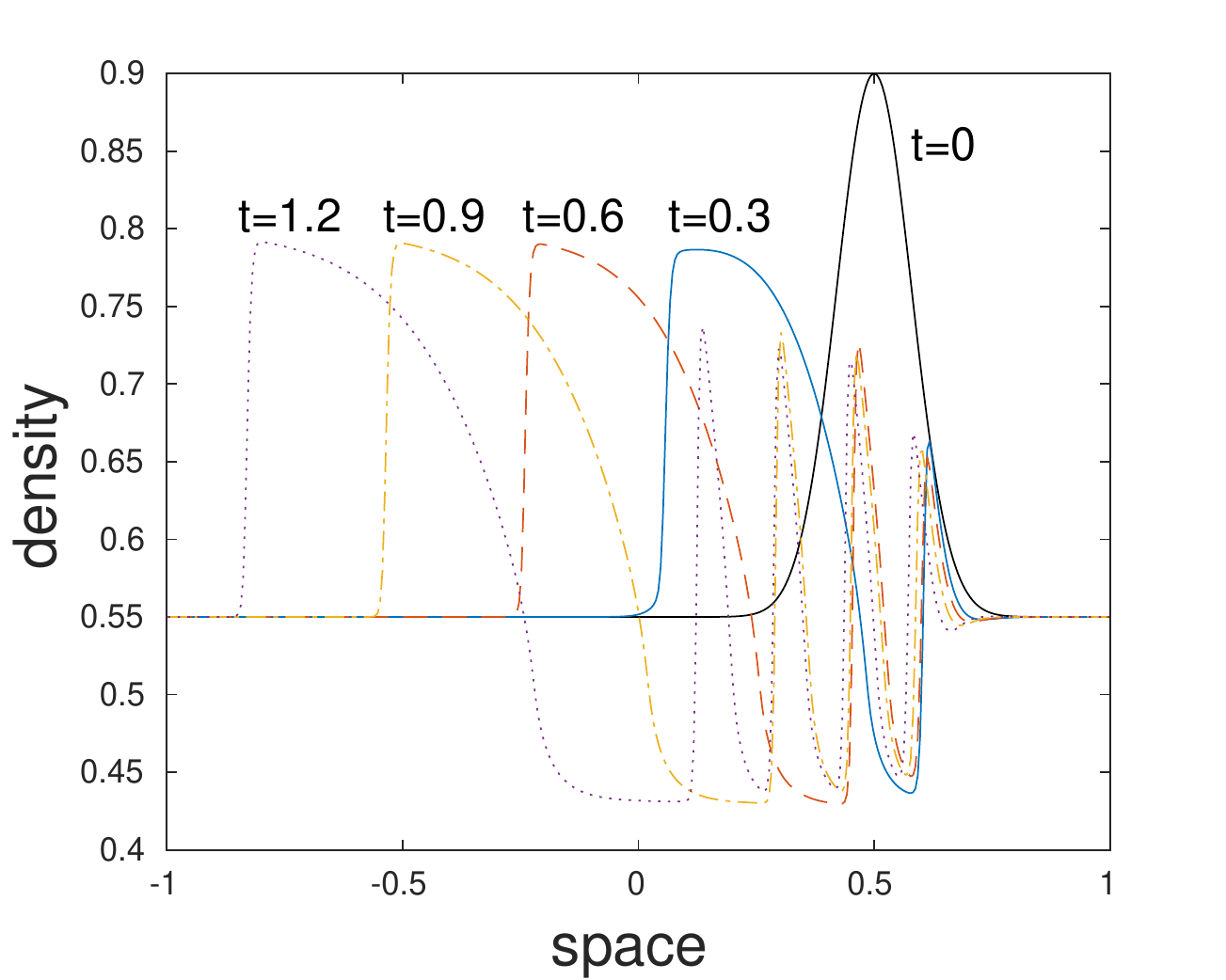}
	\includegraphics[width=0.49\textwidth]{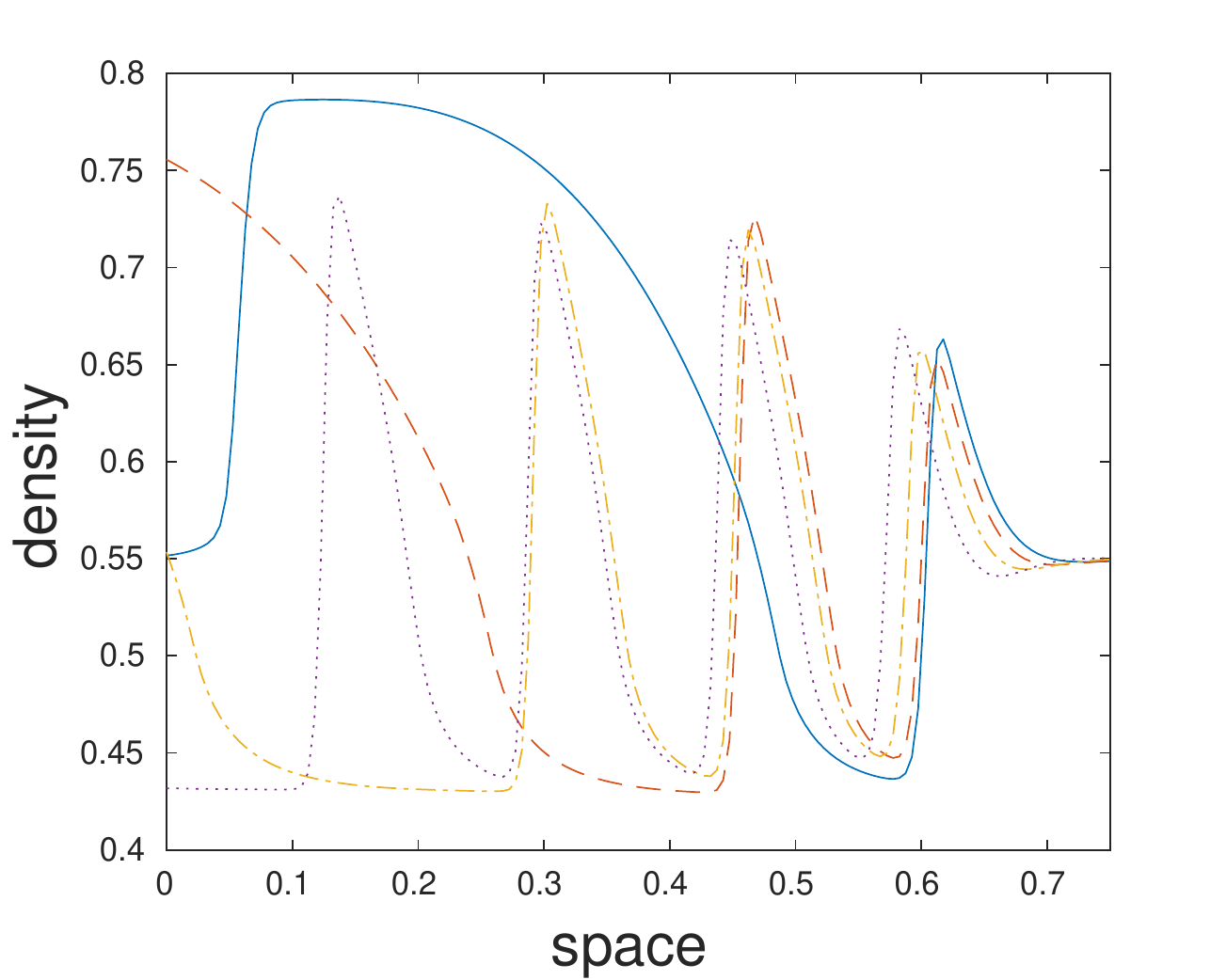}
	\\
	\includegraphics[width=0.49\textwidth]{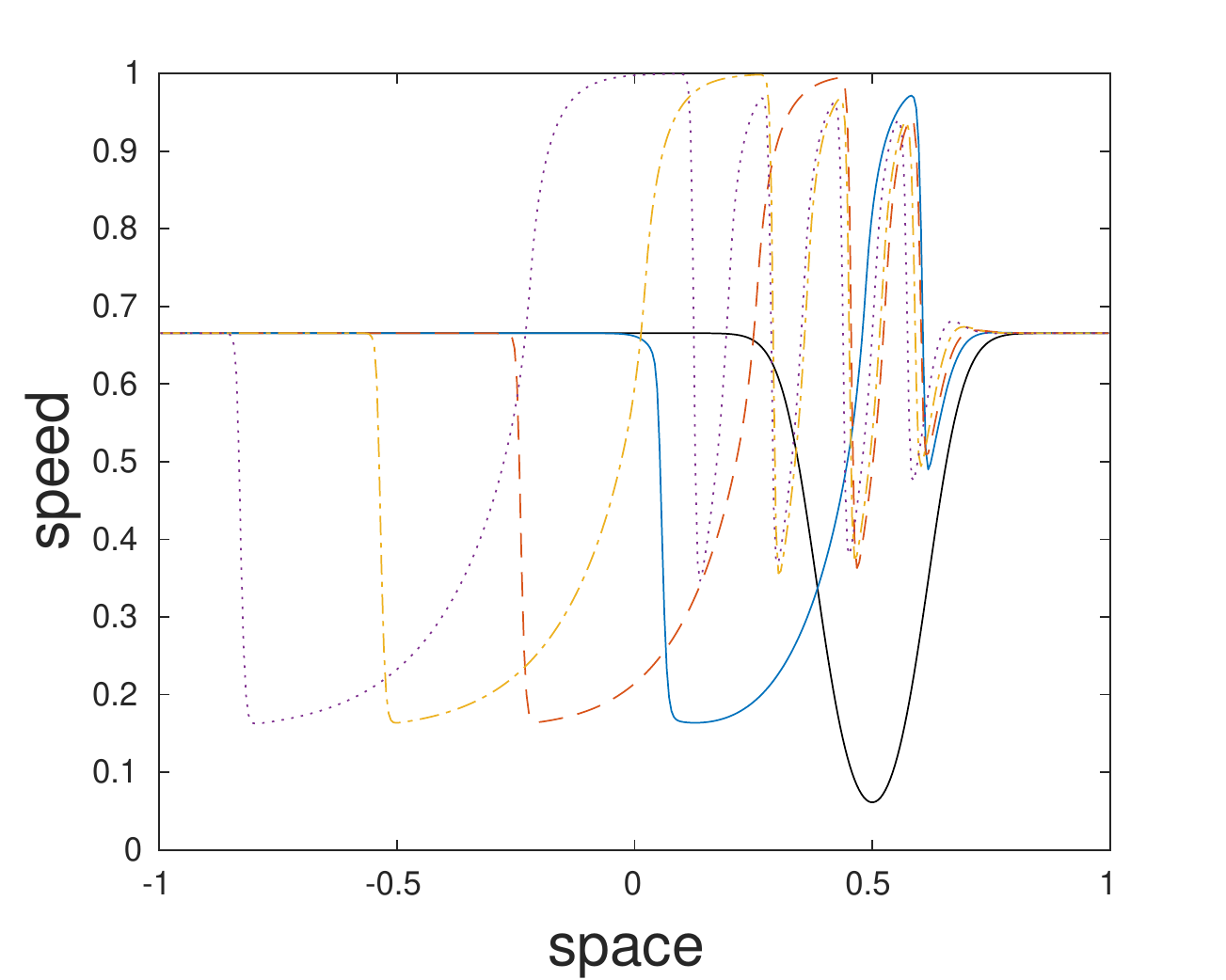}
	\includegraphics[width=0.49\textwidth]{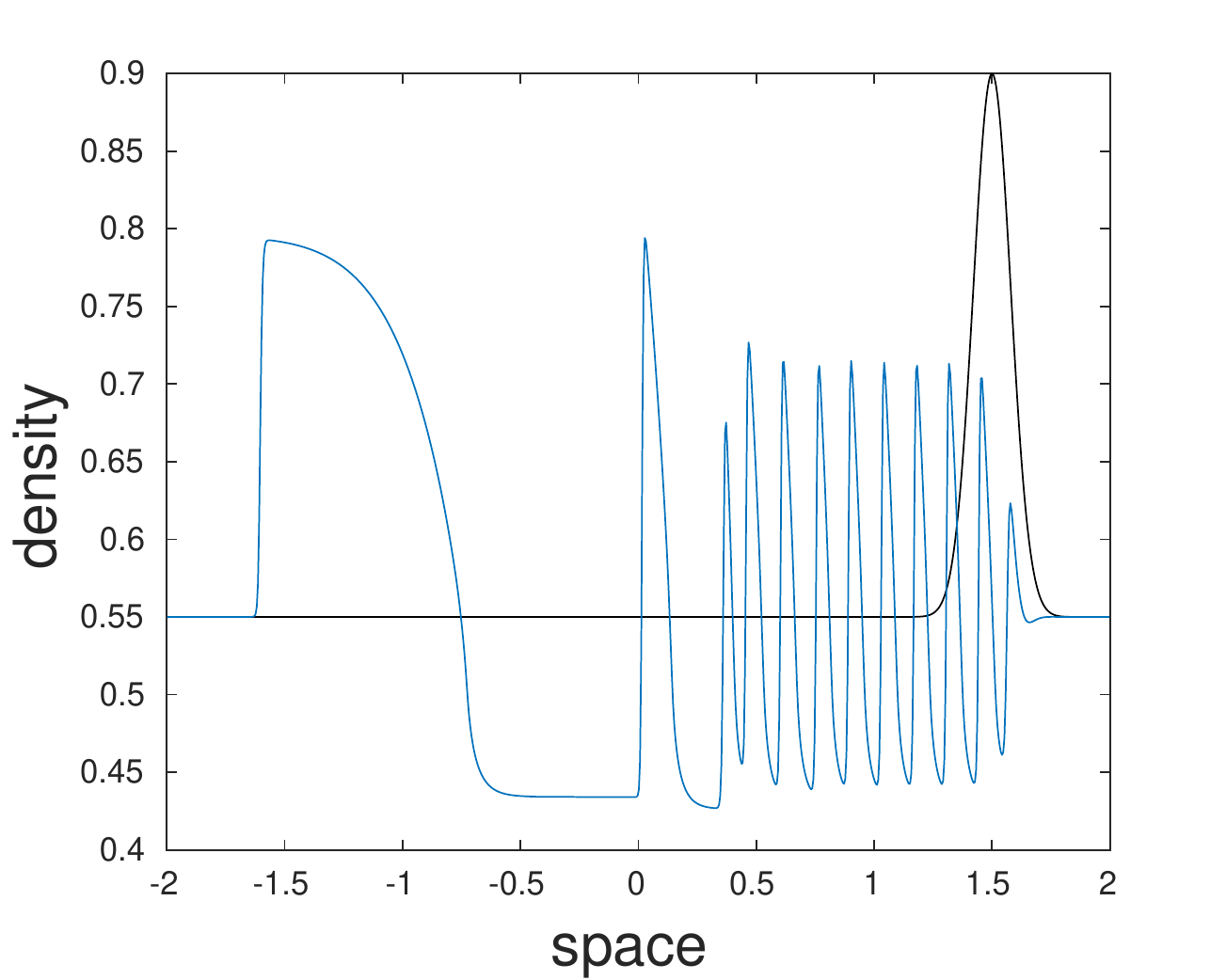}
	\caption{Evolution of a high density perturbation with the ARZ model. Top-left: density profile at different instants in time. Top-right: zoom of the density profile in the region where unstable waves occur. Bottom-left: speed profile at different instants in time. Bottom-right: density profile at final time $T_M=3$.\label{fig:boundedWaves}}
\end{figure}

In Figure~\ref{fig:boundedWaves} we numerically study the evolution of a density perturbation with the ARZ model as \minor{we already} did for the BGK model in the previous section. Here we consider the case in which the bump in the density \minor{intersects} a regime where the diffusion is negative. In particular, we consider the equilibrium flux $Q_\text{eq}$ and the equilibrium velocity $U_\text{eq}$ prescribed by the three-velocity kinetic distribution at equilibrium, see~\eqref{eq:macroQuantitiesEq}. The pressure or hesitation function is chosen as $h(\rho) = 2\rho$. This is the situation described in Figure~\ref{fig:subchCondF} where $\mu(\rho)<0$ for $\rho\in(0.5,0.66)$. The initial condition is given by~\eqref{eq:densityBump} with $\rho_{\min}=0.55$ and $\rho_{\max}=0.9$. The space domain is $[-1,1]$ and discretized with $N_x=400$ cells. The top-left panel shows the density profile at different instants in time. We observe that unstable waves appear in the regime where the diffusion coefficient is negative and they move backward. The top-right panel shows a zoom in the region where the unstable waves are rising and the bottom-left panels shows the evolution of the macroscopic speed. The growth is bounded as confirmed in the bottom-right panel where we show the density profile for a longer final time and using periodic boundary conditions.

In the following, we consider models that allow for a Chapman-Enskog expansion as advection-diffusion equation with a negative diffusion coefficient in \minor{a subset $(\rho_1,\rho_2)$} properly contained in $[0,\rho_M]$. \minor{The violation of the sub-characteristic condition lead to instabilities that in turn can be regarded as models for stop and go waves.} These waves are however unbounded if the regime of negative diffusion is not properly contained in $[0,\rho_M]$ as in the case of the classical BGK model. 

\section{Derivation of a modified BGK-type model for traffic flow} \label{sec:correctBGK}

In this section we address the issue of deriving a modified BGK-type equation for traffic flow satisfying the following properties.

\begin{property} \label{p:property1}
	A kinetic model for traffic flow \minor{should be} weakly-unstable in the sense of Definition~\ref{def:stability}.
\end{property}
\begin{property} \label{p:property3}
	The \revision{second-order system of moments} of a kinetic model for traffic flow \minor{should} recover an  ARZ-type macroscopic model.
\end{property}


The starting point for the derivation of the BGK-type model are microscopic follow-the-leader (FTL) models. The FTL model is proved to converge to the ARZ model in the macroscopic limit~\cite{aw2002SIAP}.

The kinetic model derived with this approach will also automatically guarantee Property~\ref{p:property3}. 
Moreover, since we have shown that \revision{an ARZ-type model} may have a negative diffusion coefficient in a small density regime, see for instance Figure~\ref{fig:subchCondD}-\ref{fig:subchCondF}, we expect that the BGK-type equation will also satisfy Property~\ref{p:property1}.

\revision{We will also observe that the ARZ-type second-order macroscopic model derived by the new mesoscopic structure is rather general since it keeps the details that kinetic provides, and it allows to recover the classic ARZ model as special case.}

\subsection{Microscopic follow-the-leader model} \label{sec:microFTL}

In~\cite{aw2002SIAP} a first-order semi-discretization in space of the Lagrangian formulation of the ARZ model is \minor{shown to be} equivalent to the microscopic follow-the-leader model:
\begin{equation} \label{eq:modelFTL}
\begin{aligned}
\dot{x}_i &= v_i\\
\dot{v}_i & = K(x_i,x_{i+1},v_i,v_{i+1}) +  \frac{1}{\epsilon} ( U_\text{eq}(\rho_i) - v_i ),
\end{aligned}
\end{equation}
with
\begin{equation} \label{eq:microInteractionTerm}
	K(x_i,x_{i+1},v_i,v_{i+1}) = C_\gamma \frac{v_{i+1}-v_i}{(x_{i+1}-x_i)^{\gamma+1}}
\end{equation}
and where $x_i$ and $v_i$ are the microscopic position and velocity of the vehicle $i$ at time $t\in\mathbb{R}^+$, $U_\text{eq}(\rho_i)$ is a target speed depending on the value of the local density $\rho_i = \frac{\Delta X}{x_{i+1}-x_i}$. The quantity $\Delta X$ is the typical length of a vehicle. In the following computations we assume $\Delta X$ normalized to $1$. The constants $C_\gamma>0$, $\gamma>0$ and the relaxation time $\epsilon$ are given parameters. The link between the microscopic and the macroscopic model is also studied in~\cite{DiFrancescoFagioliRosini2017,DiFrancescoEtAl2017,HertyMoutariVisconti2018}.

Following~\cite{aw2002SIAP}, we consider an alternative but equivalent formulation of the model~\eqref{eq:modelFTL}  based on the quantity $w_i := v_i + p(\rho_i)$.
The function $p = p(\rho_i)$ is the so-called traffic pressure and satisfies $p(\rho) \geq 0$, $p^\prime(\rho)>0$. In terms of the new variable $w_i$, model~\eqref{eq:modelFTL} reads 
\begin{equation} \label{eq:modelFTLw}
\begin{aligned}
\dot{x}_i &= v_i = w_i - p(\rho_i)\\
\dot{w}_i & = \frac{1}{\epsilon} ( U_\text{eq}(\rho_i) + p(\rho_i) - w_i ).
\end{aligned}
\end{equation}
provided $p(\rho)$ is chosen as
$$
	\frac{\mathrm{d}}{\mathrm{d}t} p(\rho_i) = - K(x_i,x_{i+1},v_i,v_{i+1})
$$
with $K$ given by~\eqref{eq:microInteractionTerm}. Model~\eqref{eq:modelFTLw} is identical to~\eqref{eq:modelFTL} but the introduction of $w_i$ allows us to rewrite the acceleration equation as a \minor{relaxation step.} 
\minor{Note that} the choice of the function $p(\rho)$ is determined by the modeling of $K$. Since a different interaction term leads to different pressure functions, \minor{we consider the case of a general function $K$.}



Next we introduce stochastic particle interactions in terms of the desired speed $w_i$ that is an approximation of~\eqref{eq:modelFTLw}. Then, we derive the corresponding kinetic equation for the evolution of the distribution of the desired velocity via the relaxation algorithm, see~\cite[Section 4.2.2]{PareschiToscaniBOOK}. Finally, we rewrite the kinetic equation in terms of the actual microscopic speeds of vehicles and we verify that Property~\ref{p:property1} and~\ref{p:property3} hold true for the new kinetic model for traffic.

\subsection{A BGK-type model for traffic} \label{sec:derivationBGK}

Let $g = g(t,x,w) : \mathbb{R}^+ \times \mathbb{R} \times W \to \mathbb{R}^+$ be the kinetic distribution function such that $g(t,x,w)\mathrm{d}x\mathrm{d}w$ gives the number of vehicles at time $t\in\mathbb{R}^+$ with position in $[x,x+\mathrm{d}x] \subset \mathbb{R}$ and desired speed in $[w,w+\mathrm{d}w] \subset W$. The desired speed $w$ is assumed to be the speed that drivers want to keep in ``optimal'' situations. 
We define $W:=[w_{\min},+\infty)$ the space of the microscopic desired speeds where $w_{\min}>0$ may be interpreted as the minimum speed limit in free-flow conditions. 

The macroscopic density, i.e. the number of vehicles per unit length, at time $t$ and position $x$ is defined by
\begin{equation} \label{eq:densityg}
\rho(t,x) := \int_W g(t,x,w) \mathrm{d}w,
\end{equation}
and we define the macroscopic quantity
\begin{equation} \label{eq:defq}
q(t,x) := \int_W w \, g(t,x,w) \mathrm{d}w
\end{equation}

In order to derive the evolution equation for the kinetic distribution $g=g(t,x,w)$, we first prescribe the update of the microscopic states in a probabilistic interpretation that has a close link with the microscopic particle model~\eqref{eq:modelFTLw}. Let us assume a set of $N$ particles $(x_i^0,w_i^0)$ is given at time $t_0$ characterized by the position $x_i^0$ and the desired speed $w_i^0$. We can think of $(x_i^0,w_i^0)$ as samples from the distribution $g(t=0,x,w)$. The update of the microscopic states $(x_i,w_i)$ at time $t_0+\Delta t$ is generated by the following algorithm.

\begin{description}
	\item[Step 1.] Reconstruct the density $\rho_i = \frac{1}{x_{i+1}-x_i}$ at time $t$
	\item[Step 2.] Compute $x_i = x_i^0 + (w_i^0 - p(\rho_i)) \Delta t$
	\item[Step 3.] For each $i$, with probability $\frac{\Delta t}{\epsilon}$ replace $w_i^0$ with a velocity sample from a distribution $M_g$ with density $\rho_i$; otherwise set $w_i = w_i^0$.
\end{description}

So far, the distribution $M_g = M_g(w;\rho)$ only has to fulfill the requirement
$$
	\int_W M_g(w;\rho) \mathrm{d}w = \rho(t,x).
$$
However, we additionally require that
\begin{equation} \label{eq:MgRelation}
	\minor{\frac{1}{\rho(t,x)} \int_W w \, M_g(w;\rho) \mathrm{d}w = U_\text{eq}(\rho) + p(\rho).}
\end{equation}
In fact, we observe that under this assumption the probabilistic microscopic algorithm prescribes an average speed which is equal to the output velocity prescribed by a first-order time discretization of~\eqref{eq:modelFTLw}. 

According to~\cite[Section 4.2.2]{PareschiToscaniBOOK}, we  interpret the interactions introduced above as a stochastic particle method solving

\begin{equation} \label{eq:kineticW}
	\partial_t g(t,x,w) + \partial_x \big[ (w-p(\rho)) g(t,x,w) \big] = \frac{1}{\epsilon} \left( M_g(w;\rho) - g(t,x,w) \right)
\end{equation}
This equation  is still a BGK-type equation since the collision kernel is linear and describes the relaxation of $g$ towards a given distribution $M_g$ parameterized by the density $\rho$. 

\revision{It is important to point-out that, compared to classic kinetic theory, our approach is different in the sense that $M_g$ is an ``equilibrium distribution'' with a modified idea of microscopic velocity. Thanks to~\eqref{eq:MgRelation}, we impose a-priori $M_g$ but it is still based on the knowledge of the classical Maxwellian $M_f$, which is related to the classical concept of microscopic velocity, by means of $U_\text{eq}(\rho) := \frac{1}{\rho} \int v M_f \mathrm{d}v$. We do not impose $M_f$ a-priori (and so $U_\text{eq}(\rho)$ and consequently $M_g$), but we use the equilibrium distribution $M_f$ that comes out from the modeling of microscopic interactions of the spatially homogeneous kinetic model in Section~\ref{sec:homogeneous}. Therefore, the added details provided by a kinetic approach are still present.}
	
\revision{Moreover, we observe that Maxwellian $M_f$ of other kinetic models can be employed, allowing to define $M_g$ and the BGK model~\eqref{eq:kineticW} with the modified concept of microscopic speeds. Before concluding this section, we further investigate relation between $M_f$ and $M_g$.}

\minor{Let} $f=f(t,x,v)$ be the  kinetic distribution on \minor{the microscopic speed $v$ of vehicles}, such that $f(t,x,v)\mathrm{d}x\mathrm{d}v$ gives the number of vehicles at time $t$ with position in $[x,x+\mathrm{d}x] \subset \mathbb{R}$ and speed in $[v,v+\mathrm{d}v] \subset V:=[0,+\infty)$. From~\eqref{eq:kineticW} it is possible to derive the equation for $f$ via a nonlinear change of variables. 
\begin{equation} \label{eq:v}
	v = v(\rho) = w - p(\rho).
\end{equation}
The deviation is controlled by $p(\rho)$ in such a way that $p(0)=0$ so that $v=w$, i.e. drivers can behave as they wish in the vacuum regime, and $p(\rho_M) \leq w_{\min}$ in order to ensure the bound $v \geq 0$. 
In terms of $f=f(t,x,v)$ we define the density as
\begin{equation} \label{eq:densityf}
	\rho(t,x) := \int_V f(t,x,v) \mathrm{d}v, 
\end{equation}
and the flux of vehicles as
\begin{equation} \label{eq:flux}
	(\rho u)(t,x) := \int_V v \, f(t,x,v) \mathrm{d}v. 
\end{equation}
Further,  
$
	M_g(w;\rho) = M_g(v+p(\rho);\rho), \ w\in W, \ v\in V
$
and $M_f=M_f(v;\rho)$, \revision{and setting 
\begin{equation} \label{eq:equilibriumf}
	M_g(v+p(\rho);\rho) := M_f(v;\rho), \quad v\in V
\end{equation}
relation~\eqref{eq:MgRelation} is satisfied if $U_\text{eq}(\rho)$ is the equilibrium velocity provided by the Maxwellian $M_f$ of a kinetic model for traffic. In fact,} for each fixed density $\rho$, we have that
$$
	\rho = \int_W M_g(w;\rho) \mathrm{d}w = \int_V M_g(v+p(\rho);\rho) \mathrm{d}v = \int_V M_f(v;\rho) \mathrm{d}v
$$
and
$$
	\rho U_\text{eq}(\rho) + \rho p(\rho) = \int_W w M_g(w;\rho) \mathrm{d}w = \int_V (v+p(\rho)) M_g(v+p(\rho);\rho) \mathrm{d}v = \int_V v M_f(v;\rho) \mathrm{d}v + \rho p(\rho).
$$
Hence, $
	\int_V v M_f(v;\rho) \mathrm{d}v = \rho U_\text{eq}(\rho).
$

\subsection{Properties of the BGK-type model} \label{sec:properties}

We verify that Property~\ref{p:property1} and Property~\ref{p:property3} hold true for the kinetic equation~\eqref{eq:kineticW}.
\paragraph{Property~\ref{p:property3} - Macroscopic limits.}

\revision{We show that Property~\ref{p:property3} holds by computing the system of moment equations of~\eqref{eq:kineticW}, up to second order.}

The continuity equation is obtained by integration in $w$ space as
$$
\partial_t \rho(t,x) + \partial_x \left( q(t,x) - \rho p(\rho) \right) = 0
$$
where $q$ is defined by~\eqref{eq:defq}. If one assumes that $g$ is at equilibrium , $g=M_g$, we obtain the classical first-order model for traffic flow
\begin{equation} \label{eq:firstOrdLimit}
	\partial_t \rho(t,x) + \partial_x \left(\rho U_\text{eq}(\rho)\right) = 0.
\end{equation}
Instead, the second-order moment equation gives the system
\begin{equation} \label{eq:secondOrdLimit}
	\begin{aligned}
		\partial_t \rho(t,x) + \partial_x \left( q(t,x) - \rho p(\rho) \right) &= 0\\
		\partial_t q(t,x) + \partial_x \left( \int_W w^2 \, g(t,x,w) \mathrm{d}w - p(\rho) q(t,x) \right) &= \frac{1}{\epsilon} \left( \rho U_\text{eq}(\rho) + \rho p(\rho) - q(t,x) \right).
	\end{aligned}
\end{equation}
System~\eqref{eq:secondOrdLimit} is not closed since the knowledge of the second moment of $g$ is required. In particular, the ARZ model~\eqref{eq:modelARZ} is recovered by closing as 
\begin{equation} \label{eq:closureARZ}
	\int_W w^2 \, g(t,x,w) \mathrm{d}w \approx \frac{q(t,x)^2}{\rho(t,x)}
\end{equation}
and taking the function $p(\rho)=h(\rho)$. 

\paragraph{Property~\ref{p:property1} - Chapman-Enskog expansion.} In order to verify that Property~\ref{p:property1} holds true for the kinetic equation~\eqref{eq:kineticW} we use Chapman-Enskog expansion approximating
$$
	g(t,x,w) = M_g(w;\rho) + \epsilon g_1(t,x,w), \quad \text{with} \ \int_W g_1(t,x,w) \mathrm{d}w = 0
$$
and defining $Q_\text{eq}(\rho) = \rho U_\text{eq}(\rho)$.
Following similar computations given in Section~\ref{sec:motivation} for the classical BGK equation,  we obtain
\begin{align*}
	\partial_t \rho(t,x) + \partial_x Q_\text{eq}(\rho) =& - \epsilon \partial_x \int_W (w-p(\rho)) g_1(t,x,w) \mathrm{d}w\\
	=& \epsilon \partial_x \int_W (w-p(\rho)) \big[ \partial_t M_g(w;\rho) + \partial_x (w-p(\rho)) M_g(w;\rho) \big] \mathrm{d}w\\
	&= \epsilon \partial_x \Big[ \partial_t \left( Q_\text{eq}(\rho) + \rho p(\rho) \right) -p(\rho) \partial_t \rho + \partial_x \int_W w^2 M_g(w;\rho) \mathrm{d}w \Big.\\
		&\Big. -\partial_x p(\rho) \left( Q_\text{eq}(\rho) + \rho p(\rho) \right) - p(\rho) \partial_x \left( Q_\text{eq}(\rho) + \rho p(\rho) \right) + p(\rho) \partial_x \left(\rho p(\rho) \right) \Big]\\
	&=\epsilon \partial_x \Big[ Q^\prime_\text{eq}(\rho) \partial_t \rho + \rho p^\prime(\rho) \partial_t \rho + \partial_x \int_W w^2 M_g(w;\rho) \mathrm{d}w \Big.\\
		&\Big. -2 p(\rho) Q^\prime_\text{eq}(\rho) \partial_x \rho - Q_\text{eq}(\rho) p^\prime(\rho) \partial_x \rho - 2 \rho p(\rho) p^\prime(\rho) \partial_x \rho - p(\rho)^2 \partial_x \rho \Big]\\
	&=\epsilon \partial_x \Big[ -Q^\prime_\text{eq}(\rho)^2 \partial_x \rho - \rho p^\prime(\rho) Q^\prime_\text{eq}(\rho) \partial_x \rho + \partial_x \int_W w^2 M_g(w;\rho) \mathrm{d}w \Big.\\
		&\Big. -2 p(\rho) Q^\prime_\text{eq}(\rho) \partial_x \rho - Q_\text{eq}(\rho) p^\prime(\rho) \partial_x \rho - 2 \rho p(\rho) p^\prime(\rho) \partial_x \rho - p(\rho)^2 \partial_x \rho \Big].
\end{align*}
Note that, using the definition~\eqref{eq:equilibriumf} of the equilibrium distribution for the velocity $v$, we  compute
\begin{align*}
	\partial_x \int_W w^2 M_g(w;\rho) \mathrm{d}w =& \partial_x \int_V w^2 M_g(v+p(\rho);\rho) \mathrm{d}v \\
	=&\partial_x \int_V (v+p(\rho))^2 M_f(v;\rho) \mathrm{d}v\\
	=&\Big[ \int_V v^2 \partial_\rho M_f(v;\rho) \mathrm{d}v + p(\rho)^2 + 2 \rho p(\rho) p^\prime(\rho) \Big.\\
	& + 2 p(\rho) Q_\text{eq}^\prime(\rho) + 2 Q_\text{eq}(\rho) p^\prime(\rho) \Big] \partial_x \rho
\end{align*}
Thus, the kinetic equation~\eqref{eq:kineticW} solves the advection-diffusion equation
$$
\partial_x \rho(t,x) + \partial_x Q_\text{eq}(\rho(t,x)) = \epsilon \partial_x \Big[ \mu(\rho) \partial_x \rho(t,x) \Big]
$$
with
\begin{equation} \label{eq:kineticDiffCoeff}
\mu(\rho) = -Q^\prime_\text{eq}(\rho)^2 + \int_V v^2 \partial_\rho M_f(v;\rho) \mathrm{d}v - \rho p^\prime(\rho) Q^\prime_\text{eq}(\rho) + Q_\text{eq}(\rho) p^\prime(\rho).
\end{equation}
It is easy to verify that~\eqref{eq:kineticDiffCoeff} is exactly~\eqref{eq:diffARZeulnoncons} if the closure~\eqref{eq:closureARZ} is considered. Observe that, compared to~\eqref{eq:diffBGK}, the diffusion coefficient~\eqref{eq:kineticDiffCoeff} contains two additional terms which depend on the function $p(\rho)$. Therefore, it is possible, for a given distribution $M_f$, to find a suitable $p(\rho)$ such that $\mu(\rho) > 0$ also in the congested regime. Recall that $\mu(\rho)$ given in~\eqref{eq:diffBGK} was unconditionally negative in the congested phase of traffic for the classical BGK model~\eqref{eq:BGK}. In particular, it is possible to find $p(\rho)$ in order to guarantee that Property~\ref{p:property1} holds for the equation~\eqref{eq:kineticW} and thus have a small regime where instabilities can occur but they are bounded.

\begin{figure}[t!]
	\centering
	\begin{subfigure}[t]{0.5\textwidth}
		\centering
		\includegraphics[width=\textwidth]{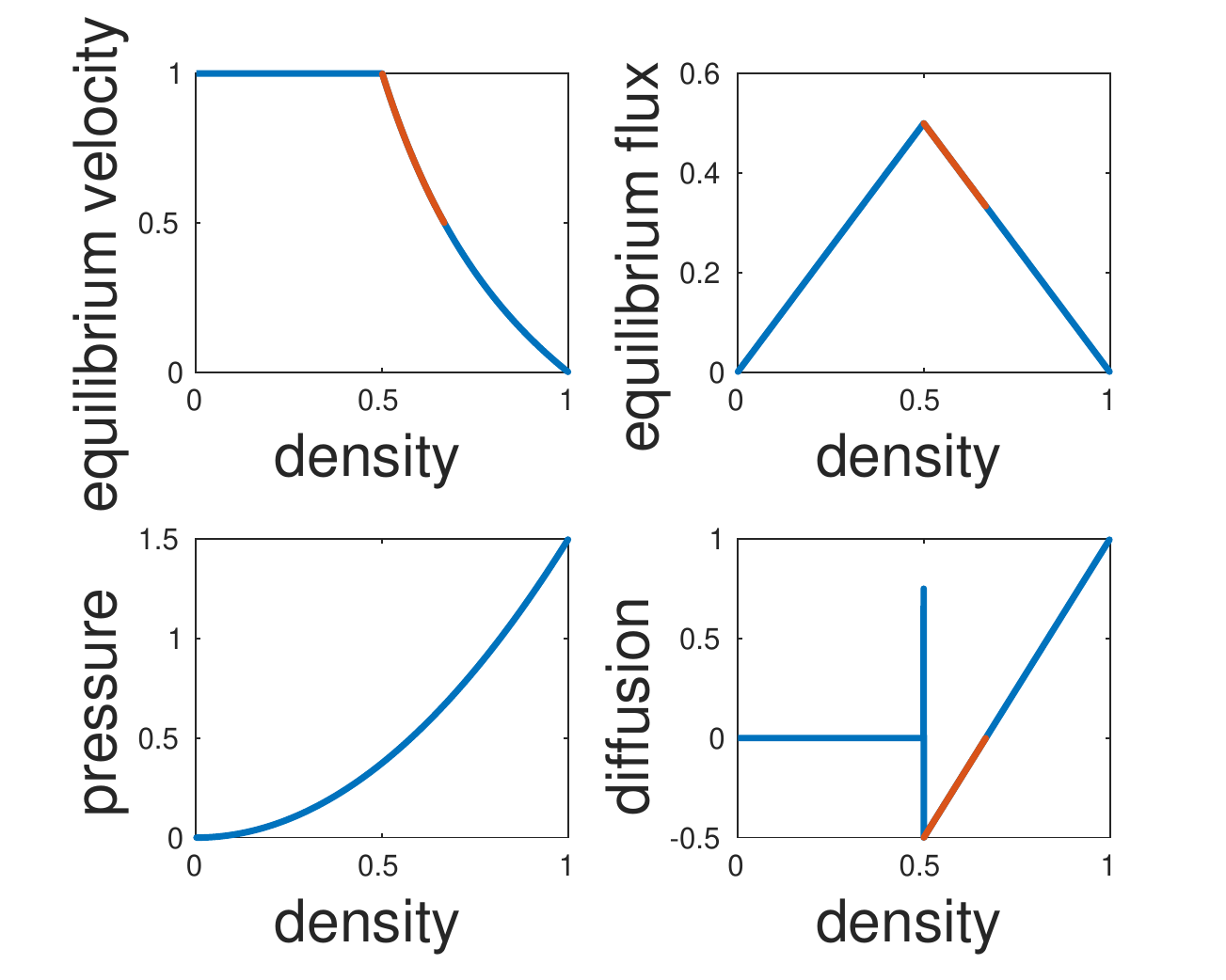}
		\caption{Closure $U_\text{eq}(\rho) = \frac{1}{\rho} \int_V v M_f(v;\rho)\mathrm{d}v$ with two speeds and pressure $p(\rho)=\frac32\rho^2$.\label{fig:BGK2speeds1}}
	\end{subfigure}%
	~
	\begin{subfigure}[t]{0.5\textwidth}
		\centering
		\includegraphics[width=\textwidth]{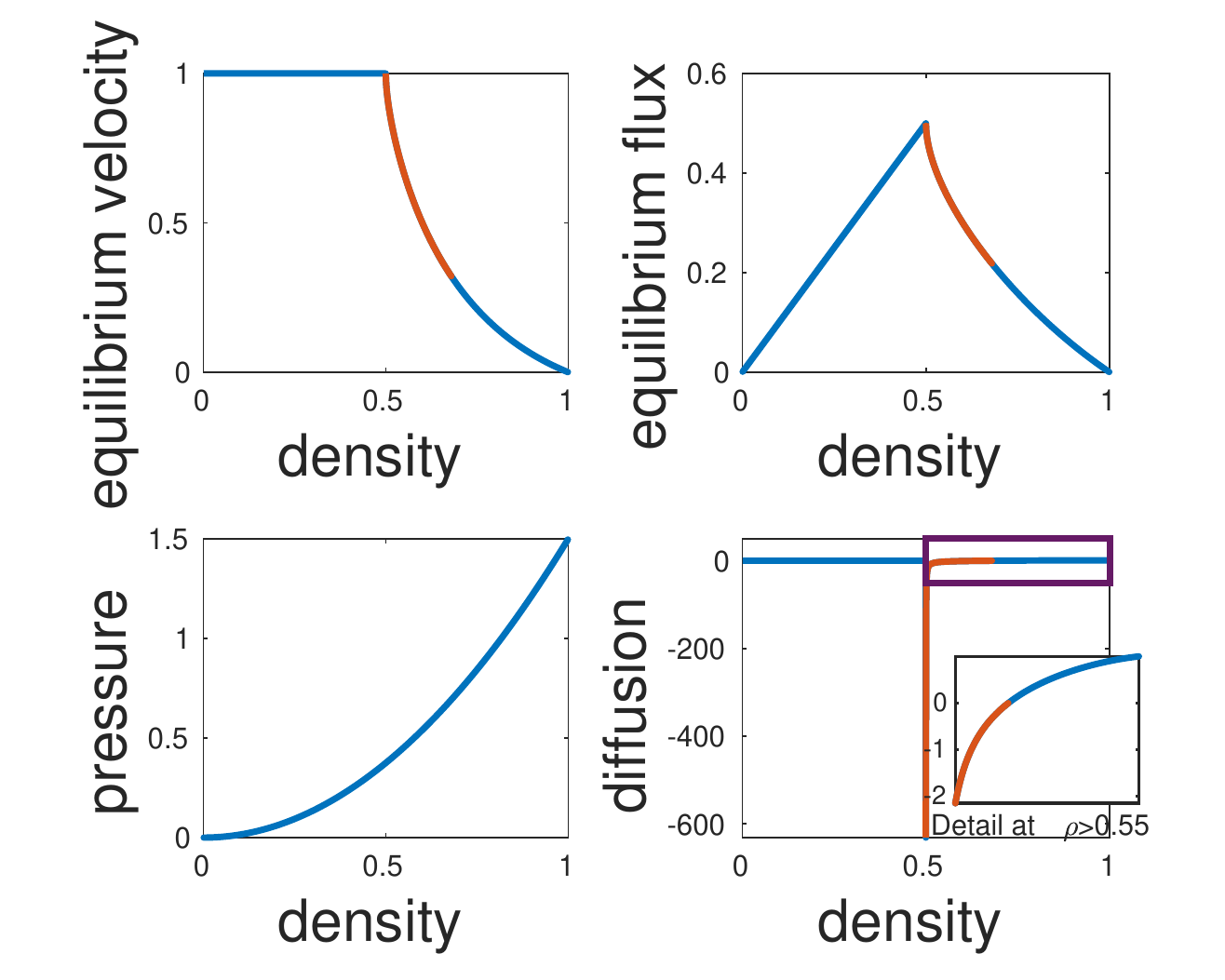}
		\caption{Closure $U_\text{eq}(\rho) = \frac{1}{\rho} \int_V v M_f(v;\rho)\mathrm{d}v$ with three speeds and pressure $p(\rho)=\frac32\rho^2$.\label{fig:BGK3speeds1}}
	\end{subfigure}
	\\
	\begin{subfigure}[t]{0.5\textwidth}
		\centering
		\includegraphics[width=\textwidth]{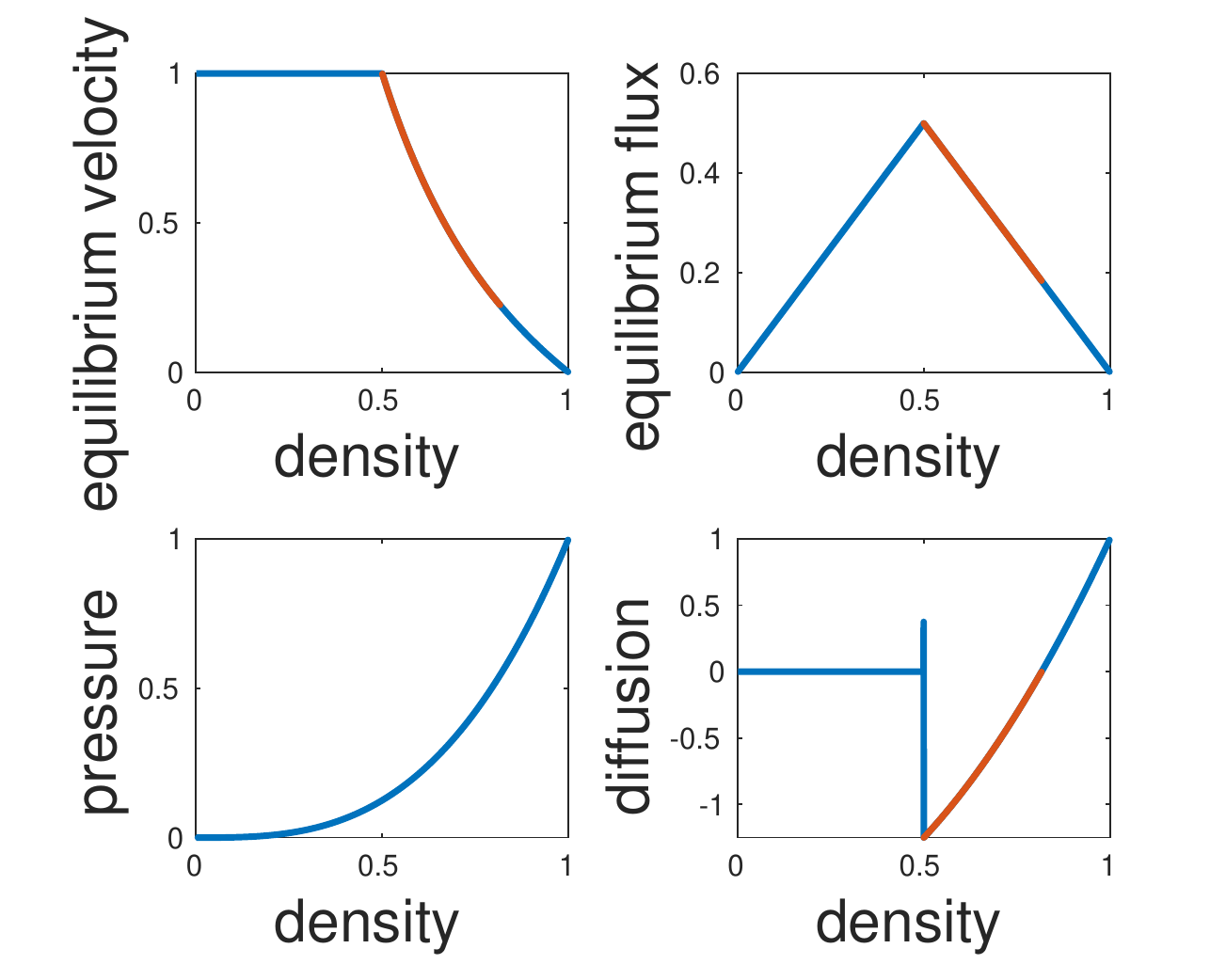}
		\caption{Closure $U_\text{eq}(\rho) = \frac{1}{\rho} \int_V v M_f(v;\rho)\mathrm{d}v$ with two speeds and pressure $p(\rho)=\rho^3$.\label{fig:BGK2speeds2}}
	\end{subfigure}%
	~
	\begin{subfigure}[t]{0.5\textwidth}
		\centering
		\includegraphics[width=\textwidth]{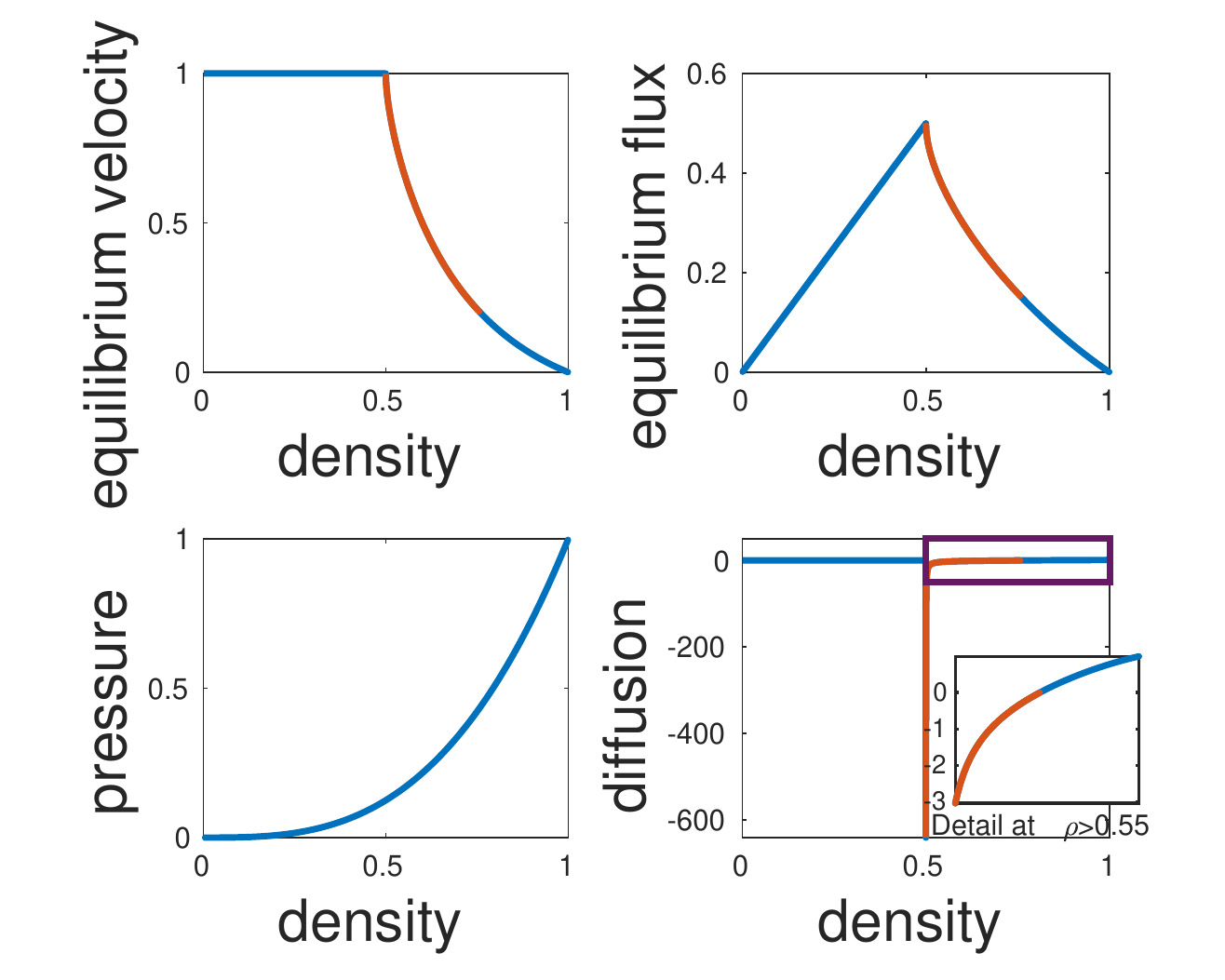}
		\caption{Closure $U_\text{eq}(\rho) = \frac{1}{\rho} \int_V v M_f(v;\rho)\mathrm{d}v$ with three speeds and pressure $p(\rho)=\rho^3$.\label{fig:BGK3speeds2}}
	\end{subfigure}
	\caption{The right bottom panels show the sign of the diffusion coefficient~\eqref{eq:kineticDiffCoeff} for the BGK-type model~\eqref{eq:kineticW} with the closure provided by the homogeneous model in Section~\ref{sec:homogeneous} and different pressure functions. Blue lines correspond to regimes where the coefficient is positive and red lines to regimes where it is negative.\label{fig:diffusionNewBGK}}
\end{figure}

Setting $Q_\text{eq}(\rho) = \rho U_\text{eq}(\rho)$ the second two terms of the diffusion coefficient~\eqref{eq:kineticDiffCoeff} can be written in terms of the equilibrium speed function as
$$
	\mu(\rho) = -Q^\prime_\text{eq}(\rho)^2 + \int_V v^2 \partial_\rho M_f(v;\rho) \mathrm{d}v - \rho^2 p^\prime(\rho) U_\text{eq}^\prime(\rho).
$$
Therefore, $\mu(\rho)=\mu_\text{BGK}(\rho)+C(\rho)$ where $C(\rho) = - \rho^2 p^\prime(\rho) U_\text{eq}^\prime(\rho)\geq 0$ since $p$ and $U_\text{eq}$ are an increasing and a non-increasing function of the density, respectively. This means that, for enough large $C(\rho)$, the additional term helps to make the diffusion coefficient~\eqref{eq:diffBGK} positive. In Figure~\ref{fig:diffusionNewBGK} we numerically show this result for the case of the homogeneous kinetic model reviewed in Section~\ref{sec:homogeneous} with $2$ and $3$ microscopic speeds and two different pressure functions, i.e. $p(\rho) = \frac32 \rho^2$ and $p(\rho) = \rho^3$. Observe that, in particular, there is a small density regime where the diffusion is negative and thus equation~\eqref{eq:kineticW} unstable in the viscous limit. Compare Figure~\ref{fig:diffusionNewBGK} with Figure~\ref{fig:diffBGK}. The given formulations of the pressure functions provide exactly the interaction term~\eqref{eq:microInteractionTerm} in the microscopic follow-the-leader model for suitable choices of the parameters $C_\gamma$ and $\gamma$.

%

\section{Conclusions} \label{sec:conclusions}

In this paper we  studied the stability of the classical BGK formulation of a kinetic model for traffic flow. Via Chapman-Enskog expansion we have shown that the BGK model leads to an advection-diffusion equation with a negative diffusion coefficient for a density greater than the critical one. This leads to unbounded growth of waves.

The issue has been addressed by deriving a modified BGK-type equation as mesoscopic description of the classical microscopic follow-the-leader model. The new model \revision{has a second-order moment system of ARZ--type} and, as direct consequence, it result to be conditionally well-posed. The instabilities occurring in the regime of high densities have now a bounded growth and can be regarded as model for stop and go waves. Further, with the new kinetic model, we have also introduced the mesoscopic description between the microscopic follow-the-leader model and the macroscopic ARZ model. 

\section*{Acknowledgments}
This work has been supported by HE5386/13-15.
\smallskip
\\M.~Herty and G.~Visconti would like to thank the German Research Foundation (DFG) for the kind support within the Cluster of Excellence ``Internet of Production'' (IoP ID390621612).
\smallskip
\\G.~Puppo and G.~Visconti are members of the ``National Group for Scientific Computation (GNCS-INDAM)''.

\bibliographystyle{plain}
\bibliography{AppuntiTraBib}
\end{document}